\documentclass[acmsmall, nonacm, table]{acmart}

\pdfoutput=1

\usepackage{thmtools}
\usepackage{setspace}
\usepackage{caption,subcaption}
\usepackage{amsmath,amsfonts,amsthm}
\usepackage{cleveref}
\usepackage{xspace}
\usepackage{algorithm}
\usepackage{xcolor}
\usepackage{soul}
\definecolor{gray}{rgb}{0.8,0.8,0.8}
\definecolor{lightgray}{rgb}{0.9,0.9,0.9}
\sethlcolor{gray}
\usepackage{algpseudocodex}
\usepackage{tcolorbox}
\tcbset{mystyle/.style={colback=gray, sharp corners, boxrule=0mm, 
left=0mm, right=0mm, bottom=0mm, top=0mm, %boxsep=0mm,
before skip=2mm,% Default value of 'top'
 after skip=0.5\baselineskip,
},}
\usepackage{longtable,array,ragged2e,booktabs}
\usepackage{enumitem}
\usepackage{tablefootnote}

\algrenewcommand\algorithmicrepeat{\textbf{do}}

\usetikzlibrary{shapes.geometric}

\newcommand{\set}[1]{\left\{ #1\right\}}

\newcommand{\R}{\mathbb{R}}
\newcommand{\N}{\mathbb{N}}
\newcommand{\Rpositive}{\R_{\geq 0}}

\newcommand{\localskew}{\mathcal{L}}

\newcommand{\FC}{\operatorname{\textbf{FC}}}
\newcommand{\FCone}{\operatorname{\text{\textbf{FC-1}}}}
\newcommand{\FCtwo}{\operatorname{\text{\textbf{FC-2}}}}
\newcommand{\FCthree}{\operatorname{\text{\textbf{FC-3}}}} 
\newcommand{\SC}{\operatorname{\text{\textbf{SC}}}}
\newcommand{\SCone}{\operatorname{\text{\textbf{SC-1}}}}
\newcommand{\SCtwo}{\operatorname{\text{\textbf{SC-2}}}}
\newcommand{\SCthree}{\operatorname{\text{\textbf{SC-3}}}}
\newcommand{\JC}{\operatorname{\text{\textbf{JC}}}}
\newcommand{\JCone}{\operatorname{\text{\textbf{JC-1}}}}
\newcommand{\JCtwo}{\operatorname{\text{\textbf{JC-2}}}}
\newcommand{\JCthree}{\operatorname{\text{\textbf{JC-3}}}}
\newcommand{\Pone}{{\textnormal{(P1)}}\xspace}
\newcommand{\Ptwo}{{\textnormal{(P2)}}\xspace}
\newcommand{\Pthree}{{\textnormal{(P3)}}\xspace}
\newcommand{\Pfour}{{\textnormal{(P4)}}\xspace}
\newcommand{\Pfive}{{\textnormal{(P5)}}\xspace}
\newcommand{\prefix}{\operatorname{prefix}}
\newcommand{\suffix}{\operatorname{suffix}}

\newcommand{\Cor}{\mathcal{C}}
\newcommand{\median}[1]{\operatorname{median}\set{#1}}
\newcommand{\own}{\mathrm{own}}

\newtheorem{theorem}{Theorem}[section]
\newtheorem{lemma}[theorem]{Lemma}
\newtheorem{observation}[theorem]{Observation}
\newtheorem{corollary}[theorem]{Corollary}

\theoremstyle{definition}
\newtheorem{definition}[theorem]{Definition}

\title{Clock Synchronization with Gradient TRIX}
\author{Christoph Lenzen}
\email{lenzen@cispa.de}
\affiliation{%
 \institution{CISPA Helmholtz Center for Information Security}
 \city{Saarbr{\"u}cken}
 \country{Germany}}

\author{Shreyas Srinivas}
\email{shreyas.srinivas@cispa.de}
\affiliation{%
\institution{CISPA Helmholtz Center for Information Security}
\city{Saarbr{\"u}cken}
\country{Germany}
}
\date{}

\begin{document}
\setcounter{page}{0}
\thispagestyle{empty}

\begin{abstract}
Gradient clock synchronization (GCS) algorithms minimize the worst-case clock offset between the nodes in a distributed network of diameter $D$ and size $n$.
They achieve optimal offsets of $\Theta(\log D)$ locally, i.e., between adjacent nodes~\cite{lenzen2010tight} and $\Theta(D)$ globally~\cite{biaz01closed}.
A key open problem in this area is to achieve fault tolerance at minimal edge replication overhead.

In this work, we achieve this goal under the assumption of an average-case distribution of faults, i.e., nodes fail with independent probability $p\in o(n^{-1/2})$.
In more detail, we present a self-stabilizing GCS algorithm for a grid-like directed graph with in- and out-degrees of $3$.
Note that even for tolerating a single fault, this degree is necessary, and if $p$ was larger, it would not hold with probability $1-o(1)$ that each node has at most one faulty in-neighbor.
Our algorithm achieves asymptotically optimal local skew of $\Theta(\log D)$ with probability $1-o(1)$; this holds under general worst-case assumptions on link delay and clock speed variations, provided they change slowly relative to the speed of the system. 

On the one hand, our results are of practical interest. As we discuss, the fault model is suitable for synchronously clocked hardware. Since our algorithm can simultaneously sustain a constant number of arbitrary changes due to faults in each clock cycle, it achieves sufficient robustness to dramatically increase the size of synchronously clocked systems. 

On the other hand, our results are of a theoretical interest. We show that for a worst-case distribution of $f$ faulty nodes within our fault model's locality constraints, our algorithm achieves local skew $O(5^f\log D)$. With probabilistically distributed faults, this becomes $O(\log D)$. Moreover, our work opens up avenues for further investigation of fault-tolerant synchronization, in particular trade-offs between fault distribution and edge density. 
\end{abstract}

% \begin{CCSXML}
%     <ccs2012>
%        <concept>
%            <concept_id>10010583</concept_id>
%            <concept_desc>Hardware</concept_desc>
%            <concept_significance>500</concept_significance>
%            </concept>
%        <concept>
%            <concept_id>10010583.10010633</concept_id>
%            <concept_desc>Hardware~Very large scale integration design</concept_desc>
%            <concept_significance>500</concept_significance>
%            </concept>
%        <concept>
%            <concept_id>10010147.10010919.10010172</concept_id>
%            <concept_desc>Computing methodologies~Distributed algorithms</concept_desc>
%            <concept_significance>500</concept_significance>
%            </concept>
%      </ccs2012>
% \end{CCSXML}
    
% \ccsdesc[500]{Fault Tolerance}
% \ccsdesc[500]{Clock Synchronisation}
% \ccsdesc[500]{Distributed algorithms}
% 
% \ccsdesc[500]{Hardware}
% \ccsdesc[500]{Hardware~Very large scale integration design}

% \keywords{Clock Synchronisation, Fault Tolerance, VLSI, Self-Stabilisation}

\maketitle
\begin{acks}
    This work was supported in part by the European Research Council (ERC) through the European Union’s Horizon 2020 Research and Innovation Programme under Grant 716562. Further, Shreyas Srinivas is a member of the Saarbr\"ucken Graduate School of Computer Science.
\end{acks}
\section{Introduction}

In their seminal work from 2004~\cite{fan04gcs}, Fan and Lynch introduced the task of Gradient Clock Synchronization (GCS). In a distributed network with imperfect reference clocks at each process and communication channels with uncertain message delays, it requires each process to construct logical clocks with the minimum possible clock \emph{skew}, i.e.\ absolute worst-case difference between clock outputs. The crucial distinction between GCS and ``classic'' clock synchronization frameworks is to not only consider the \emph{global skew}, i.e., the maximum skew between any pair of nodes, but also the \emph{local skew}, the maximum skew between neighbors.

The main insights motivating minimization of \emph{local skew} are:
\begin{itemize}
  \item In many cases, the skew between adjacent nodes is the appropriate measure of quality.
  \item The global skew grows at least linearly with the diameter $D$ of the network~\cite{biaz01closed}.
\end{itemize}

Lenzen, Locher, and Wattehofer.~\cite{lenzen2010tight} achieve optimal local skew with their GCS
algorithm and Bund et al.~\cite{bund2019fault} make the algorithm resilient to
localized Byzantine faults through heavy node and edge replication. However, practical applications demand fault tolerance with minimal network connectivity, which remains an open question.
\begin{tcolorbox}[title =\textsc{Fault-Tolerant Clock Synchronization Problem (Informal)}]
  Compute at each node of a distributed system a logical clock with the following properties.
  \begin{itemize}
    \item \textbf{Minimizing Global Skew}: The skew between any pair of nodes, i.e., the \emph{global skew} is minimized as a function of the network diameter $D$: $\Theta(D)$. 
    \item \textbf{Minimizing Local Skew}: The skew between adjacent pairs of nodes, i.e., the \emph{local skew} is minimized as a function of the network diameter $D$: $\Theta(\log D)$.
    \item \textbf{Fault Tolerance}: An unknown set of permanently faulty nodes, upto $f$ per neighbourhood can be tolerated, i.e., each node has at most $f$ faulty neighbors. We call this $f$-local fault tolerance.
    \item \textbf{Self-Stabilization}: After system-wide transient (i.e., temporary) faults, the processes re-converge to optimal skews.
    \item \textbf{Optimal Edge Density}: Achieve the above in a network topology with minimal node degree.
  \end{itemize}
\end{tcolorbox}

This work provides a positive answer for this question for grid-like graphs and a reasonable distribution of faults, motivated by the application of clocking VLSI systems. It synthesizes two distinct lines of work, namely \emph{gradient clock synchronization}~\cite{lenzen2010tight} and \emph{fault tolerant clock distribution}~\cite{Dolev2016a,lenzen20trix} which respectively achieve optimal clock skews and low-overhead fault tolerance. Doing so requires reconciling seemingly mutually exclusive approaches of the respective works, which constitutes one of the main technical challenges of this paper. Additionally, our work is an exercise in theory building. We make modelling choices that realistically account for the kind of networks and faults that occur in VLSI circuits, accepting some limitations on the generality of the fault model for strong performance guarantees. We carefully discuss these choices, their consequences, and related open questions.

\paragraph*{Replication for Fault Tolerance.}
%###
Bund et al.~\cite{bund2019fault} obtain a fault-tolerant variant of the GCS algorithm from~\cite{lenzen2010tight} by replicating nodes and edges.
They do so by simulating the (non-fault-tolerant) algorithm on the original network, replacing each node by a clique of size $3f+1$ and each edge by a biclique.
The clique then synchronizes internally using the classic fault-tolerant Lynch-Welch algorithm~\cite{welch88}, and, with some acrobatics, the resulting local outputs can be interpreted as a joint cluster clock executing the algorithm from \cite{lenzen2010tight} in lieu of the corresponding node of the original network. Thus, given an arbitrary network as input, one achieves gradient clock synchronization in the corresponding replicated network with up to $f$ faults per clique. 

This approach ticks many of the above boxes. Skews are asymptotically optimal, fault-tolerance is as desired, and using self-stabilization properties of the GCS algorithm and standard techniques, it is highly plausible that self-stabilization could be achieved. However, the edge replication factor of $\Theta(f^2)$ is more than required, as a node connectivity of $2f+1$ is sufficient to maintain (possibly sub-optimal) synchronization. From a practical viewpoint, solving even the case of $f=1$ would be big step forward, but even here, the constants of the solution are poor: for $f=1$, node degrees will be at least $15$.

Practical applications are typically highly sensitive to the network resources required. For example, floorplanning for VLSI systems is concerned with minimising wire length and congestion \cite[Section 1.10.1]{harris2010cmos}\ \cite{vlsiGraphLayout}.
Concretely, this is a crucial concern when distributing a system clock in hardware, where the factor $5$ gap to the minimum node degree of $2f+1=3$ is prohibitive:
this ``moderate constant'' causes a headache to the engineer trying to route all of these edges with few layers and precise timing, substantially increasing communication delay uncertainty.
This, in turn, directly translates into increased skews, placing the break-even point with prior art beyond relevant limits.

%###
\paragraph*{Clock Distribution.}
%###
These considerations motivate the quest for getting as close to the minimum required connectivity as possible.
This line of investigation led to the study of \emph{fault-tolerant clock distribution} in low-degree networks~\cite{Dolev2016a,lenzen20trix}.
Both of these works have the following in common: they assume that the clock signal is generated at a central location and forwarded along by nodes in a grid-like graph.
Their simple pulse forwarding schemes are self-stabilizing by design and resilient to isolated faults.
The basic idea is to propagate the signal from layer to layer, having each node wait for two nodes signaling a clock pulse before locally generating and forwarding their own pulse.
Moreover, it is assumed that in absence of faults, delays are changing only slowly over time.
Thus, matching the input frequency to the expected delay between grid layers results in clock pulses that are well-synchronized between adjacent layers.

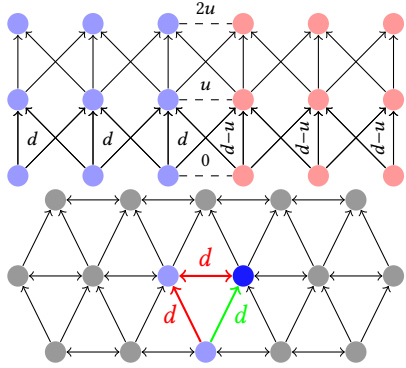
\begin{figure}[t]
  \centering
  \begin{subfigure}{0.4\textwidth}
    \begin{tikzpicture}[main/.style = {draw, circle},scale = 2.0] 
    \tikzstyle{vertex} = [
            circle,
            fill = {
            black! 25
            },
            minimum size=8pt, 
            inner sep=2pt]  

    % source nodes
    \node[vertex, fill = blue!40] (s1) at (-1.0, -0.5) {}; 
    \node[vertex, fill = blue!40] (s2) at (-0.5,-0.5) {};
    \node[vertex, fill = blue!40] (s3) at (0, -0.5) {};
    \node[vertex, fill = red!40] (s4) at (0.5,-0.5) {};
    \node[vertex, fill = red!40] (s5) at (1.0,-0.5) {};
    \node[vertex, fill = red!40] (s6) at (1.5,-0.5) {};

    % layers
    \node[vertex, fill = blue!40] (l11) at (-1.0, 0) {}; 
    \node[vertex, fill = blue!40] (l12) at (-0.5,0) {};
    \node[vertex, fill = blue!40] (l13) at (0, 0) {};
    \node[vertex, fill = red!40] (l14) at (0.5,0) {};
    \node[vertex, fill = red!40] (l15) at (1.0,0) {};
    \node[vertex, fill = red!40] (l16) at (1.5,0) {};

    % layers
    \node[vertex, fill = blue!40] (l21) at (-1.0, 0.5) {}; 
    \node[vertex, fill = blue!40] (l22) at (-0.5,0.5) {};
    \node[vertex, fill = blue!40] (l23) at (0, 0.5) {};
    \node[vertex, fill = red!40] (l24) at (0.5,0.5) {};
    \node[vertex, fill = red!40] (l25) at (1.0,0.5) {};
    \node[vertex, , fill = red!40] (l26) at (1.5,0.5) {};

    %Skews
    \draw[dashed] (s3) -- (s4) node[above, midway] {${\scriptstyle 0}$};
    \draw[dashed] (l13) -- (l14) node[above, midway] {${\scriptstyle u}$};
    \draw[dashed] (l23) -- (l24) node[above, midway] {${\scriptstyle 2u}$};

    %Edges

    \draw[->] (s1) -- (l11) node[midway, right] {${\scriptstyle d}$};
    \draw[->] (s1) -- (l12);
    \draw[->] (s2) -- (l11);
    \draw[->] (s2) -- (l12) node[midway, right] {${\scriptstyle d}$};
    \draw[->] (s2) -- (l13);
    \draw[->] (s3) -- (l12);
    \draw[->] (s3) -- (l13) node[midway, right] {${\scriptstyle d}$};
    \draw[->] (s3) -- (l14);
    \draw[->] (s4) -- (l13);
    \draw[->] (s4) -- (l14) node[midway, above, sloped] {${\scriptstyle d-u}$};
    \draw[->] (s4) -- (l15);
    \draw[->] (s5) -- (l14);
    \draw[->] (s5) -- (l15) node[midway, above, sloped] {${\scriptstyle d-u}$};
    \draw[->] (s5) -- (l16); 
    \draw[->] (s6) -- (l15);
    \draw[->] (s6) -- (l16) node[midway, above, sloped] {${\scriptstyle d-u}$};

    \draw[->] (s1) -- (l11);
    \draw[->] (s1) -- (l12);
    \draw[->] (s2) -- (l11);
    \draw[->] (s2) -- (l12);
    \draw[->] (s2) -- (l13);
    \draw[->] (s3) -- (l12);
    \draw[->] (s3) -- (l13);
    \draw[->] (s3) -- (l14);
    \draw[->] (s4) -- (l13);
    \draw[->] (s4) -- (l14);
    \draw[->] (s4) -- (l15);
    \draw[->] (s5) -- (l14);
    \draw[->] (s5) -- (l15);
    \draw[->] (s5) -- (l16);
    \draw[->] (s6) -- (l15);
    \draw[->] (s6) -- (l16);

    \draw[->] (l11) -- (l21);
    \draw[->] (l11) -- (l22);
    \draw[->] (l12) -- (l21);
    \draw[->] (l12) -- (l22);
    \draw[->] (l12) -- (l23);
    \draw[->] (l13) -- (l22);
    \draw[->] (l13) -- (l23);
    \draw[->] (l13) -- (l24);
    \draw[->] (l14) -- (l23);
    \draw[->] (l14) -- (l24);
    \draw[->] (l14) -- (l25);
    \draw[->] (l15) -- (l24);
    \draw[->] (l15) -- (l25);
    \draw[->] (l15) -- (l26);
    \draw[->] (l16) -- (l25);
    \draw[->] (l16) -- (l26);
\end{tikzpicture}  
  \end{subfigure}
  \begin{subfigure}{0.4\textwidth}
    \begin{tikzpicture}[main/.style = {draw, circle},scale = 2.0] 
    \tikzstyle{vertex} = [
            circle,
            fill = {
            black! 25
            },
            minimum size=8pt, 
            inner sep=2pt]  

    % source nodes
 
    \node[vertex, fill = black!40] (s12) at (-0.5,-0.5) {};
    \node[vertex, fill = black!40] (s13) at (0, -0.5) {};
    \node[vertex, fill = blue!40] (s14) at (0.5,-0.5) {};
    \node[vertex, fill = black!40] (s15) at (1.0,-0.5) {};
    \node[vertex, fill = black!40] (s16) at (1.5,-0.5) {};

    \node[vertex, fill = black!40] (s21) at (-0.75, 0) {}; 
    \node[vertex, fill = black!40] (s22) at (-0.25, 0) {};
    \node[vertex, fill = blue!40] (s23) at (0.25, 0) {};
    \node[vertex, fill = blue!90] (s24) at (0.75, 0) {};
    \node[vertex, fill = black!40] (s25) at (1.25, 0) {};
    \node[vertex, fill = black!40] (s26) at (1.75, 0) {};
 
    \node[vertex, fill = black!40] (s32) at (-0.5,0.5) {};
    \node[vertex, fill = black!40] (s33) at (0, 0.5) {};
    \node[vertex, fill = black!40] (s34) at (0.5,0.5) {};
    \node[vertex, fill = black!40] (s35) at (1.0,0.5) {};
    \node[vertex, fill = black!40] (s36) at (1.5,0.5) {};

    \draw[<->] (s12)--(s13);
    \draw[<->] (s13)--(s14);
    \draw[<->] (s14)--(s15);
    \draw[<->] (s15)--(s16);
    \draw[<->] (s21)--(s22);
    \draw[<->] (s22)--(s23);
    \draw[<->, thick, color=red] (s23)--(s24) node[midway,above] {$d$};
    \draw[<->] (s24)--(s25);
    \draw[<->] (s25)--(s26);
    \draw[<->] (s32)--(s33);
    \draw[<->] (s33)--(s34);
    \draw[<->] (s34)--(s35);
    \draw[<->] (s35)--(s36);

    \draw[->] (s12)--(s21);
    \draw[->] (s12)--(s22);
    \draw[->] (s13)--(s22);
    \draw[->] (s13)--(s23);
    \draw[->, color=red, thick] (s14)--(s23) node[midway, left] {$d$};
    \draw[->, color=green, thick] (s14)--(s24) node[midway, right] {$d$};
    \draw[->] (s15)--(s24);
    \draw[->] (s15)--(s25);
    \draw[->] (s16)--(s25);
    \draw[->] (s16)--(s26);

    \draw[->] (s21)--(s32);
    \draw[->] (s22)--(s32);
    \draw[->] (s22)--(s33);
    \draw[->] (s23)--(s33);
    \draw[->] (s23)--(s34);
    \draw[->] (s24)--(s34);
    \draw[->] (s24)--(s35);
    \draw[->] (s25)--(s35);
    \draw[->] (s25)--(s36);
    \draw[->] (s26)--(s36);
\end{tikzpicture}  
  \end{subfigure}
  \hfill
  \caption{TRIX~\cite{lenzen20trix} (left) and HEX~\cite{Dolev2016a} (right) grids. TRIX uses the naive pulse forwarding scheme of waiting for the second copy of each pulse before forwarding it. We see how the TRIX grid can accumulate a local skew of $\Theta(uD)$ in layer $D$. In the HEX grid, each node waits for two copies of a pulse from in-neighbors. However, $2$ of the $4$ in-neighbors are on the same layer, causing a skew of $d$ if a neighbor on the preceding layer crashes.\label{fig:prev_grid_structures}}
\Description[Previous grid structures: TRIX and HEX]{TRIX (left) and HEX (right) grids. TRIX uses the naive pulse forwarding scheme of waiting for the second copy of each pulse before forwarding it. We see how the TRIX grid can accumulate a skew of $\Theta(uD)$. In the HEX grid, each node waits for two copies of a pulse from in-neighbours. However, $2$ of the $4$ in-neighbors are on the same layer, causing a skew of $d$ if a neighbor on the preceding layer crashes.}
\end{figure}

The above works differ in the used grid structure, cf.~\Cref{fig:prev_grid_structures}, and achieved skew bounds:
\begin{itemize}
  \item Denoting by $d-u$ and $d$ the minimum and maximum end-to-end communication delay, in a grid of width $D$, ~\cite{Dolev2016a} bounds the local skew by $d+O(u^2 D/d)$. Since in practice $d\gg u$, this is a non-trivial bound. Unfortunately, the fact that $d\gg u$ also means that the additive term of $d$ renders this bound far too large for applications. Even worse, for each fault this bound increases by another $d$.
  \item In~\cite{lenzen20trix}, each fault adds at most $u$ to the local skew. Observe that the used grid also has the minimum required connectivity, as each node has only $3$ incoming and outgoing edges each. Alas, these advantages come at the expense of poor scaling of worst-case skews with the number of layers: on layer $D$, adjacent nodes may pulse up to $uD$ time apart.
\end{itemize}

% We remark that in order to tolerate failure of an arbitrary component of the system, the clock source has to be replicated and the replicas have to be synchronized in a fault-tolerant and self-stabilizing manner.
% However, here one can employ techniques for fully connected networks~\cite{khanchandani18self,lenzen19almost};
% using them in a single location for $f=1$ does not constitute a scalability issue.

% Some colours for table cells : \cellcolor[rgb]{0.8,0.8,0.8}

%###
\paragraph*{Our Contribution.}
%###
In a nutshell, we provide a solution to the fault-tolerant clock synchronization problem under the assumption of a uniform distribution of faulty nodes, with a constant number of them changing their behavior concurrently, cf.~\Cref{cor:bound}.
While we acknowledge the limitations in the fault model, nonetheless our results significantly advance the state of the art, cf.~\Cref{table:summary_desiderata}.
Moreover, as we discuss in detail in~\Cref{sec:model}, the fault model is well-justified in the application setting of clocking synchronous hardware.
\begin{table}[t!]
  \begin{center}
    \begin{tabular}{|p{6.2em}|p{4.7em}|p{5.7em}|p{5.2em}|p{3.7em}|p{7.8em}|}
      \hline
      method & global skew & local skew & resilience & self-stab. & graph topology\\
      \hline
      LW~\cite{welch88} & $O(1)$ & $O(1)$ & $<n/3$ & no & complete ($D=1$)\\
      \hline
      KL~\cite{khanchandani18self} & $O(1)$ & $O(1)$ & $<n/3$ & yes & complete ($D=1$)\\
      \hline
      HEX~\cite{Dolev2016a}& $O(dD)$ & $d + O(u^2D/d)$ & $1$-local & yes & grid-like, \mbox{suboptimal~degree}\\
      \hline
      TRIX~\cite{lenzen20trix}& $O(uD^2)$ & $O(uD)$ & $1$-local & yes & grid-like, \mbox{optimal~degree}\\
      \hline
      GCS~\cite{lenzen2010tight} & $O(uD)$ & $O(u\log D)$ & crashes only & yes & arbitrary\\
      \hline
      Fault-tolerant GCS~\cite{bund2019fault} & $O(uD)$ & $O(u\log D)$ & $f$-local & yes & $\Theta(f^2)$-augmented arbitrary graph\tablefootnote{Given a graph topology $G$, the augmented graph contains a $3f+1$-clique of replica vertices for each node $v$ in $G$ and $\Theta(f^2)$ copies of each edge $\set{v,w}\in G$ corresponding to all the possible pairs of the replicas of $v$ and $w$}\\
      \hline
      Gradient~TRIX \textbf{(this work)} & $O(uD)$ & $O(u\log D)$ & independent $p\in o(n^{-1/2})$ & yes & grid-like, \mbox{optimal~degree}\\
      \hline
      Gradient~TRIX \textbf{(this work)} & $O(uD)$ & $O(5^f u\log D)$ & $1$-local,~\phantom{XXX} $f=$ \#faults & yes & grid-like, \mbox{optimal~degree}\\
      \hline
    \end{tabular}
  \end{center}
  \caption{Comparison to related work. With the exception of GCS, ``resilience'' refers to Byzantine fault-tolerance, i.e., worst-case behavior of faulty nodes. However, in our work the fault model is restricted in that only few faulty nodes change their behavior within a short amount of time. In turn, we are the first to simultaneously achieve optimal skew bounds, self-stabilization, and minimal degrees.\label{table:summary_desiderata}
\vspace*{-.5cm}}
\end{table}

Our results are achieved using the same grid as in~\cite{lenzen20trix}, but with a different rule for forwarding pulses.
Our novel algorithm is designed as a discrete and fault-tolerant counterpart to the GCS algorithm from~\cite{lenzen2010tight}.
Making this work requires substantial conceptual innovation and technical novelty.
On the conceptual level, like~\cite{bund2019fault}, our algorithm simulates a discretized variant of the (non-fault-tolerant) GCS algorithm from~\cite{lenzen2010tight}.\footnote{Traditional presentations of GCS describe an algorithmic template, i.e., a set of constraints that must be obeyed by any GCS algorithm at all time instants $t \in \Rpositive$~\cite{lecture}. Usually, a discretization of this template yields an algorithm for each node that acts at discrete time points to estimate the clocks of its neighbours and adjusts its clock speed. Our approach ``splits a GCS node in time'' across an entire grid column, inherently discretizing when (simulated) nodes act. However, we do not formally execute a simulation argument, as this would necessitate to first generalize existing (continuous) GCS schemes. Nonetheless, we refer to the underlying simulation idea for the purpose of exposition.}
The key difference is that we do not physically replicate nodes or edges, avoiding the corresponding overheads.
Instead, we ``split'' the simulation of a node across time, associating only a short time slice with each simulating node, corresponding to the local forwarding of the pulse.

Abstractly, this can be phrased as follows.
Taking an arbitrary \emph{base graph} of minimum degree $2$,\footnote{From a theoretical point of view, the cleanest choice for the base graph might be a cycle. For the considered VLSI circuit application, for physical layout a line with replicated and connected endpoints is better.} we create copies of the graph, referred to as layers. 
Each layer represents a ``time step'' of the GCS algorithm on the base graph.
For each node, there is an edge from its copy on a given layer to the copies of itself and its neighbors on the next.
The pulses forwarded along these edges serve very different functions:
\begin{itemize}
  \item The pulse messages sent to copies of neigbhors correspond to the GCS algorithm's messages for estimating clock offsets to neighbors.
  \item The pulse messages sent between copies of the same node convey its local time from one of its copies to the next.
\end{itemize}
Note that this turns a permanently faulty node in the grid into a simulated node being faulty in a single time step only.
This is of vital importance, because it enables us to rely on the self-stabilization properties of the GCS algorithm from~\cite{lenzen2010tight}.
These are implicitly shown in~\cite{kuhn10dynamic_arxiv};
we prove them explicitly in the different setting of this work.

However, by itself this does not guarantee bounded skew between correct nodes, since we also need to contain the effect of such a ``transient'' fault on the state of the simulated algorithm.
Otherwise, a fault would increase skews arbitrarily, effectively corrupting downstream nodes:
at any given node, the smallest or largest time at which a pulse from neighbors on the preceding layer is received could be determined by a faulty node.
We can overcome this issue if there is at most one faulty in-neighbor.
The key observation to controlling the impact of a faulty node on the pulse time lies in that it can indeed affect only one of three reception times:
the (i) smallest \emph{or} (ii) largest time at which a pulse from copies of neighbors on the previous layer is received, \emph{or} (iii) the time at which the pulse from the copy of the node itself is received.
In particular, the median of these three times lies within the interval spanned by the correct in-neighbors' pulse times.
By imposing a constraint to always tie the time at which a pulse is generated closely to this median, in addition to the standard ``slow'' and ``fast'' conditions of the GCS algorithm, we can limit the local impact of a fault on skews.

In summary, we seek to simultaneously simulate a time-discrete variant of the GCS algorithm from~\cite{lenzen2010tight}, while also guaranteeing that pulse forwarding times are, up to a sufficiently small deviation, identical to median reception times plus a fixed offset.
Unfortunately, no existing GCS algorithm that achieves a small local skew~\cite{kuhn10dynamic,kuhn09reference,lenzen08clock,lenzen2010tight} can be used for this purpose as-is, since their decision rules are in conflict with the above ``stick to the median'' requirement.

As our main technical contribution, we resolve this conflict, simultaneously adapting the resulting algorithm to the discrete setting.
To do so, we determine suitably weakened discrete variants of the slow and fast conditions introduced in~\cite{kuhn09reference}.
In essence, we allow that a simulated node whose pulse time is ahead of all of its neighbors' pulse times to delay its next pulse by the difference to the fastest neighbor;
an analogous rule applies to nodes pulsing later than all of their neighbors.
From the perspective of the GCS algorithm in~\cite{lenzen2010tight}, that we build on, this constitutes a potentially arbitrarily large clock ``jump,'' which we leverage to implement the stick-to-the-median requirement despite the arbitrary changes in timing that faulty nodes may apply to their pulse messages.
To prevent uncontrolled oscillatory behavior arising from adjacent nodes ``jumping'' in opposite directions, we introduce an additional condition, which we refer to as the \emph{jump condition}.
Essentially, it slightly reduces how large jumps are to avoid that uncertainty in message delays and local clock speeds cause nodes to ``overswing,'' potentially resulting in arbitrarily large skews, cf.~\Cref{fig:need_for_JC}.

Turning so many knobs at once meant that it was not clear that such a scheme would work.
Indeed, bounding the skew of this novel algorithm turned out to be highly challenging, as jumps that delay pulses rather than speeding them up invalidate the fundamental assumption that clocks progress at rate at least $1$ present in all prior work~\cite{kuhn10dynamic,kuhn09reference,lenzen08clock,lenzen2010tight}.
As a result, the main technical hurdle and contribution turned out to be proving a bound on the local skew $\localskew_{\ell}$ between neighbors in the same layer $\ell$ for the fault-free case.
\begin{restatable}{theorem}{thmlocal}\label{thm:local}
If there are no faults, then $\localskew_{\ell}\le 4\kappa (2+\log D)$ for all $\ell\in \N$.
\end{restatable}
Here $D$ is the diameter of the base graph, and $\kappa$ is a parameter of the algorithm;
choosing the input clock frequency to be $1/(2d)$ results in $\kappa\in \Theta(u+(\vartheta-1)d)$ for nodes with local reference clocks running at rates between $1$ and $\vartheta>1$.
All of our results require that $d\gg u+(\vartheta-1)d$, or equivalently, that the local skew remains small compared to $d$.
Note that if this condition does not hold, we are outside the parameter range of interest: then skews become large compared to the desired duration of a clock cycle and clock frequency has to be reduced substantially.

To address faults, we bound how they can affect timing.
Due to the aforementioned stick to the median rule, we can bound the local impact of a fault on timing in terms of the local skew.
However, applying this argument repeatedly, skews grow exponentially in the number of faults.\!\!\!\!
\begin{restatable}{theorem}{thmexp}\label{thm:fault_worst_case}
If there are at most $f$ faulty nodes and none in layer~$0$, then $\localskew_{\ell}\in O(5^f\kappa\log D)$.
\end{restatable}
While tolerating a constant number of faults is certainly better than tolerating none, this is unsatisfactory, since the requirement of one faulty in-neighbor holds with probability $1-o(1)$ for a fairly high independent probability of $p\in o(1/\sqrt{n})$.
Given that the topology we are most interested in is roughly a square grid, i.e., there are roughly $\sqrt{n}$ layers, the naive approach outlined above does not result in a non-trivial bound on the local skew for any $p\in \omega(1/n)$.

To address this issue, we provide an improved analysis exploiting that our base graph has slow-growing neighborhoods as a function of distance.
As the $d$-hop neighborhood grows linearly with $d$, the number of nodes in layers $\ell'\in [\ell-n^{1/12},\ell]$ that affect the pulse time of a node in layer $\ell$ is in $\Theta(n^{1/6})$.
Thus, if nodes fail with probability $p\in o(1/\sqrt{n})$, the probability that there are more than $2$ faulty nodes within distance $n^{1/12}$ that affect a given node is $o(1/n)$.
Intuitively, this buys enough time for the self-stabilization properties of the simulated algorithm to reduce its local skew again before it spirals out of control.
\begin{restatable}{theorem}{thmlocalfaults}\label{thm:local_faults}
  With probability $1-o(1)$, $\localskew_{\ell}\in O(\kappa \log D)$ for all $\ell \in \N$.
\end{restatable}
The final step is to extend this bound on the local skew within a layer to one that includes adjacent nodes in different layers.
As we propagate pulses layer by layer, we cannot hope to match pulse times of the $k$-th pulse between different layers.
Instead, we match the input period to the nominal time a pulse spends on each layer.
This works neatly so long as there are no \emph{changes} in message delay, clock speed, and behavior of faulty nodes between consecutive pulses.\!\!\!\!
\begin{restatable}{theorem}{thmbound}\label{thm:bound}
If faulty nodes do not change the timing of their output pulses, then $\localskew\in O(\kappa \log D)$ with probability $1-o(1)$.
\end{restatable}
To a large extent, this strong assumption is justified in our specific context.
Clock speeds of modern systems are in the gigahertz range, and the amount of change in timing that occurs within a single clock cycle is much smaller than over the lifetime of a system~\cite{Tschanz2009},~\cite[Ch.~7]{harris2010cmos}.
Similarly, the by far most common faults are static faults and delay faults with a static timing profile.\footnote{That is, a given timing profile of input transitions results in the same relative timing of output transitions.}
From the point of view of the receiving node, this results in an early or late pulse, respectively, without any change in behavior between pulses.
Of course, timing will still change slowly, the above benign faults will occur at some point, before which the nodes worked correctly, and some faults may be more severe.
Using once more that faulty nodes' impact on timing is bounded by the local skew, the bound from \Cref{thm:bound} extends to a constant number of arbitrary faults in each pulse alongside small changes in delays and hardware clock speeds.
\begin{restatable}{corollary}{corbound}\label{cor:bound}
With probability $1-o(1)$, $\localskew \in O(\kappa \log D)$ even when in each pulse (i) a constant number of faulty nodes change their output behavior and timing, (ii) link delays vary by up to $n^{-1/2}u \log D$, and (iii) hardware clock speeds vary by up to $n^{-1/2}(\vartheta-1)\log D$.
\end{restatable}
Finally, if all else fails, we can fall back on the ability of the pulse progation algorithm to recover from arbitrary transient faults.
In constrast to the \emph{simulated} GCS algorithm, achieving self-stabilization of the pulse propagation scheme itself is straightforward due to the directionality of the propagation.
We stress that in contrast to the comparative ease at which this result is obtained, it is a very powerful and crucial fault tolerance property; this ease is the result of the key design decision to propagate pulses in a directed fashion.

\begin{restatable}{theorem}{thmself}\label{thm:self-stab}
The pulse propagation algorithm can be implemented in a self-stabilizing way.
It stabilizes within $O(\sqrt{n})$ pulses.
\end{restatable}
In light of these results, we view this work as a major step towards simultaneously achieving high performance and strong robustness in the practical setting of clock distribution in hardware.
In alignment with the theoretical question motivating this work, we achieve an asymptotically optimal local skew at the minimum possible node degree under the assumption of node failures with probability $o(n^{-1/2})$.

%###
\paragraph*{The Bigger Picture.}
%###
The above discussion left open how our work advances the state of the art in a broader context.
There are several points we would like to highlight in this context.
\begin{enumerate}
  \item The fault model is more limited than proven necessary by known impossibilities. In the specific application context, this is well-motivated: there is no attacker inducing a worst-case distribution of faults or coordinating the actions of faulty nodes. However, from a theoretical point of view, a minimal set of assumptions is desirable. We view our work as a first and important step in this direction. As discussed above, reducing the skews introduced in the wake of faulty nodes required substantial changes to the algorithm and far-reaching adjustments to the proof strategy. A logical next step would be to explore whether fault-tolerant gradient clock synchronization is feasible with an undirected degree-$3$ network. Achieving this is likely to be necessary to overcome the limitation that only few changes in timing between consecutive pulses can be handled by our solution. We exploit the directional propagation to make self-stabilization simple and re-interpret permanent faults as transient ones from the perspective of simulating the non-fault-tolerant GCS algorithm, but this renders the above downside inherent to our scheme.
  \item We consider an average-case distribution of faults. For low-degree networks, this is a must: if an adversary can choose the position of faults, very few faults can overcome the connectivity of the network. By justifying this restriction for an important practical application, we demonstrate the utility of this modeling choice. We hope that this will inspire follow-up work on fault-tolerant clock synchronization in average-case fault models.
  \item We establish that the case of $f=1$ can be handled with minimal node in-degrees of $3$. This may open up the way towards a general scheme achieving resilience to $f$ local faults with in-degree $2f+1$.
  \item Our main technical contribution might prove useful in different contexts. Prior work~\cite{kuhn10dynamic,kuhn09reference,lenzen08clock,lenzen2010tight} did not allow for adjusting clocks ``backwards,'' which is pivotal in proving the self-stabilization properties of the simulated GCS algorithm. We anticipate that our novel techniques will prove useful in different settings. For instance, one could seek to modify the algorithm from~\cite{lenzen2010tight} to recover from transient faults of bounded duration without increasing the local skew between non-faulty nodes by more than a constant factor.
\end{enumerate}

%###
\paragraph*{Organization of this Article.}
%###
In \Cref{sec:model}, we discuss the system model, introduce the graph on which we run our synchronization algorithm, and motivate our modeling choices, including its non-standard aspects.
We then present a simplified version of the algorithm that better highlights the conceptual approach in \Cref{sec:algorithm}.
We follow with the formal derivation of the skew bounds in \Cref{sec:analysis}.
\Cref{app:layer0} discusses how to generate synchronized pulses on layer~$0$.
The full algorithm and its equivalence to the simplified one if there are no faulty predecessors are shown in \Cref{app:algo}.
Making the algorithm self-stabilizing is discussed in \Cref{app:self-stab}.

\section{Modeling}\label{sec:model}

We use a non-standard model, which is motivated by the specific setting outlined in the introduction.
Accordingly, we will emphasize and discuss model choices where this seems prudent.

%###
\paragraph*{Setting.}
%###
Recall that our motivating application is to provide a synchronized clock signal to a large System-on-Chip.
Physically, this means that we need to provide the clock signal to a rectangular area; for simplicity, we will assume the most common case of it being square.
We want to supply a uniform grid of nodes in the square area with this signal, which then will serve as roots of relatively small local clock trees supplying the low-level components with the clock signal.
If these trees contribute a maximum clock skew of $\Delta$ and the skew between adjacent grid points is at most $\localskew$, the triangle inequality guarantees a worst-case skew of $\localskew + 2\Delta$ between adjacent components of the System-on-Chip.
The local clock trees can be designed using standard methodology.
Therefore, in the following we will focus exclusively on the grid of their roots.

A key assumption we make is that communication delay between correct adjacent nodes changes only slowly with time.
This enables us to generate synchronized pulses at all grid nodes by matching the input frequency with the (inverse) propagation time between consecutive layers.
This is justified for two reasons:
\begin{itemize}
  \item The dominant sources of uncertainty in propagation delay are inaccuracies in component fabrication, aging, and temperature and frequency variations that are slow relative to the time it takes to propagate an input clock pulse across even a large System-on-Chip~\cite{Tschanz2009}. For example temperature-induced delay variations occur on the order of milliseconds \cite[Section 7.2.2]{harris2010cmos}
  \item Changing delays of \emph{all} links between a pair of adjacent layers by up to $\delta$ increases skew bounds by at most $\delta$, cf.~\Cref{lem:drift_fault}. 
\end{itemize}

In order to generate sufficiently synchronized pulses at the nodes of layer~$0$, a straightforward solution is to use a simple path fed by the clock source, see \Cref{app:layer0}. This scheme suffices for our fault model,\footnote{Tolerating one local fault is also straightforward by using a redundant path; as this does not increase the resilience of the system asymptotically, we stick to the simplest scheme.} since the probability that there is any (permanent) fault in layer~$0$ is $o(1)$ due to the small total number of nodes in a single layer. As we show in \Cref{cor:layer0}, this approach is self-stabilizing and results in well-synchronized input pulses.

In a perfect grid, all layers would consist of a path.
Unfortunately, this results in the issue that the endpoints of the path, lacking one neighbor, would have only two adjacent nodes in the preceding and subsequent layer.
A naive solution is to insert additional edges between the boundary nodes, turning the layer into a cycle and the entire graph into a cylinder (with some special treatment of layer~$0$).
However, realizing such a solution on the square would result in far too long edges between boundary nodes or require to, essentially, replicate each layer, effectively doubling the number of nodes and edges in the graph.

Instead, we choose to replicate the boundary nodes only, which then provides the ``missing'' input to the next layer.
Note that this increases the degree of the nodes next to the boundary nodes by one.
We cope with this by a general analysis allowing for the layers to be copies of an arbitrary base graph of minimum degree~$2$.
%In \Cref{fig:base_graph} and \Cref{fig:base_grid}, we show the base graph and the resulting synchronization network in our assumed setting.
In \Cref{fig:base_graph,fig:layers_of_our_graph}, we show the base graph and the connectivity of nodes between adjacent layers of our synchronization network in our assumed setting, respectively.

%###
\paragraph*{Network Graph.}
%###
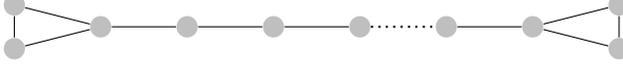
\begin{figure}[t]
  \centering
  \begin{tikzpicture}[main/.style = {draw, circle},scale = 2.3] 
    \tikzstyle{vertex} = [
            circle,
            fill = {
            black! 25
            },
            minimum size=8pt, 
            inner sep=2pt]  

    \def\rowpos{0.0, 1.0, 2.0, 3.0}
    \def\colpos{-1.0, -0.5, 0.0, 0.5, 1.0, 1.5}

    \newcommand{\rowdraw}[2]{
        \foreach \ci[count=\i from 0] in \colpos {
            \node[vertex] (s\i#2) at (\ci, #1) {};
        }
        \node[vertex](beg_s_0#2) at (-1.5, #1+0.125) {};
        \node[vertex](beg_s_1#2) at (-1.5, #1-0.125) {};
        \node[vertex](end_s_0#2) at (2.0, #1+0.125) {};
        \node[vertex](end_s_1#2) at (2.0, #1-0.125) {};
    }
    \rowdraw{0.0}{0}
    \draw (s00)--(s10);
    \draw (s10)--(s20);
    \draw (s20)--(s30);
    \draw[dotted, thick] (s30)--(s40);
    \draw (s40)--(s50);
    \draw (s00)--(beg_s_00);
    \draw (s00)--(beg_s_10);
    \draw (beg_s_00)--(beg_s_10);
    \draw (s50)--(end_s_00);
    \draw (s50)--(end_s_10);
    \draw (end_s_00)--(end_s_10);
    
\end{tikzpicture}
  \caption{Base graph $H$ used in this work. Rather than using a cycle, which would result in a TRIX grid, we replicate the end nodes of a line to ensure a minimum degree of $2$. Alternatively, one could use a line and exploit that the probability that one of the $O(\sqrt{n})$ boundary nodes fails is $o(1)$.\label{fig:base_graph}}
  %\Description[Base Graph $H$ of our grid]{}
\end{figure}
We are given a simple connected base graph $H=(V,E)$ of minimum degree $2$ and diameter $D\in \N_{>0}$.
For $v,w\in V$, denote by $d(v,w)\le D$ the distance from $v$ to $w$ in $H$.
To derive the graph $G=(V_G,E_G)$ we use for synchronization, for each $\ell\in \N$ we create a copy $V_{\ell}$ of $V$.
Denoting by $(v,\ell)$ the copy of $v\in V$ in $V_{\ell}$, we define $E_{\ell}:=\{((v,\ell),(w,\ell+1))\,|\,\{v,w\}\in E \vee v=w\}$.
We now obtain $G$ by setting $V_G:=\bigcup_{\ell \in \N}V_{\ell}$ and $E_G:=\bigcup_{\ell \in \N} E_{\ell}$.
That is, for each \emph{layer} $\ell\in \N$ we have a copy of $v\in V$, which has outgoing edges to the copies of itself and all its neighbors on layer $\ell+1$. Here, $\ell$ is bounded from above by some value in $\Theta(\sqrt{n})$.
We slightly abuse notation by neglecting this bound on $\ell$ in lemma statements and proofs.
Since $G$ is a DAG, we refer to out-neighbors as \emph{successors} and in-neighbors as \emph{predecessors}.

\begin{figure}[t]
  \centering
  \begin{tikzpicture}[main/.style = {draw, circle},scale = 2.3] 
    \tikzstyle{vertex} = [
            circle,
            fill = {
            black! 25
            },
            minimum size=8pt, 
            inner sep=2pt]  

    \def\rowpos{0.0, 1.0, 2.0, 3.0}
    \def\colpos{-1.0, -0.5, 0.0, 0.5, 1.0, 1.5}

    \newcommand{\rowdraw}[2]{
        \foreach \ci[count=\i from 0] in \colpos {
            \node[vertex] (s\i#2) at (\ci, #1) {};
        }
        \node[vertex](beg_s_0#2) at (-1.5, #1+0.125) {};
        \node[vertex](beg_s_1#2) at (-1.5, #1-0.125) {};
        \node[vertex](end_s_0#2) at (2.0, #1+0.125) {};
        \node[vertex](end_s_1#2) at (2.0, #1-0.125) {};
    }
    \rowdraw{0.25}{0}
    \rowdraw{0.75}{1}

    \draw[->] (beg_s_00) to[bend left] (beg_s_01);
    \draw[->] (beg_s_00)--(beg_s_11);
    \draw[->] (beg_s_00)--(s01);

    \draw[->] (beg_s_10) to[bend left] (beg_s_01);
    \draw[->] (beg_s_10) to[bend left] (beg_s_11);
    \draw[->] (beg_s_10) to (s01);
    \draw[->] (s00)--(beg_s_01);
    \draw[->] (s00)--(beg_s_11);
    \draw[->] (s00)--(s01);
    \draw[->] (s00)--(s11);

    \draw[->] (s10)--(s01);
    \draw[->] (s10)--(s11);
    \draw[->] (s10)--(s21);

    \draw[->] (s20)--(s11);
    \draw[->] (s20)--(s21);
    \draw[->] (s20)--(s31);

    \draw[->,dotted, thick] (s30)--(s41);
    \draw[->] (s30)--(s31);
    \draw[->] (s30)--(s21);

    \draw[->, dotted, thick] (s40)--(s31);
    \draw[->] (s40)--(s41);
    \draw[->] (s40)--(s51);
    
    \draw[->] (s50)--(s41);
    \draw[->] (s50)--(s51);
    \draw[->] (s50)--(end_s_11);
    \draw[->] (s50)--(end_s_01);

    \draw[->] (end_s_00) to[bend right] (end_s_01);
    \draw[->] (end_s_00)--(end_s_11);
    \draw[->] (end_s_00)--(s51);

    \draw[->] (end_s_10) to[bend right] (end_s_01);
    \draw[->] (end_s_10) to[bend right] (end_s_11);
    \draw[->] (end_s_10) to (s51);
    
\end{tikzpicture}
  \caption{\label{fig:layers_of_our_graph} Layer structure of $G$ resulting from our choice of $H$. Most nodes have in- and out-degree $3$, some $4$.}
  %\Description[Connectivity between adjacent layers of our grid]{}
\end{figure}
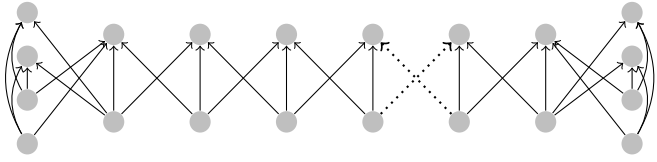

%###
\paragraph*{Fault Model.}
%###
An unknown subset $F\subset V_G$ is \emph{faulty}, meaning that these nodes do not adhere to the protocol.
Edge faults are mapped to node faults, i.e., if edge $((v,\ell),(w,\ell+1))$ is faulty, we instead consider $(v,\ell)$ (or $(w,\ell+1)$, if preferred) faulty.
We assume that each node fails independently with probability $p\in o(1/\sqrt{n})$.\footnote{We stress that this requirement is not stronger than that of~\cite{Dolev2016a,lenzen20trix} and~\cite{bund2019fault} for $f=1$ in any practical sense. If faults correlate in a way that they cluster together, it is likely that neighbors fail. Assuming independence (or, more generally, negative correlation) captures ``faults do not cluster'' in the most straightforward way that allows us to exploit this property beyond immediate neighbors.}
In particular, this entails that with probability $1-o(1)$, no node has two faulty predecessors, i.e., faults are $1$-local.
We assume this to be the case throughout our analysis.
Thus, for all $\ell\in \N$ and $v\in V$, $|(\{(v,\ell)\}\cup \bigcup_{\{v,w\}\in E}\{(w,\ell)\})\cap F|\le 1$.

Faulty nodes behave arbitrarily, subject to the constraint that at most a constant number of faulty nodes change their timing behavior between consecutive pulses.

%###
\paragraph*{Communication.}
%###
Each node has the ability to broadcast pulse messages on its outgoing edges.
If node $v_{\ell}\in V_{\ell}$ broadcasts at time $t_{v,\ell}$, its successors receive its message at (potentially different) times from $[t_{v,\ell}+d-u,t_{v,\ell}+d]$.
The maximum end-to-end \emph{delay} $d$ includes any delay caused by computation.
Typically, the delay \emph{uncertainty} $u$ is much smaller than $d$.
As discussed above, we assume delays to be static (or changing at a negligible rate cf.~\Cref{cor:bound}), i.e., each edge $e=((v,\ell),(w,\ell+1))$ has an unknown, but fixed associated delay $\delta_e\in [d-u,d]$ applied to each pulse sent from $(v,\ell)$ to $(w,\ell+1)$.

Note that faulty nodes can send pulses at arbitrary times, without being required to broadcast;
even if physical node implementations disallow point-to-point communication, edge faults could still result in this behavior.

%###
\paragraph*{Local Clocks and Computations.}
%###
Each node is able to approximately measure the progress of time by means of a local time reference.
We model this by node $(v,\ell)$ having query access to a \emph{hardware clock} $H_{v,\ell}\colon \Rpositive\to \Rpositive$ satisfying
\begin{align*}
\forall t < t'\in \Rpositive,\ t'-t\leq H_{v,\ell}(t') - H_{v,\ell}(t) \leq \vartheta(t' - t).
\end{align*}
for some $\vartheta>1$.
No known phase relation is assumed between the hardware clocks.
The algorithm will use them exclusively to measure how much time passes between local events.
As for delays, we assume that hardware clock speeds are static (or changing slowly).
This is justified in the same way as for delays.

Computations are deterministic.
However, in addition to receiving a message, the hardware clock reaching a time value previously determined by the algorithm can also trigger computations and possibly the broadcast of a pulse.

%###
\paragraph*{Output and Skew.}
%###
The goal of the algorithm is to synchronize the \emph{pulses} generated by correct nodes.
Our measure of quality is the worst-case \emph{local skew} the algorithm guarantees.
We define the local skew as the largest offset between the $k$-th pulses of adjacent nodes on the same layer or pulses $k$ and $k+1$ of adjacent nodes on layers $\ell$ and $\ell+1$, whichever is larger.
Formally, for $\ell \in \N$ we define
\begin{equation*}
\localskew_{\ell}:=\sup_{k\in \N}\max_{\substack{\{v,w\}\in E\\ (v,\ell),(w,\ell)\notin F}}\{|t_{v,\ell}^k-t_{w,\ell}^k|\},\quad
\localskew_{\ell,\ell+1}:=\sup_{k\in \N}\max_{\substack{((v,\ell),(w,\ell+1))\in E_{\ell}\\ (v,\ell),(w,\ell+1)\notin F}}\{|t_{v,\ell}^{k+1}-t_{w,\ell+1}^k|\},
\end{equation*}
and $\localskew :=\sup_{\ell \in \N}\max\{\localskew_{\ell},\localskew_{\ell,\ell+1}\}$.
This deviates from the standard definition of the local skew:
\begin{itemize}
\item The definition is adjusted to pulse synchronization, which can be viewed as an essentially equivalent time-discrete variant of clock synchronization~\cite{lecture}.
\item Between consecutive layers, we synchronize consecutive pulses. After initialization, which is complete once the first pulse propagated through the grid, this is equivalent to a layer-dependent index shift of pulse numbers.
\end{itemize}

We assume that correct nodes on layer $0$ generate well-synchronized pulses at times $t_{v,0}^k$ for $k\in \N_{>0}$ at a frequency of our choice.
For our purposes, it suffices that $\localskew_0\le \kappa$.
We discuss how to ensure this in \Cref{app:layer0}.
Other correct nodes generate pulses $t_{v,\ell}^k$, $k\in \N_{>0}$, based on the pulse messages received from their predecessors.

Any frequency error of layer~$0$ translates to an increase in $\localskew_{\ell,\ell+1}$ for all $\ell$.
For the sake of notational simplicity, we assume the frequency provided by layer~$0$ to match the one we choose perfectly, subsuming its error instead in the drift of the hardware clocks, i.e., $\vartheta$.
In other words, whatever drives the frequency of layer~$0$ is defining the ``true'' time $t$.
In practice, this has little effect, since the time reference used will be based on the best clock available in the system.

\section{Algorithm}\label{sec:algorithm}

In this section, we discuss the pulse forwarding algorithm.
We provide a simplified version of the algorithm that behaves identical so long as the predecessors of the executing node are correct.
The full algorithm needs to handle the possibility that faulty nodes send multiple messages or none at all.
This complicates bookkeeping and loop control, distracting from the principles underlying the algorithm's operation.
Accordingly, we defer the full algorithm to \Cref{app:algo}, where we show the equivalence to the simplified variant when there are no faulty predecessors.

\subsection{Simplified Pulse Forwarding Algorithm}

The algorithm proceeds in iterations corresponding to pulses.
In each iteration, node $(v,\ell)$
\begin{enumerate}
  \item timestamps the arrival times of the pulses of its predecessors using its hardware clock,
  \item determines a correction value $\Cor_{v,\ell}$ based on these timestamps, and
  \item forwards the pulse $\Lambda-d-\Cor_{v,\ell}$ time after receiving the pulse from $v_{\ell-1}$, measured by its hardware clock.  
\end{enumerate}
If all reception times are close to each other, then $\Cor_{v,\ell}$ will be small.
Recalling that messages are in transit for roughly $d$ time, this translates to $\Lambda$ being the nominal time for a pulse to propagate from layer $\ell-1$ to layer $\ell$.
We need to choose $\Lambda$ large enough such that the above sequence can be always realized.
That is, we need to consider how far apart the reception times of messages from the previous layer can be, and ensure that $\Lambda-d$ exceeds this value plus the resulting correction $\Cor_{v,\ell}$.
\begin{algorithm}[t]
  \caption{Simplified pseudocode for discrete GCS at node $(v,\ell)$, $\ell>0$. As shown in \Cref{lem:equivalence}, this code is equivalent to~\Cref{alg:Discretised_Gradient_TRIX} in the absence of faults. The parameters $\Lambda$ and $\kappa$ will be determined later, based on the analysis.}
  \label{alg:Simplified_Gradient_TRIX}
  \begin{algorithmic}
% \SetKwRepeat{Repeat}{do}{until}
    \Loop
      \State $H_{\own},H_{\min},H_{\max} := \infty$
      \Repeat
        \If{received pulse from $(v,\ell-1)$}
          \State $H_{\own} := H_{v,\ell}(t)$
        \EndIf
        \If{received pulse from first $(w,\ell-1)$, $\{v,w\}\in E$}
          \State $H_{\min}:=H_{v,\ell}(t)$
        \EndIf
        \If{received pulse from last $(w,\ell-1)$, $\{v,w\}\in E$}
          \State $H_{\max}:=H_{v,\ell}(t)$
        \EndIf
      \Until{$H_{\own}, H_{\min}, H_{\max}<\infty$}
      \State $\Cor_{v,\ell}:=\min_{s\in \N}\{\max\{H_{\own}-H_{\max}+4s\kappa,H_{\own}-H_{\min}-4 s\kappa \}\}-\kappa/2$
      \If{$\Cor_{v,\ell}<0$}
        \State $\Cor_{v,\ell}:=\min\{H_{\own}-H_{\min} -\kappa/2+2\kappa,0\}$
      \ElsIf{$\Cor_{v,\ell}>\vartheta\kappa$}
        \State $\Cor_{v,\ell}:=\max\{H_{\own}-H_{\max} -\kappa/2-\kappa,\vartheta\kappa\}$
      \EndIf
      \State wait until $H_{v,\ell}(t) = H_{\own} + \Lambda - d - \Cor_{v,\ell}$
      \State \textbf{broadcast pulse}
    \EndLoop
  \end{algorithmic}
\end{algorithm}
Assuming that this precondition holds, \Cref{alg:Simplified_Gradient_TRIX} implements the above approach.
In each loop iteration, it initializes three reception times to $\infty$:
\begin{itemize}
  \item $H_{\own}$, which stores the arrival time of the pulse from $(v,\ell-1)$. From the perspective of the simulated GCS algorithm, this reflects the state of the node $v\in V$ simulated by $(v,\ell)$, $\ell\in \N$.
  \item $H_{\min}$, which stores the minimum arrival time of a pulse from a neighbor $w_{\ell-1}$, $w\neq v$. This corresponds to the first pulse received from a neighbor $w$ of $v$ in $G$ in this iteration.
  \item $H_{\max}$, which stores the maximum arrival time of a pulse from a neighbor $w_{\ell-1}$, $w\neq v$. This corresponds to the last pulse received from a neighbor $w$ of $v$ in $G$ in this iteration.
\end{itemize}
The do-until loop fills these variables with the correct values.
At the heart of the algorithm lies the computation of $C_{v,\ell}$.
If there were no faults, one could always compute
\begin{equation*}
\Delta:=\min_{s\in \N}\left\{\max\{H_{\own}-H_{\max}+4s\kappa,H_{\own}-H_{\min}-4 s\kappa \}\right\}-\frac{\kappa}{2}
\end{equation*}
and then choose the closest value from the range $[0,\vartheta \kappa]$, i.e., set (i) $\Cor_{v,\ell}:=0$ if $\Delta<0$, (ii) $\Cor_{v,\ell}:=\vartheta \kappa$ if $\Delta > \vartheta \kappa$, and (iii) $\Cor_{v,\ell}:=\Delta$ else.
% \begin{equation*}
% \Cor_{v,\ell}:=\begin{cases}
% \Delta & \mbox{if $\Delta \in [0,\vartheta \kappa]$,}\\
% 0 & \mbox{if $\Delta <0$, and}\\
% \vartheta \kappa & \mbox{if $\Delta > \vartheta \kappa$.}
% \end{cases}
% \end{equation*}

To get intuition on this choice, observe that $\min_{x\in \R}\{\max\{H_{\own}-H_{\max}+x,H_{\own}-H_{\min}-x\}\}$ is attained when $H_{\own}-H_{\max}+x=H_{\own}-H_{\min}-x$.
This is equivalent to $x=(H_{\max}-H_{\min})/2$, i.e., if $\kappa$ was infinitesimally small, we had that $H_{\own}-\Delta=(H_{\max}+H_{\min})/2$.
Moreover, if each node could accurately determine the time each pulse received by it was sent, the reception times of the pulse messages could serve as exact proxies for the actual pulse forwarding times of the nodes on layer $\ell-1$.
In iteration $k$, this would mean to generate the pulse at $(v,\ell)$ faster if $(v,\ell-1)$ generated its pulse later than the average of $\min_{\{v,w\}\in E}\{t_{w,\ell-1}^k\}$ and $\max_{\{v,w\}\in E}\{t_{w,\ell-1}^k\}$.
Thus, any $(v,\ell)$ for which $t_{v,\ell-1}^k-\min_{\{v,w\}\in E}\{t_{w,\ell-1}^k\}>\max_{\{v,w\}\in E}\{t_{w,\ell-1}^k\}-t_{v,\ell-1}^k$ would choose $C_{v,\ell}>0$, attempting to reduce $\max_{\{v,w\}\in E}\{|t_{v,\ell}^k-t_{w,\ell}^k|\}$ compared to $\max_{\{v,w\}\in E}\{|t_{v,\ell-1}^k-t_{w,\ell-1}^k|\}$.
This can be viewed as trying to reduce the local skew by a greedy strategy.

Unfortunately, this naive strategy fails to account for inaccuracies due to message delay uncertainty and drifting hardware clocks.
Nonetheless, we follow this strategy up to deviations of $O(\kappa)$.
The additional terms serve the following purposes:
\begin{itemize}
  \item Considering only discrete choices for $x\in 4\kappa \N$ rather than arbitrary $x\in \R$ is the key ingredient that makes the algorithmic approach succeed, cf.~\cite{kuhn09reference}. Essentially, this is necessary because there is no way to determine $t_{v_{\ell-1},k}-t_{w_{\ell-1},k}$ precisely. Discretizing observed skews in units of $\kappa\in \Theta(u+(\vartheta-1)(\Lambda-d))$ enables a delicate strategy that alternates between overestimating skews to locally generate the next pulse earlier for the sake of ``catching up'' with others and underestimating skews to ``wait'' for others catch up.
  \item Substracting $\kappa/2$ accounts for errors in measuring skews, which are caused by uncertainty in message delay and hardware clock speed.
  \item To limit the damage done by a faulty predecessor of $(v,\ell)$, we ensure that $(v,\ell)$ generates its pulse without too large of a deviation from the \emph{median} of $t_{v,\ell-1}$, $\min_{\{v,w\}\in E}\{t_{w,\ell-1}^k\}$, and $\max_{\{v,w\}\in E}\{t_{w,\ell-1}^k\}$ (plus the nominal offset of $\Lambda$). This is achieved by permitting corrections $C_{v,\ell}<0$ if $(v,\ell-1)$ clearly generated its pulse earlier than $\min_{\{v,w\}\in E}\{t_{w,\ell-1}^k\}$ and $C_{v,\ell}>\vartheta \kappa$ if it clearly generated its pulse later than $\max_{\{v,w\}\in E}\{t_{w,\ell-1}^k\}$, respectively.
\end{itemize}

To further motivate the last point, recall that there can be at most one fault among the predecessors of $(v,\ell)$.
A single faulty predecessor can affect only one of the three values $H_{\own}$, $H_{\min}$, and $H_{\max}$: control $H_{\own}$ arbitrarily, $H_{\min}$ to be smaller than the minimum reception time from a correct node $(w,\ell-1)$, $\{v,w\}\in E$, or $H_{\max}$ to exceed the maximum reception time from correct nodes $(w,\ell-1)$, $\{v,w\}\in E$.
Hence, ensuring that pulses are generated with only a small offset relative to $\median{H_{\own}, H_{\min}, H_{\max}}+\Lambda-d$ indeed limits the damage that a fault can do.

Achieving all of the desired properties is non-trivial, leading to the fairly involved choice of $\Cor_{v,\ell}$.
It can be viewed as simultaneously implementing relaxed fast and slow conditions (as introduced in~\cite{kuhn09reference}), an additional jump condition required to make the GCS algorithm work under these relaxed fast and slow conditions, and the requirement to stick close to the median of predecessors' pulse times.
In \Cref{sec:conditions}, we specify the (relaxed) slow and fast condition, as well as the jump condition, and show that the algorithm implements them.
\Cref{lem:fault_self,lem:fault_neighbor} show that the algorithm also enforces that pulses deviate little from the time interval spanned by correct predecessors (offset by $\Lambda$).

There is some freedom in the choice of parameters.
For simplicity, we fix a good choice of $\kappa$ and note that $d$ must satisfy a lower bound $B\in O(\sup_{\ell\in \N}\{\localskew_{\ell}\}+\kappa)$.
Observe that this constraint simply means that the skew bounds are useful, as a skew that is of similar size as the maximum end-to-end delay requires to slow the system down substantially.
Finally, $\Lambda$ must be at least $d+O(\sup_{\ell\in \N}\{\localskew_{\ell}\})$, which due to the previous constraint holds e.g.\ for the choice $\Lambda=2d$.
Formally, for a sufficiently large constant $C$,%\footnote{We do not attempt to optimize constants in this work.} 
\begin{align}
\kappa & := 2\left(u+\left(1-\frac{1}{\vartheta}\right)(\Lambda-d)\right),\label{eq:kappa}\\
\Lambda &\ge C\vartheta(\sup_{\ell\in \N}\{\localskew_{\ell}\} + u) + d,\mbox{ and}\label{eq:Lambda}\\
d &\ge C(\vartheta(\sup_{\ell\in \N}\{\localskew_{\ell}\} + u)+\kappa).\label{eq:d}
\end{align}

%###
\subsubsection*{Complete Algorithm}
%###
The complete algorithm cannot wait for messages from all predecessors to determine when to send its pulse, as a faulty node not sending its pulse then would deadlock all its descendants.
As discussed above, the hardware clock time of the next pulse time does not deviate \emph{much} from $\median{H_{\own}, H_{\min}, H_{\max}}+\Lambda-d$, but does depend on $\max\{H_{\min},H_{\own},H_{\max}\}$ in some cases.
However, we will prove that $\localskew_{\ell-1}$ is small enough such that all pulse messages from correct nodes will be received in time.
Hence, it is sufficient to wait until $\median{H_{\own}, H_{\min}, H_{\max}}+\vartheta \localskew_{\ell-1}$ (or later) according to $H_{v,\ell}$.
Provided that $\Lambda-d$ is large enough, this implies that any message for computing $\Cor_{v,\ell}$ missing is due to a fault;
in fact, at the point in time when this becomes clear, $\Cor_{v,\ell}$ is already determined, regardless of how late the message would arrive.

The complete algorithm differs from \Cref{alg:Simplified_Gradient_TRIX} by covering the case that a signal does not arrive in time.
Intuitively, one can treat the respective message arrival time ($H_{\own}$ or $H_{\max}$, $H_{\min}$ is not possible) as $\infty$, while allowing such an $\infty$ to cancel out in substraction:
\begin{itemize}
  \item If $H_{\own}=\infty$, then $C_{v,\ell}\in H_{\own}-H_{\max}-O(\kappa)$, and $(v,\ell)$ will generate its pulse at local time $H_{\own}+\Lambda-d-C_{v,\ell}\in H_{\max}+\Lambda-d+O(\kappa)$.
  \item If $H_{\max}=\infty$ and $H_{\own}\ge H_{\min}$, then $C_{v,\ell}\in H_{\own}-H_{\min}\pm \Theta(\kappa)$ and $(v,\ell)$ will generate its pulse at local time $H_{\own}+\Lambda-d-C_{v,\ell}\in H_{\min}+\Lambda-d\pm O(\kappa)$.
  \item If $H_{\max}=\infty$ and $H_{\own}< H_{\min}$, then $C_{v,\ell}\in [0,2\kappa]$ and $(v,\ell)$ will generate its pulse at local time $H_{\own}+\Lambda-d-C_{v,\ell}\in H_{\own}+\Lambda-d-O(\kappa)$.
\end{itemize}
Note that in all cases, the pulse is generated with an offset of $\Lambda-d-\Theta(\kappa)$ from the median reception time.
The complete algorithm follows the above intuition, leveraging the fact that there is no need to wait indefinitely to determine that the missing signal is late, and is given in \Cref{app:algo}.

Last, but not least, it is of interest to make the pulse forwarding algorithm self-stabilizing~\cite{dijkstra74self}.
Due to the design choice of propagating the clock signal from a single source along a DAG, this will immediately translate to the overall scheme being self-stabilizing, so long as the clock generation is self-stabilizing, too.
This is straightforward, because one can assume that the signals from the previous layer are already well-synchronized.
Thus, all that nodes need to do is to detect when all but possibly one (faulty) pulse signal arrive in close temporal proximity to determine when to clear their memory and start a new iteration of the main loop.
In~\Cref{app:self-stab}, we discuss how this can be achieved using standard techniques.

\section{Analysis}\label{sec:analysis}

We now analyze the pulse progagation scheme under the assumption that layer~$0$ generates well-synchronized pulses.
We discuss a suitable method for achieving this in \Cref{app:layer0}.
Our analysis proceeds along the following lines:
\begin{enumerate}
  \item We show that, if the local skew is small enough compared to $\Lambda$, i.e., \Cref{eq:Lambda} holds, all correct nodes execute their iterations as intended. That is, each correct node on layer $\ell>0$ receives the $k$-th pulses of its correct predecessors in its $k$-th loop iteration. This is deferred to \Cref{app:algo}. We then proceed under the assumption that this holds true, which will be justified retroactively once we establish that the local skew is bounded.
  \item Since delays and hardware clock speeds are (approximated as being) static, any (substantial) change in relative timing of consecutive pulses is due to faulty nodes. Thus, the task of bounding the local skew reduces to bounding the intra-layer skew $\localskew_{\ell}$ for a single pulse, since such a bound must take into account the full variability introduced by faulty nodes. This reasoning is deferred to \Cref{app:final}.
  \item Based on potentials, we analyze $\localskew_{\ell}$ in the absence of faults. The results entail not only bounded skew, but also that the potentials recover if they become unexpectedly large.
  \item We show that faulty nodes have limited impact on the potentials. From this and the above recovery property, we infer that skews behave favorably also when there are faults.
\end{enumerate}
As stated above, the first two steps of our line of reasoning are deferred to the appendix, alongside some basic helper lemmas given in \Cref{app:basic}.
The main challenge is to bound $\localskew_{\ell}$ for a single pulse.
Due to the first step, we know that the $k$-th pulse at correct nodes depends only on the $k$-th pulses of their predecessors (\Cref{lem:correct_pulsing}).
Therefore, in the following fix $k$ and denote the $k$-th pulse time of correct $(v,\ell)\in V_G$ by $t_{v,\ell}$.

Recall that for $v,w\in V$, we denote by $d(v,w)$ their distance in the base graph $H$.
Our analysis is built around the following potential functions.
\begin{definition}[Potential Functions]
  Let $v,w\in V$ and $s,\ell\in \N$.
  We define
  \begin{align*}
    \psi^s_{v,w}(\ell) &:= t_{v,\ell} - t_{w,\ell} - 4s\kappa d(v,w),
    &\Psi^s(\ell) := \max_{v,w\in V}\{\psi^s_{v,w}(\ell)\},\\
    \xi^s_{v,w}(\ell) &:= t_{v,\ell} - t_{w,\ell} - (4s-2)\kappa d(v,w),~\mbox{and} 
    &\Xi^s(\ell) := \max_{v,w\in V}\{\Xi^s_{v,w}(\ell)\}.
  \end{align*}
\end{definition}
Bounding $\Psi^s(\ell)$ readily translates to bounding $\localskew_{\ell}$.
\begin{observation}\label{obs:skew}
If for $s,\ell\in \N$ and some $\Psi^s\in \Rpositive$ it holds that $\Psi^s(\ell)\le \Psi^s$, then $\localskew_{\ell}\le \Psi^s+4s\kappa$.
\end{observation}
\begin{proof}
Fix $k\in \N$ and suppose that $\{v,w\}\in E$ maximizes $|t_{v,\ell} - t_{w,\ell}|$.
W.l.o.g., assume that $t_{v,\ell}\ge t_{w,\ell}$.
Since $\{v,w\}\in E$, we have that $d(v,w)=1$.
Hence,
% \begin{equation*}
$|t_{v,\ell} - t_{w,\ell}|=t_{v,\ell} - t_{w,\ell}=\psi^s_{v,w}(\ell)+4s\kappa\le \Psi^s(\ell)+4s\kappa\le \Psi^s+4s\kappa$.
% \end{equation*}
Since $k\in \N$ is arbitrary, it follows that $\localskew_{\ell}\le \Psi^s+4s\kappa$.
\end{proof}
In summary, the goal of our analysis will be to bound $\Psi^s(\ell)$ by a small value for some $s$ satisfying $4s\kappa\in O(u\log D)$.

We first study the behavior of the algorithm if there are no faults.
Accordingly, this will be tacitly assumed in all statements of this section, with the expection of \Cref{sec:faults}.
Note that by \Cref{lem:equivalence}, this means that we may also tacitly assume that \Cref{alg:Simplified_Gradient_TRIX} is run by all nodes in layers $\ell\in \N_{>0}$.
In \Cref{sec:faults}, we will then bound the impact of faulty layers on the potential.

\subsection{The Slow, Fast, and Jump Conditions}\label{sec:conditions}

The key to bounding the local skew without faults is to find the right balance between two conflicting goals:
choosing $\Cor_{v,\ell}$ large enough to ``catch up'' to predecessors $w_{\ell-1}\neq v_{\ell-1}$ that generated their pulse earlier than $v_{\ell-1}$, but small enough to ``wait'' for predecessors $w_{\ell-1}\neq v_{\ell-1}$ that generated their pulse later than $v_{\ell-1}$.
The following condition, illustrated in \Cref{fig:SC_and_FC_demo}, captures what we need regarding the latter.
\begin{definition}[Slow Condition]\label{def:sc}
For all $s\in \N$, correct layers $\ell-1\in \N$, and $v_{\ell}\in V_{\ell}\setminus F$, we require the \emph{slow condition} $\operatorname{\SC}(s):=\SCone(s) \lor \SCtwo(s) \lor \SCthree$ to hold, where
\begin{align*}
\SCone(s)\colon & \frac{\Cor_{v,\ell}}{\vartheta}\le t_{v,\ell-1} - \max_{\{v,w\}\in E}\{t_{w,\ell-1}\}+4s\kappa\\
\SCtwo(s)\colon & \frac{\Cor_{v,\ell}}{\vartheta}\le t_{v,\ell-1} - \min_{\{v,w\}\in E}\{t_{w,\ell-1}\}-4s\kappa\\
\SCthree\colon &  \Cor_{v,\ell} \le 0.
\end{align*}
\end{definition}
This can be viewed as a variant of the slow condition from~\cite{kuhn09reference}, adjusted to our setting by quantifying by how much $v_{\ell}$ may safely shift the timing of its pulse.
The main conceptual difference to~\cite{kuhn09reference} is that we relax the slow condition by adding $\SCthree$.
In what follows, we drop $s$ from the notation when it is clear from context.

The fast condition, also illustrated in \Cref{fig:SC_and_FC_demo}, is the counterpart to \Cref{def:sc} addressing the need to ``catch up'' to neighbors that are ahead.
\begin{definition}[Fast Condition]\label{def:fc}
For all $s\in \N_{>0}$, correct layers $\ell-1\in \N_{>0}$, and $v_{\ell}\in V_{\ell}\setminus F$, we require the \emph{fast condition} $\FC(s):=\FCone(s) \lor \FCtwo(s) \lor \FCthree$ to hold, where
\begin{align*}
\FCone(s)\colon & \Cor_{v,\ell}\ge t_{v,\ell-1} - \max_{\{v,w\}\in E}\{t_{w,\ell-1}\}+(4s-2)\kappa+\kappa\\
\FCtwo(s)\colon & \Cor_{v,\ell}\ge t_{v,\ell-1} - \min_{\{v,w\}\in E}\{t_{w,\ell-1}\}-(4s-2)\kappa+\kappa\\
\FCthree\colon &  \Cor_{v,\ell} \ge \kappa.
\end{align*}
\end{definition}
This can be viewed as a variant of the fast condition from~\cite{kuhn09reference}, adjusted to our setting by quantifying by how much $v_{\ell}$ may safely shift the timing of its pulse.
The main conceptual difference to~\cite{kuhn09reference} is that we relax the fast condition by adding $\FCthree$.

In addition, note that there is an additive term of $\kappa$ that does not change sign.
Its purpose is to account for the fact that our simulation of the GCS algorithm from~\cite{lenzen2010tight} operates in discrete time steps corresponding to the layers.
The continuous versions of the GCS algorithm in~\cite{kuhn09reference,kuhn10dynamic,lenzen2010tight} can choose this term arbitrarily small.
In contrast, we need it to exceed the maximum error in time measurement accumulated in a step.
We remark that, in principle, one could choose this term different from $\kappa$.
However, since both need to meet the same lower bound of $u+(1-1/\vartheta)(\Lambda-d)$, there is no asymptotic gain in introducing a separate parameter.

Our relaxation of the slow and fast conditions adds a substantial complication.
From the perspective of the time-continuous variant of the algorithm in~\cite{kuhn09reference}, we now allow for arbitrarily large clock ``jumps,'' rather than bounded clock rates.
In our discrete version, the rate bound from~\cite{kuhn09reference} corresponds to $\Cor_{v,\ell}\in [0,\vartheta \kappa]$.
Without this additional constraint, the slow and fast conditions are insufficient to bound skews.

This is illustrated in \Cref{fig:need_for_JC}, showing an execution that satisfies $\SC$ and $\FC$, but suffers from skews that grow without bound.
The key issue is that adjacent nodes could ``jump'' in opposite directions, resulting in an oscillatory behavior in which measurement errors accumulate indefinitely.
To avoid this kind of behavior, we add an additional condition that ``dampens'' such oscillations, yet limits by how much a faulty predecessor can cause an increase in skew.
\begin{definition}[Jump Condition]\label{def:jump}
For all correct layers $\ell-1\in \N_{>0}$ and $v_{\ell}\in V_{\ell}\setminus F$, we require the \emph{jump condition} $\JC:=\JCone \lor \JCtwo\lor \JCthree$ to hold, where
\begin{align*}
\JCone\colon & \kappa<\frac{\Cor_{v,\ell}}{\vartheta}\le t_{v,\ell-1} - \max_{\{v,w\}\in E}\{t_{w,\ell-1}\}-\kappa\\
\JCtwo\colon & 0>\Cor_{v,\ell}\ge t_{v,\ell-1} - \min_{\{v,w\}\in E}\{t_{w,\ell-1}\}+\kappa\\
\JCthree\colon & 0\le \frac{\Cor_{v,\ell}}{\vartheta}\le \kappa.
\end{align*}
\end{definition}
We prove that the slow, fast, and jump condition are correctly implemented in \Cref{lem:slow_holds,lem:fast_holds,lem:jump_holds} in \Cref{app:basic}.

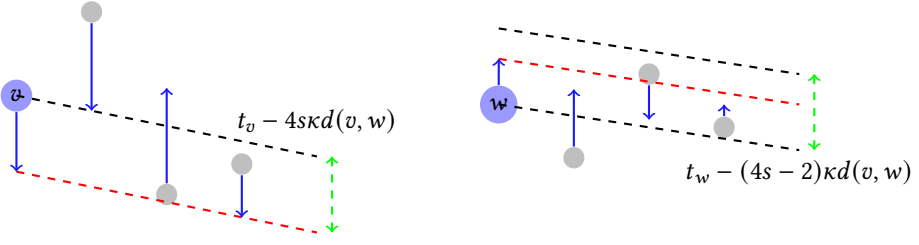
\begin{figure}[t!]
  \centering
  \begin{subfigure}{0.4\textwidth}
    \centering
    \begin{tikzpicture}[main/.style = {draw, circle},scale = 2.0] 
    \tikzstyle{vertex} = [
            circle,
            fill = {
            black! 25
            },
            minimum size=8pt, 
            inner sep=2pt] 
    \tikzstyle{invisiblevertex} = [
        circle,
        draw = none,
        fill = {
        white!24
        },
        minimum size=0pt, 
        inner sep=2pt] 
    %example
    \node[vertex,fill=blue!40] (main) at (-1.0,0) {$v$};

    %\node[vertex] (left_neighbour) at (-0.5,-0.35) {};
    %\node[invisiblevertex] (left_neighbour_virt) at (-0.5,0) {};

    %\node[vertex] (left_left_neighbour) at (-1.0,-0.25) {};
    %\node[invisiblevertex] (left_left_neighbour_virt) at (-1.0,0) {};
    
    \node[vertex] (r1_neighbour) at (-0.5,0.55) {};
    \node[invisiblevertex] (r1_neighbour_virt) at (-0.5,0) {};

    \node[vertex] (r2_neighbour) at (0.0,-0.65) {};
    \node[invisiblevertex] (r2_neighbour_virt) at (0.0,0) {};

    \node[vertex] (r3_neighbour) at (0.5,-0.45) {};
    \node[invisiblevertex] (r3_neighbour_virt) at (0.5,0) {};
    
    %\draw[thick, dotted] (r3_neighbour_virt)--(main); %--(left_left_neighbour_virt);

    %\draw[<->, dotted, thick] (left_neighbour)--(left_neighbour_virt) node[midway, left] {};

    \draw[dashed, thick] (-1.0, 0)--(1.0,-0.40) node [right,above=0.5em] {$t_v - 4s\kappa d(v,w)$};
    %\draw[dashed, thick] (-1.0, 0.5)--(1.0,0.90) node [right,above] {};
    \draw[dashed, thick, red] (-1.0, -0.5)--(1.0,-0.90) node [right,below, black] {};
    \draw[<->,color=green, thick, dashed] (1.1,-0.4)--(1.1,-0.9) node[right,midway] {};

    %\draw[->,dashed,thick] (main)--(-1.0,0.50) node [midway,left] {$\Cor_{v,\ell} < 0$};
    %\draw[-,dotted,thick] (-1.0,0.5)--(-1.0,0.7) node [midway,left] {};

    \draw[->,thick, color=blue!90] (main)--(-1.0,-0.50) node [midway,left] {}; %{$0 < \Cor_{v,\ell}$};
    \draw[->,thick,color=blue!90] (r1_neighbour)--(-0.5, -0.10) node [midway,left] {};
    \draw[->,thick,color=blue!90] (r2_neighbour)--(0.0,0.05) node [midway,left] {};
    \draw[->,thick,color=blue!90] (r3_neighbour)--(0.5,-0.8) node [midway,left] {};
\end{tikzpicture}
  \end{subfigure}
  \hspace{2em}
  \begin{subfigure}{0.4\textwidth}
    \centering
     \begin{tikzpicture}[main/.style = {draw, circle},scale = 2.0] 
    \tikzstyle{vertex} = [
            circle,
            fill = {
            black! 25
            },
            minimum size=8pt, 
            inner sep=2pt] 
    \tikzstyle{invisiblevertex} = [
        circle,
        draw = none,
        fill = {
        white!24
        },
        minimum size=0pt, 
        inner sep=2pt] 
    %example
    \node[vertex,fill=blue!40] (main) at (-1.0,0) {$w$};

    %\node[vertex] (left_neighbour) at (-0.5,-0.35) {};
    %\node[invisiblevertex] (left_neighbour_virt) at (-0.5,0) {};

    %\node[vertex] (left_left_neighbour) at (-1.0,-0.25) {};
    %\node[invisiblevertex] (left_left_neighbour_virt) at (-1.0,0) {};
    
    \node[vertex] (r1_neighbour) at (-0.5,-0.35) {};
    \node[invisiblevertex] (r1_neighbour_virt) at (-0.5,0) {};

    \node[vertex] (r2_neighbour) at (0.0,0.20) {};
    \node[invisiblevertex] (r2_neighbour_virt) at (0.0,0) {};

    \node[vertex] (r3_neighbour) at (0.5,-0.15) {};
    \node[invisiblevertex] (r3_neighbour_virt) at (0.5,0) {};
    
    %\draw[thick, dotted] (r3_neighbour_virt)--(main); %--(left_left_neighbour_virt);

    %\draw[<->, dotted, thick] (left_neighbour)--(left_neighbour_virt) node[midway, left] {};

    ;
    \draw[dashed, thick] (-1.0, 0)--(1.0,-0.30) node [right,below] {$t_w-(4s-2)\kappa d(v,w)$};
    \draw[dashed, thick] (-1.0, 0.5)--(1.0,0.2) node [right,above, black] {};
    \draw[dashed, thick, red] (-1.0, 0.3)--(1.0,0.0) node [right,below] {};
    \draw[<->,color=green, thick, dashed] (1.1,-0.3)--(1.1,0.2) node[right,midway] {};

    \draw[->,thick,color=blue!90] (main)--(-1.0,0.30) node [midway,left] {}; %{$\Cor_{w,\ell} < \kappa$};
    %\draw[->,dashed,thick] (main)--(-1.0,-0.40) node [midway,left] {$\Cor_{w,\ell} > \kappa$};
    %\draw[-,dotted,thick] (-1.0,-0.40)--(-1.0,-0.60) node [midway,left] {};
    \draw[->,thick,color=blue!90] (r1_neighbour)--(-0.5,0.10) node [midway,left] {};
    \draw[->,thick,color=blue!90] (r2_neighbour)--(0.0,-0.1) node [midway,left] {};
    \draw[->,thick,color=blue!90] (r3_neighbour)--(0.5,0.0) node [midway,left] {};
\end{tikzpicture}
  \end{subfigure}
  \caption{Slow condition (left) and fast condition (right). $\SC(s)$ is tailored to ensuring that $\max_{w\in V}\{\psi_{v,w}^s(\ell)\}$ (the length of the green arrow) cannot grow quickly. Nodes $w$ with $\Cor_{w,\ell}\le 0$ ($\SCthree$ holds) cannot apply a correction pushing them below the red line. If $\Cor_{w,\ell}>0$, then both $\SCone$ and $\SCtwo$ will ensure that there is a neighbor $x$ of $w$ such that the offset of $t_{w,\ell-1}-\Cor_{w,\ell}/\vartheta$ to the black line does not exceed the one of $t_{x,\ell-1}$. In other words, $\SC$ ensures that the blue arrows indicating $\Cor_{w,\ell}/\vartheta$ do not reach below the red line. This means that any increase of $\max_{w\in V}\{\psi_{v,w}^s(\ell)\}$ is caused by delay and clock speed variation, which in turn is bounded by $\kappa/2$ per layer. Similarly, $\FC(s)$ is tailored to ensuring that $\max_{v\in V}\{\xi_{v,w}^s(\ell)\}$ (the length of the green arrow), if positive, decreases by at least $\kappa/2$. To ensure this, $\Cor_{w,\ell}$ (indicated by blue arrows) must be large enough to reach below the red line. This is achieved by $\FC(s)$ having an additional ``slack'' term of $\kappa$, which overcomes the ``loss'' of $\kappa/2$ due to uncertainty.\label{fig:SC_and_FC_demo}}
  %\Description[Illustrating the Slow and Fast Conditions]{Slow condition (left) and fast condition (right). $\SC(s)$ is tailored to ensuring that $\max_{w\in V}\{\psi_{v,w}^s(\ell)\}$ (the length of the green arrow) cannot grow quickly. Nodes $w$ with $\Cor_{w,\ell}\le 0$ ($\SCthree$ holds) cannot apply a correction pushing them below the red line. If $\Cor_{w,\ell}>0$, then both $\SCone$ and $\SCtwo$ will ensure that there is a neighbor $x$ of $w$ such that the offset of $t_{w,\ell-1}-\Cor_{w,\ell}/\vartheta$ to the black line does not exceed the one of $t_{x,\ell-1}$. In other words, $\SC$ ensures that the blue arrows indicating $\Cor_{w,\ell}/\vartheta$ do not reach below the red line. This means that any increase of $\max_{w\in V}\{\psi_{v,w}^s(\ell)\}$ is caused by delay and clock speed variation, which in turn is bounded by $\kappa/2$ per layer. Similarly, $\FC(s)$ is tailored to ensuring that $\max_{v\in V}\{\xi_{v,w}^s(\ell)\}$ (the length of the green arrow), if positive, decreases by at least $\kappa/2$. To ensure this, $\Cor_{w,\ell}$ (indicated by blue arrows) must be large enough to reach below the red line. This is achieved by $\FC(s)$ having an additional ``slack'' term of $\kappa$, which overcomes the ``loss'' of $\kappa/2$ due to uncertainty.}
  \Description[Illustrating the Slow and Fast Conditions]{}
\end{figure}
\begin{figure}[t]
  \centering
  \begin{subfigure}{0.2\textwidth}
    \centering
    \begin{tikzpicture}
      \node[text width=1.2cm] at (1,-0.5) {};
      \node[draw,text width=1.4cm] at (1,0.5) {Layer $\ell$};
      \node[draw,text width=2cm] at (1,3.5) {Layer $\ell+1$};
      \node[draw,text width=2cm] at (1,7.0) {Layer $\ell+2$};
    \end{tikzpicture}
  \end{subfigure}
  \begin{subfigure}{0.35\textwidth}
    \centering
    \begin{tikzpicture}[main/.style = {draw, circle},scale = 2.0] 
    \tikzstyle{vertex} = [
            circle,
            fill = {
            black! 25
            },
            minimum size=8pt, 
            inner sep=2pt] 
    
    %example
    \node[vertex,fill=blue!40] (main) at (0,0.55) {};
    \node[vertex,fill=white, draw=none] (main_virt) at (0,-0.85) {};

    \node[vertex] (left_neighbour) at (-0.5,-0.55) {};
    \node[vertex, draw=none, fill=white, dotted] (left_neighbour_virt) at (-0.5,0.55) {};

    \node[vertex] (left_left_neighbour) at (-1.0,0.55) {};
    \node[vertex, draw=none, fill=white, dotted] (left_left_neighbour_virt) at (-1.0,-0.85) {};
    
    \node[vertex] (right_neighbour) at (0.5,-0.55) {};
    \node[vertex, draw=none, fill=white, dotted] (right_neighbour_virt) at (0.5,0.55) {};

    \node[vertex] (right_right_neighbour) at (1.0,0.55) {};
    \node[vertex, draw=none, fill=white, dotted] (right_right_neighbour_virt) at (1.0,-0.85) {};

    \draw[->, color=blue!90, thick] (right_neighbour)--(right_neighbour_virt) node[midway, right] {};
    \draw[->, color=blue!90, thick] (right_right_neighbour)--(right_right_neighbour_virt) node[midway, right] {};
    \draw[->, color=blue!90, thick] (left_neighbour)--(left_neighbour_virt) node[midway, left] {};
    \draw[->, color=blue!90, thick] (left_left_neighbour)--(left_left_neighbour_virt) node[midway, left] {};
    \draw[->, color=blue!90, thick] (main)--(main_virt);
    
    \draw[dashed] (right_right_neighbour)--(right_neighbour)--(main)--(left_neighbour)--(left_left_neighbour);
\end{tikzpicture}

\vspace{2em}
\begin{tikzpicture}[main/.style = {draw, circle},scale = 2.0] 
    \tikzstyle{vertex} = [
            circle,
            fill = {
            black! 25
            },
            minimum size=8pt, 
            inner sep=2pt] 
    
    %example
    \node[vertex,fill=blue!40] (main) at (0,-0.35) {};
    \node[vertex,fill=white, draw=none] (main_virt) at (0,0.35) {};

    \node[vertex] (left_neighbour) at (-0.5,0.35) {};
    \node[vertex, draw=none, fill=white, dotted] (left_neighbour_virt) at (-0.5,-0.55) {};
   
    \node[vertex] (left_left_neighbour) at (-1.0,-0.35) {};
    \node[vertex, draw=none, fill=white, dotted] (left_left_neighbour_virt) at (-1.0,0.35) {};

    \node[vertex] (right_neighbour) at (0.5,0.35) {};
    \node[vertex, draw=none, fill=white, dotted] (right_neighbour_virt) at (0.5,-0.55) {};
    
    \node[vertex] (right_right_neighbour) at (1.0,-0.35) {};
    \node[vertex, draw=none, fill=white, dotted] (right_right_neighbour_virt) at (1.0,0.35) {};

    \draw[->, color=blue!90, thick] (right_neighbour)--(right_neighbour_virt) node[midway, right] {};
    \draw[->, color=blue!90, thick] (right_right_neighbour)--(right_right_neighbour_virt) node[midway, right] {};
    \draw[->, color=blue!90, thick] (left_neighbour)--(left_neighbour_virt) node[midway, left] {};
    \draw[->, color=blue!90, thick] (left_left_neighbour)--(left_left_neighbour_virt) node[midway, left] {};
    \draw[->, color=blue!90, thick] (main)--(main_virt);

    \draw[dashed] (right_right_neighbour)--(right_neighbour)--(main)--(left_neighbour)--(left_left_neighbour);
\end{tikzpicture}

\vspace{2em}
\begin{tikzpicture}[main/.style = {draw, circle},scale = 2.0] 
    \tikzstyle{vertex} = [
            circle,
            fill = {
            black! 25
            },
            minimum size=8pt, 
            inner sep=2pt] 
    
    %example
    \node[vertex,fill=blue!40] (main) at (0,0.25) {};
    \node[vertex,draw=none, fill=white] (main_virt) at (0,-0.55) {};

    \node[vertex] (left_neighbour) at (-0.5,-0.25) {};
    \node[vertex, draw=none, fill=white, dotted] (left_neighbour_virt) at (-0.5,0.25) {};

    \node[vertex] (left_left_neighbour) at (-1.0,0.25) {};
    \node[vertex, draw=none, fill=white, dotted] (left_left_neighbour_virt) at (-1.0,-0.55) {};
    
    \node[vertex] (right_neighbour) at (0.5,-0.25) {};
    \node[vertex, draw=none, fill=white, dotted] (right_neighbour_virt) at (0.5,0.25) {};

    \node[vertex] (right_right_neighbour) at (1.0,0.25) {};
    \node[vertex, draw=none, fill=white, dotted] (right_right_neighbour_virt) at (1.0,-0.55) {};
    
    \draw[->, color=blue!90, thick] (right_neighbour)--(right_neighbour_virt) node[midway, right] {};
    \draw[->, color=blue!90, thick] (right_right_neighbour)--(right_right_neighbour_virt) node[midway, right] {};
    \draw[->, color=blue!90, thick] (left_neighbour)--(left_neighbour_virt) node[midway, left] {};
    \draw[->, color=blue!90, thick] (left_left_neighbour)--(left_left_neighbour_virt) node[midway, left] {};
    \draw[->, color=blue!90, thick] (main)--(main_virt);
    \draw[dashed] (right_right_neighbour)--(right_neighbour)--(main)--(left_neighbour)--(left_left_neighbour);
\end{tikzpicture}
  \end{subfigure}
  \hspace{2em}
  \begin{subfigure}{0.35\textwidth}
    \centering
     \begin{tikzpicture}[main/.style = {draw, circle},scale = 2.0] 
    \tikzstyle{vertex} = [
            circle,
            fill = {
            black! 25
            },
            minimum size=8pt, 
            inner sep=2pt] 
    
    %example
    \node[vertex,fill=blue!40] (main) at (0,0.25) {};
    \node[vertex,draw=none, fill=white] (main_virt) at (0,-0.15) {};

    \node[vertex] (left_neighbour) at (-0.5,-0.25) {};
    \node[vertex, draw=none, fill=white, dotted] (left_neighbour_virt) at (-0.5,0.15) {};

    \node[vertex] (left_left_neighbour) at (-1.0,0.25) {};
    \node[vertex, draw=none, fill=white, dotted] (left_left_neighbour_virt) at (-1.0,-0.15) {};
    
    \node[vertex] (right_neighbour) at (0.5,-0.25) {};
    \node[vertex, draw=none, fill=white, dotted] (right_neighbour_virt) at (0.5,0.15) {};

    \node[vertex] (right_right_neighbour) at (1.0,0.25) {};
    \node[vertex, draw=none, fill=white, dotted] (right_right_neighbour_virt) at (1.0,-0.15) {};
    
    \draw[->, color=blue!90, thick] (right_neighbour)--(right_neighbour_virt) node[midway, right] {};
    \draw[->, color=blue!90, thick] (right_right_neighbour)--(right_right_neighbour_virt) node[midway, right] {};
    \draw[->, color=blue!90, thick] (left_neighbour)--(left_neighbour_virt) node[midway, left] {};
    \draw[->, color=blue!90, thick] (left_left_neighbour)--(left_left_neighbour_virt) node[midway, left] {};
    \draw[->, color=blue!90, thick] (main)--(main_virt);
    \draw[dashed] (right_right_neighbour)--(right_neighbour)--(main)--(left_neighbour)--(left_left_neighbour);
\end{tikzpicture}

\vspace{2em}
\begin{tikzpicture}[main/.style = {draw, circle},scale = 2.0] 
    \tikzstyle{vertex} = [
            circle,
            fill = {
            black! 25
            },
            minimum size=8pt, 
            inner sep=2pt] 
    
    %example
    \node[vertex,fill=blue!40] (main) at (0,-0.35) {};
    \node[vertex,fill=white, draw=none] (main_virt) at (0,0.25) {};

    \node[vertex] (left_neighbour) at (-0.5,0.35) {};
    \node[vertex, draw=none, fill=white, dotted] (left_neighbour_virt) at (-0.5,-0.25) {};
   
    \node[vertex] (left_left_neighbour) at (-1.0,-0.35) {};
    \node[vertex, draw=none, fill=white, dotted] (left_left_neighbour_virt) at (-1.0,0.35) {};

    \node[vertex] (right_neighbour) at (0.5,0.35) {};
    \node[vertex, draw=none, fill=white, dotted] (right_neighbour_virt) at (0.5,-0.25) {};
    
    \node[vertex] (right_right_neighbour) at (1.0,-0.35) {};
    \node[vertex, draw=none, fill=white, dotted] (right_right_neighbour_virt) at (1.0,0.25) {};

    \draw[->, color=blue!90, thick] (right_neighbour)--(right_neighbour_virt) node[midway, right] {};
    \draw[->, color=blue!90, thick] (right_right_neighbour)--(right_right_neighbour_virt) node[midway, right] {};
    \draw[->, color=blue!90, thick] (left_neighbour)--(left_neighbour_virt) node[midway, left] {};
    \draw[->, color=blue!90, thick] (left_left_neighbour)--(left_left_neighbour_virt) node[midway, left] {};
    \draw[->, color=blue!90, thick] (main)--(main_virt);

    \draw[dashed] (right_right_neighbour)--(right_neighbour)--(main)--(left_neighbour)--(left_left_neighbour);
\end{tikzpicture}

\vspace{2em}
\begin{tikzpicture}[main/.style = {draw, circle},scale = 2.0] 
    \tikzstyle{vertex} = [
            circle,
            fill = {
            black! 25
            },
            minimum size=8pt, 
            inner sep=2pt] 
    
    %example
    \node[vertex,fill=blue!40] (main) at (0,0.55) {};
    \node[vertex,fill=white, draw=none] (main_virt) at (0,-0.35) {};

    \node[vertex] (left_neighbour) at (-0.5,-0.55) {};
    \node[vertex, draw=none, fill=white, dotted] (left_neighbour_virt) at (-0.5,0.45) {};

    \node[vertex] (left_left_neighbour) at (-1.0,0.55) {};
    \node[vertex, draw=none, fill=white, dotted] (left_left_neighbour_virt) at (-1.0,-0.35) {};
    
    \node[vertex] (right_neighbour) at (0.5,-0.55) {};
    \node[vertex, draw=none, fill=white, dotted] (right_neighbour_virt) at (0.5,0.45) {};

    \node[vertex] (right_right_neighbour) at (1.0,0.55) {};
    \node[vertex, draw=none, fill=white, dotted] (right_right_neighbour_virt) at (1.0,-0.35) {};

    \draw[->, color=blue!90, thick] (right_neighbour)--(right_neighbour_virt) node[midway, right] {};
    \draw[->, color=blue!90, thick] (right_right_neighbour)--(right_right_neighbour_virt) node[midway, right] {};
    \draw[->, color=blue!90, thick] (left_neighbour)--(left_neighbour_virt) node[midway, left] {};
    \draw[->, color=blue!90, thick] (left_left_neighbour)--(left_left_neighbour_virt) node[midway, left] {};
    \draw[->, color=blue!90, thick] (main)--(main_virt);
    
    \draw[dashed] (right_right_neighbour)--(right_neighbour)--(main)--(left_neighbour)--(left_left_neighbour);
\end{tikzpicture}
  \end{subfigure}
  \caption{On the left, it is shown how skews increase without $\JC$. While $\SC(0)$ disallows that $(v,\ell)$ speeds up its pulse by more than the equivalent of $(v,\ell-1)$ matching the earliest pulse of any $(w,\ell-1)$, $\{v,w\}\in E$, $\FC$ permits that a node $(v,\ell)$ with slow $(v,\ell-1)$ to ``overshoot,'' i.e., $\Cor_{v,\ell}$ (shown as blue arrow) gets large. This results in an amplifying oscillatory behavior. On the right, the same scenario is shown with $\JC$ in effect. $\JC$ forces the corrections to stop $\kappa$ before the earliest or latest neighbor, respectively, resulting in a dampened oscillation.\label{fig:need_for_JC}}
  \Description[Illustrating the need for a Jump Condition]{On the left, it is shown how skews increase without $\JC$. While $\SC(0)$ disallows that $(v,\ell)$ speeds up its pulse by more than the equivalent of $(v,\ell-1)$ matching the earliest pulse of any $(w,\ell-1)$, $\{v,w\}\in E$, $\FC$ permits that a node $(v,\ell)$ with slow $(v,\ell-1)$ to ``overshoot,'' i.e., $\Cor_{v,\ell}$ (shown as blue arrow) gets large. This results in an amplifying oscillatory behavior. On the right, the same scenario is shown with $\JC$ in effect. $\JC$ forces the corrections to stop $\kappa$ before the earliest or latest neighbor, respectively, resulting in a dampened oscillation.}
\end{figure}
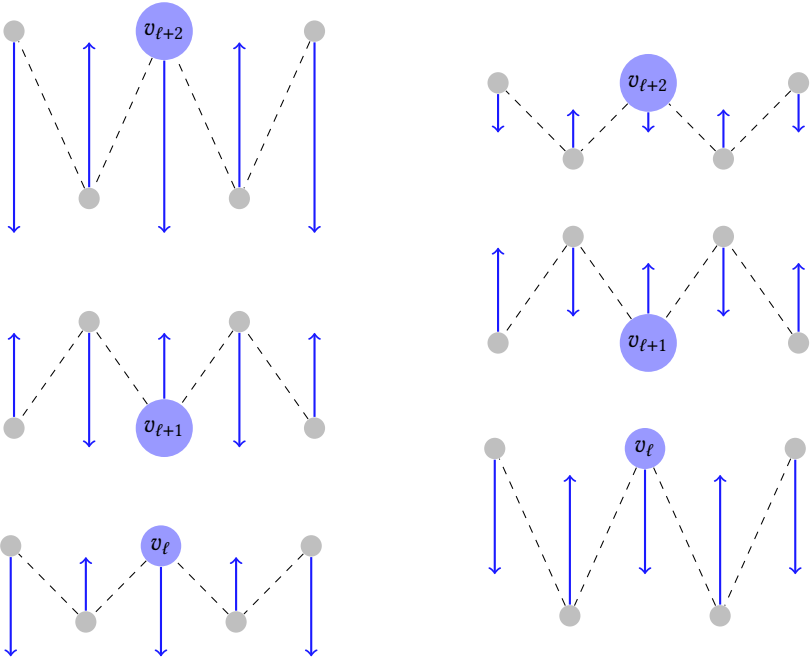

\subsection{Bounding \texorpdfstring{$\Psi^s$}{Psi} in the Absence of Faults}\label{sec:faultfree}

With the conditions established, we are ready to study how $\Psi^s(\ell)$ evolves in the fault-free setting.
The main technical challenge in bounding $\Psi^s$ lies in performing the induction step from $s-1\in \N$ to $s$.
We will argue that for $\Psi^s(\bar{\ell}\,)$ to be large for some $\bar{\ell}$, $\Xi^s(\underline{\ell}\,)$ must have been large for some $\underline{\ell} < \bar{\ell}$, with an additive term growing with $\bar{\ell}-\underline{\ell}$.
\begin{theorem}\label{thm:psi_bound}
For $s\in \N_{>0}$ and layers $\underline{\ell}\le \bar{\ell}$, it holds that
\begin{equation*}
\Psi^s(\bar{\ell}\,)\le \max\set{0,\Xi^s(\underline{\ell}\,)-(\bar{\ell}-\underline{\ell}+1)\kappa}+(\bar{\ell}-\underline{\ell}\,)\cdot\frac{\kappa}{2}.
\end{equation*}
\end{theorem}

\begin{proof}[Proof strategy]
Intuitively, we intend to argue that if $\Psi^s(\bar{\ell}\,)$ is large, so must be $\Xi^s(\underline{\ell}\,)$.
Tracing back the cause for this, we show that in every step, we have that $\Xi^s(\ell-1)$ is larger than $\Xi^s(\ell)$ by at least $\kappa/2$.
Since $\Xi^s(\bar{\ell}\,)\ge \Psi^s(\bar{\ell}\,)$, as $\psi^s_{v,w}(\ell)\ge \xi^s_{v,w}(\ell)$ for all $v$, $w$, $s$, and $\ell$, this yields the claim.
To formalize that $\Xi^s(\ell)$ must have been decreasing steadily, we seek to show that the minimal layer $\ell$ for which there are nodes $v^{\ell},w^{\ell}\in V$ satisfying that $\xi^s_{v^{\ell},w^{\ell}}(\ell)$ is large enough is $\underline{\ell}$.
To this end, we identify nodes $w$ and $v$ -- either $w^{\ell}$ and $v^{\ell}$ themselves or neighbors of them -- which cause the large skew on layer $\ell$ by exhibiting a large skew on layer $\ell-1$.
This is done based on $\SC(s)$ and $\FC(s)$, with $\JC$ kicking in for the special case that $w=v^{\ell}$ and $v=w^{\ell}$.

A key obstacle is that if $w$ is a neighbor of $w^{\ell}$, this results in a larger difference in skew than if $v$ is a neighbor of $v^{\ell}$, namely $4s\kappa$ versus $(4s-2)\kappa$.
Thus, when $w$ is closer to $v^{\ell}$ than $w^{\ell}$, we ``lose'' $2\kappa$ relative to the skew bound on layer $\ell$.
For $d(v^{\bar{\ell}},w^{\bar{\ell}})$ many steps, we can compensate for this based on the initial skew between $v^{\bar{\ell}}$ and $w^{\bar{\ell}}$, but not more.
To address this, essentially we need to show that for any additional steps ``towards'' $v^{\ell}$ there will be a corresponding step ``away'' from $v^{\ell}$, on which we ``gain'' additional $2\kappa$ relative to the skew bound on the layer $\ell$.

If corrections were always positive, this would be straightforward:
Steps towards $v^{\ell}$ would also be steps towards $v^{\bar{\ell}}$, and upon $w^{\ell}=v^{\bar{\ell}}$ we would reach a contradiction to the skew bounds shown.
Unfortunately, negative corrections foreclose this simple argument.
To address this, we introduce a third ``prover'' node $p^{\ell}$, where $p^{\bar{\ell}}=v^{\bar{\ell}}$, which never increases its distance to $w^{\ell}$;
if $p^{\ell}$ performs a negative correction, then $p$ is a neighbor of $p^{\ell}$ that is closer to $w^{\ell}$.
We then can infer that $p\neq w$ from the skew bounds.

A major complication this approach faces is the special case $p=w^{\ell}$ and $w=p^{\ell}$.
Again, $\JC$ kicks in to show that we have sufficiently large skew between $p$ and $w$.
However, now $p$ lies ``behind'' $w$ from the perspective of $v$.
A later reversal of this situation by repeating the case that $p=w^{\ell}$ and $w=p^{\ell}$ results in $w$ being farther away from $v^{\ell}$, yet $d(p,w)=d(p^{\ell},w^{\ell})$.
The proof covers this case by adding an additional $(4s-2)\kappa$ to the skew bound if the above situation occured an odd number of times.

Finally, we seek to avoid the case that $v=p^{\ell}$ and $p=v^{\ell}$ for analogous reasons.
Fortunately, here we can exploit that the skew bound between $v^{\ell}$ and $w^{\ell}$ is stronger than the one between $p^{\ell}$ and $w^{\ell}$, meaning that we can simply choose $p=v$ instead in this situation.
In the proof, we do so whenever $v$ lies on the path connecting $p^{\ell}$ and $w^{\ell}$ that we maintain to keep track of hop counts in the construction.
\end{proof}
\begin{proof}[Proof of \Cref{thm:psi_bound}]
Assume towards a contradiction that the statement of \Cref{thm:psi_bound} is false for minimal $\bar{\ell}$, i.e., there are $v^{\bar{\ell}}$ and $w^{\bar{\ell}}$ such that
\begin{align}
\psi^s_{v^{\bar{\ell}},w^{\bar{\ell}}} &> (\bar{\ell}-\underline{\ell}\ )\cdot\frac{\kappa}{2}\label{eq:contra1}\\
\mbox{and}\quad \psi^s_{v^{\bar{\ell}},w^{\bar{\ell}}}&>\Xi^s(\underline{\ell}\,)-(\bar{\ell}-\underline{\ell}\,)\cdot\frac{\kappa}{2}-\kappa\label{eq:contra2}
\end{align}
and there is no smaller $\bar{\ell}'$ for which this applies for some pair of nodes.

Let $\ell \in [\underline{\ell},\bar{\ell}]$ be minimal such that are $v^{\ell},p^{\ell},w^{\ell}\in V$, a path $Q^{\ell}$ in $H$ from $p^{\ell}$ to $v^{\ell}$, and a path $P^{\ell}$ in $H$ from $p^{\ell}$ to $w^{\ell}$ with the following properties:
\begin{enumerate}
  \item[\Pone] $w^{\ell}\neq p^{\ell}$.
  \item[\Ptwo] $w^{\ell}\neq v^{\ell}$.
  \item[\Pthree]
  $t_{p^{\ell},\ell}-t_{w^{\ell},\ell}-4s\kappa |P^{\ell}|\ge \psi_{v^{\bar{\ell}},w^{\bar{\ell}}}^s(\bar{\ell}\,)-(\bar{\ell}-\ell)\cdot\frac{\kappa}{2}>0$.
  \item[\Pfour] Denote by $|P^{\ell}|$ and $|Q^{\ell}|$ the length of $P^{\ell}$ and $Q^{\ell}$, respectively.
  With the shorthand
  \begin{align*}
  \Delta^{\ell}:=\begin{cases}
  |P^{\ell}|+|Q^{\ell}| - 1 & \mbox{if } P^{\ell} \mbox{ and } Q^{\ell} \mbox{ have the same first edge}\\
  |P^{\ell}|+|Q^{\ell}| & \mbox{else,}
  \end{cases}
  \end{align*}
  it holds that
  \begin{align*}
  t_{v^{\ell},\ell}-t_{w^{\ell},\ell}-(4s-2)\kappa \Delta^{\ell}\ge \psi_{v^{\bar{\ell}},w^{\bar{\ell}}}^s(\bar{\ell}\,)+(\bar{\ell}-\ell)\cdot\frac{\kappa}{2}+2\kappa |P^{\ell}|.
  \end{align*}
  \item[\Pfive] If $v^{\ell}\in P^{\ell}$, then $p^{\ell}=v^{\ell}$.
\end{enumerate}
To see that such an index must indeed exist, let
\begin{itemize}
  \item $p^{\bar{\ell}}:=v^{\bar{\ell}}$,
  \item $P^{\bar{\ell}}$ be a shortest path in $H$ from $p^{\bar{\ell}}$ to $w^{\bar{\ell}}$, and
  \item $Q^{\bar{\ell}}:=(p^{\bar{\ell}})=(v^{\bar{\ell}})$, i.e., the $0$-length path from $p^{\bar{\ell}}$ to $v^{\bar{\ell}}$.
\end{itemize}
This choice satisfies 
\begin{itemize}
  \item \Pone and \Ptwo, because $\Psi^s_{v^{\bar{\ell}},w^{\bar{\ell}}}(\bar{\ell}\,)\neq 0$ implies that $v^{\bar{\ell}}\neq w^{\bar{\ell}}$;
  \item \Pfour, because
  \begin{equation*}
  t_{v^{\bar{\ell}},\bar{\ell}}-t_{w^{\bar{\ell}},\bar{\ell}}-(4s-2)\kappa \Delta^{\ell}=t_{v^{\bar{\ell}},\bar{\ell}}-t_{w^{\bar{\ell}},\bar{\ell}}-(4s-2)\kappa |P^{\bar{\ell}}|
  = \psi_{v^{\bar{\ell}},w^{\bar{\ell}}}+2\kappa |P^{\bar{\ell}}|;~\mbox{and}
  \end{equation*}
  \item \Pthree and \Pfive, because $p^{\bar{\ell}}=v^{\bar{\ell}}$ (i.e., $t_{p^{\bar{\ell}},\bar{\ell}}=t_{v^{\bar{\ell}},\bar{\ell}}$ and $\Delta^{\bar{\ell}}=|P^{\bar{\ell}}|$) and \Pfour holds.
\end{itemize}

\Cref{cor:invariants} proves that in fact $\ell=\underline{\ell}$.
Note that 
\begin{align*}
d(v^{\underline{\ell}},w^{\underline{\ell}})&\le \begin{cases}
|P^{\underline{\ell}}|+|Q^{\underline{\ell}}|-2 & \mbox{if $P^{\underline{\ell}}$ and $Q^{\underline{\ell}}$ share the first edge}\\
|P^{\underline{\ell}}|+|Q^{\underline{\ell}}| & \mbox{else}
\end{cases}\\
&\le \Delta^{\underline{\ell}}
\end{align*}
and that $|P^{\underline{\ell}}|\ge 1$ due to \Pone.
Therefore, \Pfour yields that
\begin{align*}
\Xi^s(\underline{\ell}\,)
&\ge t_{v^{\underline{\ell}},\underline{\ell}}-t_{w^{\underline{\ell}},\underline{\ell}}-(4s-2)\kappa d(v^{\underline{\ell}},w^{\underline{\ell}})\\
&\ge t_{v^{\underline{\ell}},\underline{\ell}}-t_{w^{\underline{\ell}},\underline{\ell}}-(4s-2)\kappa\Delta^{\underline{\ell}}\\
&\ge \psi_{v^{\bar{\ell}},w^{\bar{\ell}}}^s(\bar{\ell}\,)+(\bar{\ell}-\underline{\ell})\cdot\frac{\kappa}{2}+2\kappa |P^{\underline{\ell}}|\\
&\ge \psi_{v^{\bar{\ell}},w^{\bar{\ell}}}^s(\bar{\ell}\,)+(\bar{\ell}-\underline{\ell})\cdot\frac{\kappa}{2}+2\kappa,
\end{align*}
contradicting \Cref{eq:contra2} and completing the proof.
\end{proof}

The remainder of \Cref{sec:faultfree} is dedicated to proving \Cref{cor:invariants}, which is the missing step in the proof of \Cref{thm:psi_bound}. 
To this end, until the end of \Cref{sec:faultfree} we consider the setting of the proof of \Cref{thm:psi_bound} and assume for contradiction that $\ell>\underline{\ell}$.
We take note of some straightforward implications.
\begin{observation}\label{obs:implications}
  For any fixed index $\ell$, we have the following implications:
  \begin{itemize}
    \item \Pthree $\Rightarrow$ \Pone
    \item \Pfour $\Rightarrow$ \Ptwo
    \item ($v^{\ell}=p^{\ell} \land$ \Pfour) $\Rightarrow$ \Pthree. 
  \end{itemize}
  Moreover,
\begin{equation*}
\psi^s_{v^\ell,w^{\ell}}(\ell)-(\bar{\ell}-\ell)\cdot\frac{\kappa}{2}>0.
\end{equation*}
\end{observation}
\begin{proof}
  We prove each implication separately.
  \begin{itemize}
    \item From \Pthree, $t_{p^\ell,\ell}- t_{w^{\ell},\ell} > 4s\kappa |P^\ell| \geq 0$. This implies $t_{p^\ell,\ell}> t_{w^{\ell},\ell}$ and hence $w^\ell \neq p^\ell$, i.e., \Pone.
    \item Note that $\Delta^{\ell}\ge 0$, $|P^{\ell}|\ge 0$, and $4s-2>0$. Hence, (P4) and \Cref{eq:contra1} imply that
    \begin{equation*}
       t_{v^\ell,\ell}-t_{w^{\ell},\ell}\geq \psi^s_{v^{\ell},w^{\ell}}(\bar{\ell}\,) >0.
    \end{equation*}
	It follows that $w^{\ell}\neq v^{\ell}$, i.e., (P2).
    \item If $v^\ell = p^\ell$, then $t_{v^{\ell},\ell}=t_{p^{\ell},\ell}$, $|Q^\ell| = 0$, and $\Delta^{\ell}=|P^{\ell}|$. Thus, \Pfour implies that
    \begin{align*}
      t_{p^{\ell},\ell}-t_{w^{\ell},\ell} -(4s-2)\kappa|P^{\ell}| &\geq \psi^s_{v^{\ell},w^{\ell}}(\bar{l}) + (\bar{\ell} - \ell)\cdot\frac{\kappa}{2} + 2\kappa|P^{\ell}|\\
      &\ge \psi^s_{v^{\ell},w^{\ell}}(\bar{l}) - (\bar{\ell} - \ell)\cdot\frac{\kappa}{2} + 2\kappa|P^{\ell}|,
    \end{align*}
    which can be rearranged to yield \Pthree.\qedhere
  \end{itemize}
\end{proof}

\subsubsection*{A Step in the Construction}
We now identify nodes that are suitable for taking the role of $v^{\ell}$, $p^{\ell}$, and $w^{\ell}$ on layer $\ell-1$.
These are either the nodes themselves or neighbors of them in $H$, where $\FC(s)$, $\SC(s)$, and $\JC$ serve to relate respective pulse times.

\begin{lemma}\label{lem:v_step}
There is a node $v\in V$ such that
\begin{equation*}
t_{v^{\ell},\ell-1}-\Cor_{v^{\ell},\ell}\le t_{v,\ell-1}-(4s-2)\kappa\Delta_v-\kappa,
\end{equation*}
where
\begin{equation*}
\Delta_v = \begin{cases}
0 & \mbox{and }v=v^{\ell},\\
-1 & \mbox{and $\{v,v^{\ell}\}$ is the last edge of $Q^{\ell}$ or the first edge of $P^{\ell}$, or}\\
1 & \mbox{and $\{v^{\ell},v\}\in E$.}
\end{cases}
\end{equation*}
\end{lemma}
\begin{proof}
By \Cref{lem:fast_holds}, $v^\ell$ obeys the fast condition. Thus one of three things is true for $v^\ell$.
\begin{itemize}
  \item $\FCone(s)$ holds. In this case, let $v=\arg\max_{\set{x,v^\ell}\in E}\{t_{x,\ell-1}\}$ and bound
    \begin{equation*}
      t_{v^{\ell},\ell-1} - \Cor_{v^{\ell},\ell} \leq  \max_{\{x,v^\ell\}\in E}\set{t_{x,\ell-1}} - (4s-2)\kappa - \kappa= t_{v,\ell-1} - (4s-2)\kappa - \kappa,
    \end{equation*}
    i.e., the claim of the lemma holds with $\Delta_v=1$.
  \item $\FCtwo(s)$ holds. In this case, let $\{v,v^{\ell}\}$ be the last edge of $Q^{\ell}$ if $|Q^{\ell}|\neq 0$ or the first edge of $P^{\ell}$ otherwise; the latter is feasible, because then $v^{\ell}=p^{\ell}$, and $|P^{\ell}|\neq 0$ due to (P1). We get that
    \begin{equation*}
      t_{v^{\ell},\ell-1} -  \Cor_{v^{\ell},\ell} \leq \min_{\{x,v^\ell\}\in E}\set{t_{x,\ell-1}} + (4s-2)\kappa - \kappa \leq t_{v,\ell-1} + (4s-2)\kappa - \kappa.
    \end{equation*}
    Thus, the claim of the lemma holds with $\Delta_v = -1$.
  \item $\FCthree$ holds. In this case,
    \begin{equation*}
      t_{v^{\ell},\ell-1} - \Cor_{v^{\ell},\ell} \leq t_{v^{\ell},\ell-1} - \kappa,
    \end{equation*}
    i.e., the claim of the lemma holds with $\Delta_v = 0$.\qedhere
\end{itemize}
\end{proof}

\begin{lemma}\label{lem:w_step}
There is a node $w\in V$ such that
\begin{equation*}
t_{w^{\ell},\ell-1}-\frac{\Cor_{w^{\ell},\ell}}{\vartheta}\ge t_{w,\ell-1}+4s\kappa\Delta_w,
\end{equation*}
where
\begin{equation*}
\Delta_w= \begin{cases}
0 & \mbox{and }w=w^{\ell},\\
-1 & \mbox{and $\{w,w^{\ell}\}$ is the last edge of $P^{\ell}$, or}\\
1 & \mbox{and $\{w^{\ell},w\}\in E$.}
\end{cases}
\end{equation*}
\end{lemma}
\begin{proof}
By \Cref{lem:slow_holds}, $w^{\ell}$ satisfies $\SC$.
We make a case distinction based on which one of $\SCone$, $\SCtwo$, and $\SCthree$ applies.
 \begin{itemize}
   \item $\SCone(s)$ holds. Let $\{w,w^{\ell}\}$ be the last edge of $P^\ell$; by (P1), $|P^{\ell}|\neq 0$, i.e., this edge exists. Then
   \begin{equation*}
    t_{w^\ell,\ell-1} - \frac{\Cor_{w^{\ell},l}}{\vartheta}\geq  \max_{\{x,w^{\ell}\}\in E}\{t_{x,\ell-1}\} - 4s\kappa
    \geq t_{w,\ell-1} - 4s\kappa,
   \end{equation*}
   i.e., the claim of the lemma holds with $\Delta_w=-1$.
   \item $\SCtwo(s)$ holds. In this case, let $w=\arg\min_{{\set{x,v^\ell}\in E}}\{t_{x,\ell-1}\}$ and bound
   \begin{align*}
    t_{w^\ell,\ell-1} - \frac{\Cor{w^{\ell},\ell}}{\vartheta}\geq  \min_{\set{x,w^{\ell}}\in E}\{t_{x,\ell-1}\} + 4s\kappa = t_{w,\ell-1} + 4s\kappa.
   \end{align*}
   Thus, the lemma holds with $\Delta_w = 1$.
   \item $\SCthree$ holds. Then
   \begin{equation*}
     t_{w^\ell,\ell-1} - \Cor_{w,\ell} \geq t_{w^\ell,\ell-1},
   \end{equation*}
   i.e., the claim of the lemma holds with $\Delta_w=0$.\qedhere
  \end{itemize}
\end{proof}

\begin{lemma}\label{lem:p_step}
There is a node $p\in V$ such that
\begin{align*}
t_{p^{\ell},\ell-1}-\Cor_{p^{\ell},\ell}\le \begin{cases}
t_{p,\ell-1}& \mbox{and $p=p^{\ell}$, or}\\
t_{p,\ell-1}-\kappa & \mbox{and $\{p^{\ell},p\}$ is the first edge of $P^{\ell}$.}
\end{cases}
\end{align*}
\end{lemma}
\begin{proof}
If $\Cor_{p^{\ell},\ell}\ge 0$, the claim holds with $p=p^{\ell}$.
Hence, suppose that $\Cor_{p^{\ell},\ell}< 0$.
Let $\{p^{\ell},p\}$ be the first edge of $P^\ell$; such an edge exists, as by (P1) we have that $p^{\ell}\neq w^{\ell}$ and hence $|P^{\ell}|\neq 0$.
By \Cref{lem:jump_holds}, $p^\ell$ satisfies $\JC$. 
As $\Cor_{p^{\ell},\ell}< 0$, $\JCtwo$ must apply.
We conclude that
\begin{equation*}
\Cor_{p^\ell,\ell} \geq t_{p^\ell,\ell-1} - \min_{\set{x,p^\ell}\in E}\set{t_{x,\ell-1}} + \kappa
\ge t_{p^\ell,\ell-1} - t_{p,\ell-1} + \kappa.
\end{equation*}
Rearranging terms, the desired inequality follows.
\end{proof}

In the following, let $(v,p,w)$ be the triple of nodes guaranteed by \Cref{lem:w_step,lem:p_step,lem:v_step}.
Denote by $\circ$ concatenation of paths, by $\prefix(R,x)$ the prefix of path $R$ ending at node $x\in R$, and by $\suffix(R,x)$ the suffix of path $R$ starting at node $x\in R$.
Let
\begin{align*}
p'&=\begin{cases}
v & \mbox{if $v$ lies on $\suffix(P^{\ell},p)$},\\
p & \mbox{else,}
\end{cases}\\
P&:=\begin{cases}
\prefix(P^{\ell},w) & \mbox{if $w$ lies on $P^{\ell}$},\\
P^{\ell}\circ (w^{\ell},w) & \mbox{else,}
\end{cases}\\
P'&:=\begin{cases}
\suffix(P,p') &\mbox{if $p'$ lies on $P$},\\
(p',w) & \mbox{else,}
\end{cases}\\
Q&:=\begin{cases}
\prefix(Q^{\ell},v) & \mbox{if $v$ lies on $Q^{\ell}$},\\
Q^{\ell}\circ \{v^{\ell},v\} & \mbox{else,}
\end{cases}\\
Q'&:=\begin{cases}
\suffix(Q,p') & \mbox{if $p'$ lies on $Q$},\\
(p',p^{\ell})\circ Q & \mbox{else.}
\end{cases}
\end{align*}
For notational convenience, in analogy to $\Delta^{\ell}$ we also define
\begin{equation*}
\Delta:=\begin{cases}
|P'|+|Q'| - 1 & \mbox{if } P' \mbox{ and } Q' \mbox{ have the same first edge}\\
|P'|+|Q'| & \mbox{else.}
\end{cases}
\end{equation*}
We will show that this construction satisfies properties (P1) to (P5) for layer $\ell-1$ with $v^{\ell-1}=v$, $p^{\ell-1}=p'$, $w^{\ell-1}=w$, $P^{\ell-1}=P'$, and $Q^{\ell-1}=Q'$;
this will constitute the desired contradiction.

However, we first point out that indeed $P'$ and $Q'$ are paths in $H$ from $p'$ to $w$ and $v$, respectively.
To this end, we first cover the special case that $p'$ does not lie on $P$.
\begin{observation}\label{obs:non_suffix}
If $p'$ does not lie on $P$, then $p'=w^{\ell}$ and either $w=p^{\ell}$ or $p'=v$.
\end{observation}
\begin{proof}
By \Cref{lem:p_step}, $p$ lies on the first edge of $P^{\ell}$.
Hence, if $p'=p$, $p'$ lies on $P$ unless $\prefix(P^{\ell},w)$ does not contain this edge.
By \Cref{lem:w_step}, this can only happen if the first edge of $P^{\ell}$ is also the last edge, i.e., $P^{\ell}=(p^{\ell},w^{\ell})=(w,p')$.

It remains to consider the case that $p'\neq p$, i.e., $p'=v$.
Again, we use that all edges but the last of $P^{\ell}$ are also contained in $P$ by \Cref{lem:w_step}.
Thus, $p'=v=w^{\ell}$.
\end{proof}
\begin{observation}\label{obs:paths}
$P'$ is a path in $H$ from $p'$ to $w$ and $Q'$ is a path in $H$ from $p'$ to $v$.
\end{observation}
\begin{proof}
To show that $P'$ is a path from $p'$ to $w$, note that by \Cref{lem:w_step}, $P$ is a path in $H$, which by definition ends at $w$.
Thus, if $P'=\suffix(P,p')$, $P'$ is a path from $p'$ to $w$ in $H$.
Otherwise, by \Cref{obs:non_suffix}, $p'=w^{\ell}$, and $\{p',w\}=\{w^{\ell},w\}\in E$ by \Cref{lem:w_step}.

To show that $Q'$ is a path from $p'$ to $v$, note that by \Cref{lem:v_step}, $Q$ is a path in $H$, which by definition ends at $v$.
If $p'=p$, by \Cref{lem:p_step} $Q'$ is also a path in $H$, which by definition begins at $p'$ and has the same endpoint as $Q$, which is $v$.
On the other hand, if $p'=v$, $\suffix(Q,p')=\suffix(Q,v)=(v)$, which is the $0$-length path from $p'=v$ to itself.
\end{proof}

\subsubsection*{Proving the Properties}
To prove \Cref{cor:invariants}, we establish that the tuple $(v,p',w,P',Q')$ satisfies properties \Pone to \Pfive for layer $\ell-1$, contradicting the minimality of $\ell$.
By \Cref{obs:paths}, indeed $P'$ and $Q'$ are paths from $v$ to $w$ and $p'$, respectively.
In the following, we will repeatedly use this fact and the property that $\{x^{\ell},x\}\in E$ for $x\in \{v,w,p\}$ whenever $x\neq x^{\ell}$, without explicitly invoking \Cref{obs:paths} and \Cref{lem:v_step,lem:w_step,lem:p_step}.

We first rule out the special case that $v=w^{\ell}$ and $w=v^{\ell}$.
\begin{lemma}\label{lem:vnotw}
The case that $v=w^{\ell}$ and $w=v^{\ell}$ is not possible.
\end{lemma}
\begin{proof}
Assume towards a contradiction that $v=w^{\ell}$ and $w=v^{\ell}$.
We use \Pfour, \Cref{lem:drift}, and \Cref{lem:w_step} to bound
\begin{align*}
-\Cor_{w,\ell} &\ge t_{w,\ell}-t_{w,\ell-1}-\Lambda\\
&= t_{v^{\ell},\ell}-(t_{w,\ell-1}-4s\kappa)-\Lambda-4s\kappa\\
& \ge t_{v^{\ell},\ell}-\left(t_{w^{\ell},\ell-1}-\frac{\Cor_{w^{\ell},\ell}}{\vartheta}\right)-\Lambda-4s\kappa\\
& = t_{v^{\ell},\ell}-\left(t_{w^{\ell},\ell-1}+d-u+\frac{\Lambda-d-\Cor_{w^{\ell},\ell}}{\vartheta}\right)-\frac{\kappa}{2}-4s\kappa\\
& \ge t_{v^{\ell},\ell}-t_{w^{\ell},\ell}-4s\kappa-\frac{\kappa}{2}\\
& \ge \psi_{v^{\bar{\ell}},w^{\bar{\ell}}}(\bar{\ell}\,)-(\bar{\ell}-\ell+1)\frac{\kappa}{2}\\
& > 0.
\end{align*}
Thus, by $\JC$, it holds that
\begin{equation*}
t_{w,\ell-1} \le t_{w^{\ell},\ell-1}+\Cor_{w,\ell}-\kappa.
\end{equation*}
Note that by \Pone, $|P^{\ell}|\neq 0$ and hence $|P^{\ell}|,\Delta^{\ell}\ge 1$.
Thus, by \Pfour and \Cref{eq:contra1}
\begin{equation*}
t_{v^{\ell},\ell}-t_{w^{\ell},\ell}-4s\kappa\ge \psi_{v^{\bar{\ell}},w^{\bar{\ell}}}^s(\bar{\ell}\,)+(\bar{\ell}-\ell)\frac{\kappa}{2}\\
\ge \psi_{v^{\bar{\ell}},w^{\bar{\ell}}}^s(\bar{\ell}\,)>0.
\end{equation*}

We distinguish two cases.
\begin{itemize}
  \item $\Cor_{w^{\ell},\ell}\le \vartheta\kappa$. Then by \Cref{lem:drift}
  \begin{align*}
  4s\kappa&<t_{v^{\ell},\ell}-t_{w^{\ell},\ell}\\
  &= t_{w,\ell}-t_{w^{\ell},\ell}\\
  &\le t_{w,\ell-1}-\Cor_{w,\ell}-\left(t_{w^{\ell},\ell-1}-\frac{\Cor_{w^{\ell},\ell}}{\vartheta}\right)+u+\left(1-\frac{1}{\vartheta}\right)(\Lambda-d)\\
  &\le u+\left(1-\frac{1}{\vartheta}\right)(\Lambda-d)\\
  &<\kappa,
  \end{align*}
  which is a contradiction, because $s\ge 1$.
  \item $\Cor_{w^{\ell},\ell}>\vartheta \kappa$. By $\JC$, it follows that
  \begin{align*}
  t_{w^{\ell},\ell-1} \ge t_{w,\ell-1}+\frac{\Cor_{w^{\ell},\ell}}{\vartheta}+\kappa,
  \end{align*}
  yielding by \Cref{lem:drift} that
  \begin{align*}
  t_{v,\ell-1}-t_{w,\ell-1}&= t_{w^{\ell},\ell-1}-t_{w,\ell-1}\\
  &\ge t_{w,\ell-1} + \frac{\Cor_{w^{\ell},\ell}}{\vartheta}+\kappa - (t_{w^{\ell},\ell-1}+\Cor_{w,\ell}-\kappa)\\
  &= t_{v^{\ell},\ell-1} - \Cor_{v^{\ell},\ell} - \left(t_{w^{\ell},\ell-1}-\frac{\Cor_{w^{\ell},\ell}}{\vartheta}\right)+2\kappa\\
  &\ge t_{v^{\ell},\ell}-t_{w^{\ell},\ell}+2\kappa-\frac{\kappa}{2}\\
  &>t_{v^{\ell},\ell}-t_{w^{\ell},\ell}+\frac{\kappa}{2}.
  \end{align*}
  Recall that by \Pone, $|P^{\ell}|\neq 0$ and hence $|P^{\ell}|,\Delta^{\ell}\ge 1$.
  Moreover, $d(v,w)=d(w^{\ell},w)\le 1$, since by \Cref{lem:w_step} $w$ is either $w^{\ell}$ or a neighbor of $w^{\ell}$.
  Therefore, \Pfour implies that
  \begin{align*}
  \psi_{v,w}^s(\ell-1)
  &= t_{v,\ell-1}-t_{w,\ell-1}- 4s\kappa d(v,w)\\
  &> t_{v^{\ell},\ell}-t_{w^{\ell},\ell} -4s\kappa|P^{\ell}|+\frac{\kappa}{2}\\
  &\ge \psi_{v^{\bar{\ell}},w^{\bar{\ell}}}^s(\bar{\ell}\,)+(\bar{\ell}-(\ell-1))\frac{\kappa}{2}.
  \end{align*}
  Thus, $v$ and $w$ satisfy \Cref{eq:contra1} and \Cref{eq:contra2} with index $\bar{\ell}$ replaced by index $\ell-1<\bar{\ell}$, contradicting the minimality of $\bar{\ell}$.\qedhere
\end{itemize}
\end{proof}
Next, we prove a helper lemma relating $t_{w^{\ell},\ell}$ and $t_{w,\ell-1}$ by a stronger bound than \Cref{lem:w_step} for the special case that $p'=w^{\ell}$ and $w=p^{\ell}$.
This follows similar reasoning as the previous lemma.
However, it does not yield an immediate contradiction, as we need to rely on the weaker bound provided by \Pthree.
\begin{lemma}\label{lem:swap_jump}
If $p'=w^{\ell}$ and $w=p^{\ell}$, then 
\begin{equation*}
t_{w^{\ell},\ell}-t_{w,\ell-1}> d-u+\frac{\Lambda-d}{\vartheta}.
\end{equation*}
\end{lemma}
\begin{proof}
We use \Pthree and \Cref{lem:drift} to bound
\begin{align*}
-\Cor_{w,\ell} &\ge t_{w,\ell}-t_{w,\ell-1}-\Lambda\\
&= t_{p^{\ell},\ell}-(t_{w,\ell-1}-4s\kappa)-\Lambda-4s\kappa\\
& \ge t_{p^{\ell},\ell}-\left(t_{w^{\ell},\ell-1}-\frac{\Cor_{w^{\ell},\ell}}{\vartheta}\right)-\Lambda-4s\kappa\\
& = t_{p^{\ell},\ell}-\left(t_{w^{\ell},\ell-1}+d-u+\frac{\Lambda-d-\Cor_{w^{\ell},\ell}}{\vartheta}\right)-\frac{\kappa}{2}-4s\kappa\\
& \ge t_{p^{\ell},\ell}-t_{w^{\ell},\ell}-4s\kappa-\frac{\kappa}{2}\\
& \ge \psi_{v^{\bar{\ell}},w^{\bar{\ell}}}(\bar{\ell}\,)-(\bar{\ell}-\ell+1)\frac{\kappa}{2}\\
& > 0.
\end{align*}
Thus, by $\JC$, it holds that
\begin{equation*}
t_{w,\ell-1} \le t_{w^{\ell},\ell-1}-\kappa.
\end{equation*}
We distinguish two cases.
\begin{itemize}
  \item $\Cor_{w^{\ell},\ell}\le \vartheta\kappa$. Then
  \begin{align*}
  t_{w^{\ell},\ell}-t_{w,\ell-1}&\ge t_{w^{\ell},\ell}-t_{w^{\ell},\ell-1}+\kappa\\
  &\ge d-u+\frac{\Lambda-d-\Cor_{w^{\ell},\ell}}{\vartheta}+\kappa\\
  &\ge d-u+\frac{\Lambda-d}{\vartheta}.
  \end{align*}
  \item $\Cor_{w^{\ell},\ell}>\vartheta \kappa$. By $\JC$, it follows that
  \begin{align*}
  t_{w^{\ell},\ell-1} \ge t_{w,\ell-1}+\frac{\Cor_{w^{\ell},\ell}}{\vartheta}+\kappa,
  \end{align*}
  yielding that
  \begin{align*}
  t_{w^{\ell},\ell}-t_{w,\ell-1}&\ge t_{w^{\ell},\ell}-t_{w^{\ell},\ell-1}+\frac{\Cor_{w^{\ell},\ell}}{\vartheta}+\kappa\\
  &> d-u+\frac{\Lambda-d}{\vartheta}.\qedhere
  \end{align*}
\end{itemize}
\end{proof}
Using \Cref{lem:swap_jump}, we establish \Pfour for the special case of $p'=w^{\ell}$ and $w=p^{\ell}$.
Note that this entails that $w$ is closer to $p^{\ell}$, yet $P$ is not shorter than $P^{\ell}$.
This is accounted for by the case distinction in the definition of $\Delta^{\ell}$, which covers the difference.
\begin{lemma}\label{lem:swap}
If $p'=w^{\ell}$ and $w=p^{\ell}$, then \Pfour holds for $v$, $p'$, $w$, $|P'|$, $|Q'|$, and layer $\ell-1$.
\end{lemma}
\begin{proof}
Denote by $\Delta_v\in \{-1,0,1\}$ the value such that
\begin{align*}
t_{v^{\ell},\ell-1}-\Cor_{v^{\ell},\ell}&\le t_{v,\ell-1}-(4s-2)\kappa\Delta_v-\kappa
\end{align*}
according to \Cref{lem:v_step}.
By \Cref{lem:drift,lem:swap_jump},
\begin{align*}
&\,t_{v,\ell-1}-t_{w,\ell-1}\\
>&\, t_{v,\ell-1}-t_{w^{\ell},\ell}+d-u+\frac{\Lambda-d}{\vartheta}\\
\ge &\,t_{v^{\ell},\ell-1}-\Cor_{v^{\ell},\ell}+(4s-2)\kappa \Delta_v+\kappa-t_{w^{\ell},\ell}+d-u+\frac{\Lambda-d}{\vartheta}\\
\ge &\,t_{v^{\ell},\ell}-\Lambda+(4s-2)\kappa \Delta_v+\kappa-t_{w^{\ell},\ell}+d-u+\frac{\Lambda-d}{\vartheta}\\
= &\,t_{v^{\ell},\ell}-t_{w^{\ell},\ell}+(4s-2)\kappa\Delta_v+\frac{\kappa}{2}\\
\ge &\,(4s-2)\kappa(\Delta^{\ell}+\Delta_v)+\psi_{v^{\bar{\ell}},w^{\bar{\ell}}}^s(\bar{\ell}\,)+(\bar{\ell}-(\ell-1))\frac{\kappa}{2}+\kappa_s|P^{\ell}|.
\end{align*}

We claim that $\Delta\le \Delta^{\ell}+\Delta_v$.
Note that plugging this into the above inequality yields
\begin{equation*}
t_{v,\ell-1}-t_{w,\ell-1}\ge (4s-2)\kappa \Delta+\psi_{v^{\bar{\ell}},w^{\bar{\ell}}}^s(\bar{\ell}\,)+(\bar{\ell}-(\ell-1))\frac{\kappa}{2}+\kappa_s|P'|,
\end{equation*}
i.e., \Pfour for $v$, $p'$, $w$, $|P'|$, $|Q'|$, and layer $\ell-1$, as desired.
Therefore, proving the above claim will complete the proof.

To show the claim, we first note that $P'=(p',w)=(w^{\ell},p^{\ell})$.
Since $w=p^{\ell}$, by \Cref{lem:w_step} we also have that $P^{\ell}=(p^{\ell},w^{\ell})=(w,p')$.
In particular, $|P^{\ell}|=|P'|$.
We distinguish two cases.
\begin{itemize}
  \item $P^{\ell}$ and $Q^{\ell}$ share the first edge. It follows that $|Q^{\ell}|\ge 2$, as otherwise $v^{\ell}=w^{\ell}$, contradicting \Ptwo. If $v=p'$, then
  \begin{equation*}
  |Q'|=|Q|=|(v)|=0\le |Q^{\ell}|-2\le |Q^{\ell}|+\Delta_v-1.
  \end{equation*}
  Otherwise, the first edge of $Q$ is the first edge of $Q^{\ell}$ and thus $P^{\ell}$. This edge is $\{p^{\ell},w^{\ell}\}=\{p^{\ell},p'\}$. Hence, $|Q'|=|\suffix(Q,p')|\le |Q|-1=|Q^{\ell}|+\Delta_v-1$. Either way, we get that
  \begin{equation*}
  \Delta\le |P'|+|Q'| \le |P^{\ell}|+|Q^{\ell}|+\Delta_v-1 = \Delta^{\ell}+\Delta_v.
  \end{equation*}
  \item $P^{\ell}$ and $Q^{\ell}$ do not share the first edge, but $P'$ and $Q'$ do. Then
  \begin{equation*}
  \Delta= |P'|+|Q'|-1 \le |P^{\ell}|+|Q^{\ell}|+\Delta_v-1 = \Delta^{\ell}+\Delta_v.
  \end{equation*}
  \item $P^{\ell}$ and $Q^{\ell}$ do not share the first edge and neither do $P'$ and $Q'$. As the first (and only) edge of $P'$ is $\{p',w\}=\{p',p^{\ell}\}$, this entails that $Q'=\suffix(Q,p')$. We distinguish two subcases.
  \begin{itemize}
    \item $|\suffix(Q,p')|\le |Q|-1$. Then
    \begin{equation*}
\Delta= |P'|+|Q'| \le |P^{\ell}|+|Q|-1\le |P^{\ell}|+|Q^{\ell}|+\Delta_v = \Delta^{\ell}+\Delta_v.
    \end{equation*}
    \item $|\suffix(Q,p')|=|Q|$ and $v^{\ell}\neq w$. Then $p'$ is the last node on $Q$, i.e., $v=p'$. As by \Cref{obs:paths} $Q'$ is a path from $p'$ to $v$, it follows that $|Q'|=0<|Q^{\ell}|$. We conclude that
    \begin{equation*}
    \Delta= |P'|+|Q'| \le |P^{\ell}|+|Q^{\ell}|+\Delta_v = \Delta^{\ell}+\Delta_v.
    \end{equation*}
    \item $|\suffix(Q,p')|=|Q|$ and $v^{\ell}= w$. As $w=p^{\ell}$ and $p'=v=w^{\ell}$ as in the previous subcase, this contradicts \Cref{lem:vnotw}.\qedhere
  \end{itemize}
\end{itemize}
\end{proof}

Before proceeding to the case that $v\neq w^{\ell}$ or $w\neq v^{\ell}$, we prove another helper statement ruling out the specific case that $v\neq p'=w$.
\begin{lemma}\label{lem:pnotw}
It is not possible that $v\neq p'=w$.
\end{lemma}
\begin{proof}
Assume towards a contradiction that $v\neq p'=w$.
Thus, $p'=p$. \Cref{lem:w_step,lem:p_step} yield that
\begin{align*}
t_{w^{\ell},\ell-1}-\frac{\Cor_{w^{\ell},\ell}}{\vartheta}&\ge t_{w,\ell-1}-4s\kappa\mbox{ and}\\
t_{p^{\ell},\ell-1}-\Cor_{p^{\ell},\ell}&\le t_{p',\ell-1}.
\end{align*}
Using \Pthree and \Cref{lem:drift}, it follows that
\begin{align*}
0&=t_{p',\ell-1}-t_{w,\ell-1}\\
&\ge t_{p^{\ell},\ell-1}-\Cor_{p^{\ell},\ell}-\left(t_{w^{\ell},\ell-1}-\frac{\Cor_{w^{\ell},\ell}}{\vartheta}\right)-4s\kappa\\
&\ge t_{p^{\ell},\ell}-t_{w^{\ell},\ell}-4s\kappa-\frac{\kappa}{2}\\
&\ge \psi^s_{v^{\bar{\ell},w^{\bar{\ell}}}}(\bar{\ell}\,)-(\bar{\ell}-(\ell-1))\frac{\kappa}{2}\\
&>0,
\end{align*}
arriving at the desired contradiction.
\end{proof}
We now establish \Pfour for the case that $v\neq w^{\ell}$ or $w\neq v^{\ell}$.
\begin{lemma}\label{lem:P4}
If $p'\neq w^{\ell}$ or $w\neq p^{\ell}$, then \Pfour holds for $v$, $p'$, $w$, $|P'|$, $|Q'|$, and layer $\ell-1$.
\end{lemma}
\begin{proof}
Denote by $\Delta_w,\Delta_v\in \{-1,0,1\}$ the values such that
\begin{align*}
t_{w^{\ell},\ell-1}-\Cor_{w^{\ell},\ell}&\ge t_{w,\ell-1}+ 4s\kappa \Delta_w\\
t_{v^{\ell},\ell-1}-\Cor_{v^{\ell},\ell}&\le t_{v,\ell-1}-(4s-2)\kappa \Delta_v -\kappa
\end{align*}
according to \Cref{lem:v_step,lem:w_step}.
Using \Pfour and \Cref{lem:drift}, we bound
\begin{align*}
&\,t_{v,\ell-1}-t_{w,\ell-1}\\
\ge &\,t_{v^{\ell},\ell-1}-\frac{\Cor_{v,\ell}}{\vartheta}+(4s-2)\kappa \Delta_v+\kappa-(t_{w^{\ell},\ell-1}-\Cor_{w,\ell}- 4s\kappa \Delta_w)\\
\ge &\,t_{v^{\ell},\ell}+(4s-2)\kappa\Delta_v+\kappa-t_{w^{\ell},\ell}+ 4s\kappa\Delta_w-\frac{\kappa}{2}\\
\ge &\,(4s-2)\kappa(\Delta^{\ell}+\Delta_v+\Delta_w)+(\bar{\ell}-(\ell-1))\frac{\kappa}{2}+\kappa_s(|P^{\ell}|+\Delta_w)\\
\ge &\,(4s-2)\kappa(\Delta^{\ell}+\Delta_v+\Delta_w)+(\bar{\ell}-(\ell-1))\frac{\kappa}{2}+\kappa_s |P'|,
\end{align*}
where the last step exploits that $|P'|=|\suffix(P,p')|\le |P|\le|P^{\ell}|+\Delta_w$.
We claim that $\Delta\le \Delta^{\ell}+\Delta_v+\Delta_w$.
Proving this claim will complete the proof, as by the above inequality then
\begin{equation*}
t_{v,\ell-1}-t_{w,\ell-1}\ge (4s-2)\kappa\Delta+(\bar{\ell}-(\ell-1))\frac{\kappa}{2}+\kappa_s |P'|,
\end{equation*}
i.e., \Pfour for $v$, $p'$, $w$, $|P'|$, $|Q'|$, and layer $\ell-1$.

By \Cref{obs:non_suffix} and the prerequisites of the lemma, $P'=\suffix(P,p')$ or $p'=v=w^{\ell}$. To cover the possibility that $P'=\suffix(P,p')$, we distinguish several cases:
\begin{itemize}
  \item $p'=p^{\ell}$. Then $P'=P$ and $Q'=Q$, as $p'$ is the first node of both $P$ and $Q$. Hence,
  \begin{equation*}
  |P'|+|Q'|=|P|+|Q|\le|P^{\ell}|+|Q^{\ell}|+\Delta_w+\Delta_v.
  \end{equation*}
  We distinguish three subcases.
  \begin{itemize}
    \item $P^{\ell}$ and $Q^{\ell}$ do not share their first edge. Then
    \begin{equation*}
    \Delta\le |P'|+|Q'|=|P^{\ell}|+|Q^{\ell}|+\Delta_w+\Delta_v=\Delta^{\ell}+\Delta_w+\Delta_v.
    \end{equation*}
    \item $P^{\ell}$, $Q^{\ell}$, and $Q'$ share the same first edge. By \Cref{lem:pnotw}, $w\neq p'$. Therefore, $P'=P\neq (p')$, which means that $P^{\ell}$ and $P'$ have the same first edge, too. Thus, $Q'$ and $P'$ have the same first edge as well, and
    \begin{equation*}
    \Delta= |P'|+|Q'|-1=|P^{\ell}|+|Q^{\ell}|-1+\Delta_w+\Delta_v=\Delta^{\ell}+\Delta_w+\Delta_v.
    \end{equation*}
    \item $P^{\ell}$ and $Q^{\ell}$ have the same first edge, but $Q'$ does not. Since $Q^{\ell}\neq (p^{\ell})$, we have that $v^{\ell}\neq p^{\ell}$. By \Pfive, this implies that $v^{\ell}\notin P^{\ell}$. In particular, $v^{\ell}$ cannot be part of the first edge of $Q^{\ell}$ and $|Q^{\ell}|\ge 2$. As $p'=p^{\ell}$, $Q$ and $Q'$ both start with $p'$. Therefore, $Q'$ is a prefix of $Q^{\ell}$. However, $Q^{\ell}$ has the same first edge as $P^{\ell}$, while $Q'$ does not. Thus, $|Q'|=0\le |Q^{\ell}|+\Delta_v-1$. We conclude that
    \begin{equation*}
    \Delta= |P'|+|Q'|\le |P^{\ell}|+|Q^{\ell}|-1+\Delta_w+\Delta_v=\Delta^{\ell}+\Delta_w+\Delta_v.
    \end{equation*}
  \end{itemize}
  \item $v=p'\neq p^{\ell}$. Then $Q'=(p')$. Moreover, by the prerequisites of the lemma, $P'=\suffix(P,p')$. Since $p'\neq p^{\ell}$, we have that $|\suffix(P,p')|\le |P|-1$. By construction, $|Q'|\le |Q|+1$. Overall,
  \begin{equation*}
  \Delta\le |P'|+|Q'|\le |P|-1+|Q|+1=|P^{\ell}|+\Delta_w+|Q^{\ell}|+\Delta_v=\Delta^{\ell}+\Delta_w+\Delta_v.
  \end{equation*}
  \item $v\neq p'\neq p^{\ell}$. Thus, $p'=p$ and by \Cref{lem:p_step} $\{p^{\ell},p'\}$ is the first edge of $P^{\ell}$. Hence, $|P'|=|\suffix(P,p')|\le |P|-1 \le |P^{\ell}|+\Delta_w-1$. We distinguish two subcases.
  \begin{itemize}
    \item $P^{\ell}$ and $Q^{\ell}$ do not share their first edge. Then
    \begin{equation*}
    \Delta\le |P'|+|Q'|\le |P^{\ell}|+\Delta_w+|Q^{\ell}|+\Delta_v=\Delta^{\ell}+\Delta_w+\Delta_v.
    \end{equation*}
    \item $P^{\ell}$ and $Q^{\ell}$ share their first edge. As $v\neq p'$, $Q$ has the same first edge as $Q^{\ell}$, i.e., $\{p^{\ell},p'\}$. Hence, $|Q'|=|\suffix(Q,p'|=|Q|-1\le |Q^{\ell}|+\Delta_v-1$. We conclude that
    \begin{equation*}
    \Delta\le |P'|+|Q'|\le |P^{\ell}|+\Delta_w+|Q^{\ell}|+\Delta_v-2<\Delta^{\ell}+\Delta_w+\Delta_v.
    \end{equation*}
  \end{itemize}
\end{itemize}
It remains to consider the case that $P'\neq \suffix(P,p')$ and $p'=v=w^{\ell}$.
Then $|Q'|=(v)$ and $|P'|=|(p',w)|=1$, implying that $\Delta=1$.
By (P2), $v^{\ell}\neq w^{\ell}=v$.
If $\Delta_v=1$, then
\begin{equation*}
\Delta = 1 \le \Delta^{\ell} \le \Delta^{\ell}+\Delta_w+\Delta_v.
\end{equation*}
By \Cref{lem:v_step}, the remaining case is that $\Delta_v=-1$ and $\{v,v^{\ell}\}$ is the last edge of $Q^{\ell}$ or the first edge of $P^{\ell}$.
By \Cref{lem:vnotw}, it is impossible that $v=w^{\ell}$, so this edge must be the last one of $Q^{\ell}$ and distinct from the first one of $P^{\ell}$.
Moreover, by the prerequisites of the lemma, $p^{\ell}\neq w$, so it must hold that $|P^{\ell}|\ge 2$.
Overall, either
\begin{itemize}
  \item $|Q^{\ell}|\ge 2$ and
  \begin{equation*}
  \Delta = 1 \le |P^{\ell}|+|Q^{\ell}|-3 \le \Delta^{\ell}-2=\Delta^{\ell}+\Delta_w+\Delta_v,\mbox{ or}
  \end{equation*}
  \item $|Q^{\ell}|=1$ and $Q^{\ell}$ and $P^{\ell}$ do not share the first edge, yielding
  \begin{equation*}
  \Delta = 1 \le |P^{\ell}|+|Q^{\ell}|-2 = \Delta^{\ell}-2=\Delta^{\ell}+\Delta_w+\Delta_v.\qedhere
  \end{equation*}
\end{itemize}
\end{proof}

\begin{corollary}\label{cor:P4}
\Pfour and \Ptwo hold for $v$, $p'$, $w$, $|P'|$, $|Q'|$, and layer $\ell-1$.
\end{corollary}
\begin{proof}
Follows from \Cref{lem:swap}, \Cref{lem:P4}, and \Cref{obs:implications}.
\end{proof}

It remains to prove \Pthree.
\begin{lemma}\label{lem:P3}
\Pthree holds for $v$, $p'$, $|P'|$, and layer $\ell-1$.
\end{lemma}
\begin{proof}
If $v=p'$, the statement readily follows from \Cref{cor:P4} and \Cref{obs:implications}.
Therefore, assume that $v\neq p'$ and hence $p'=p$ in the following.
Denote by $\Delta_w\in \{-1,0,1\}$ the value such that
\begin{align*}
t_{w^{\ell},\ell-1}-\Cor_{w^{\ell},\ell}&\ge t_{w,\ell-1}+4s\kappa \Delta_w\\
t_{p^{\ell},\ell-1}-\Cor_{p^{\ell},\ell}&\le t_{p',\ell-1}
\end{align*}
according to \Cref{lem:w_step,lem:p_step}.

Using \Pthree and \Cref{lem:drift}, it follows that
\begin{align*}
t_{p',\ell-1}-t_{w,\ell-1}
&\ge t_{p^{\ell},\ell-1}-\Cor_{p^{\ell},\ell}-\left(t_{w^{\ell},\ell-1}-\frac{\Cor_{w^{\ell},\ell}}{\vartheta}\right)+\Delta_w 4s\kappa\\
&\ge t_{p^{\ell},\ell}-t_{w^{\ell},\ell}+4s\kappa\Delta_w-\frac{\kappa}{2}\\
&\ge 4s\kappa (|P^{\ell}|+\Delta_w)+\psi^s_{v^{\bar{\ell}},w^{\bar{\ell}}}(\bar{\ell}\,)-(\bar{\ell}-(\ell-1))\frac{\kappa}{2}.
\end{align*}
If $P'=\suffix(P,p')$, then $|P'|\le |P|\le |P^{\ell}|+\Delta_w$ and \Pthree for $v$, $p'$, $|P'|$, and layer $\ell-1$ readily follows from the above inequality.

Otherwise, by the assumption that $v\neq p'$ and \Cref{obs:non_suffix}, it holds that $p'=w^{\ell}$ and $w=p^{\ell}$, and $|P'|=|P^{\ell}|$.
Using \Cref{lem:drift,lem:swap_jump} together with \Pthree, we arrive at
\begin{align*}
t_{p',\ell-1}-t_{w,\ell-1}
&\ge t_{p',\ell-1}-t_{w^{\ell},\ell}+\left(d-u+\frac{\Lambda-d}{\vartheta}\right)\\
&\ge t_{p^{\ell},\ell-1}-\Cor_{p^{\ell},\ell}-t_{w^{\ell},\ell}+\left(d-u+\frac{\Lambda-d}{\vartheta}\right)\\
&\ge t_{p^{\ell},\ell}-t_{w^{\ell},\ell}-\frac{\kappa}{2}\\
&\ge 4s\kappa|P^{\ell}|+\psi^s_{v^{\bar{\ell}},w^{\bar{\ell}}}(\bar{\ell}\,)-(\bar{\ell}-(\ell-1))\frac{\kappa}{2}\\
&\ge 4s\kappa|P'|+\psi^s_{v^{\bar{\ell}},w^{\bar{\ell}}}(\bar{\ell}\,)-(\bar{\ell}-(\ell-1))\frac{\kappa}{2},
\end{align*}
i.e., \Pthree for $v$, $p'$, $|P'|$, and layer $\ell-1$.
\end{proof}

Finally, using these results it is not hard to show that \Pfive is satisfied as well.
\begin{lemma}\label{lem:P5}
\Pfive holds for $v$, $p'$, $|P'|$, and layer $\ell-1$.
\end{lemma}
\begin{proof}
Suppose that $v$ lies on $P'$.
By \Cref{cor:P4}, $v\neq w$.
Thus, if $P'=(p',w)$, $v=p'$, i.e., \Pfive holds for $v$, $p'$, $|P'|$, and layer $\ell-1$.

Otherwise, $P'=\suffix(P,p')$, implying that $v$ lies on $\suffix(P,p')$.
As $v\neq w$, this implies that $v$ lies on $P^{\ell}$.
Assuming for contradiction that $v\neq p'=p$, by \Cref{lem:p_step} we have that $\prefix(P,p')=\prefix(P^{\ell},p')$, which equals either $(p^{\ell})=(p')$ or $(p^{\ell},p')$.
Thus, the above entails that $v$ actually lies on $\suffix(P^{\ell},p')=\suffix(P^{\ell},p)$.
As then $p'=v$, this is a contradiction and we must indeed have that $p'=v$.
\end{proof}
\begin{corollary}\label{cor:invariants}
In the proof of \Cref{thm:psi_bound}, it must hold that $\ell=\underline{\ell}$.
\end{corollary}
\begin{proof}
Assuming for contradiction that $\ell>\underline{\ell}$, \Cref{cor:P4}, \Cref{lem:P3,lem:P5}, and \Cref{obs:implications} show that layer $\ell-1$ also satisfies the properties \Pone to \Pfive for some $v^{\ell-1},p^{\ell-1},w^{\ell-1}$, and paths $P^{\ell-1}$, $Q^{\ell-1}$, contradicting the minimality of $\ell$.
\end{proof}

\subsubsection*{Bounding Skews}
With our machinery for bounding $\Psi^s$ in place, it remains to perform the induction on $s\in \N_{>0}$ to wrap things up.
To anchor the induction at $s=1$, we exploit that $\Psi^1(\ell)\le \Xi^1(\ell)+2\kappa D$.
\begin{lemma}\label{lem:psi_1}
\begin{equation*}
\Psi^1(\ell)\le \begin{cases}
\Xi^1(0) &\mbox{if $\ell <4\Xi^1(0)/\kappa$}\\
4\kappa D & \mbox{else.}
\end{cases}
\end{equation*}
\end{lemma}
\begin{proof}
Recall that $\kappa=2(u+(1-1/\vartheta)(\Lambda-d))$.
Note that $\Xi^1(\ell)\le \Psi^1(\ell)+2\kappa D$ for all $\ell\in \N$.
By \Cref{thm:psi_bound}, we thus have for any $\underline{\ell}\le \bar{\ell}$ that
\begin{align*}
\Psi^1(\bar{\ell}\,)&\le \max\left\{0,\Xi^1(\underline{\ell}\,)-(\bar{\ell}-\underline{\ell}+1)\kappa\right\}+(\bar{\ell}-\underline{\ell}\,)\frac{\kappa}{2}\\
&\le \max\left\{0,\Psi^1(\underline{\ell}\,)+2\kappa D-(\bar{\ell}-\underline{\ell}+1)\kappa\right\}+(\bar{\ell}-\underline{\ell}\,)\frac{\kappa}{2}.
\end{align*}
In particular, we have that
\begin{equation*}
\Psi^1(\ell)\le \begin{cases}
\max\left\{4\kappa D,\Xi^1(0)\right\}&\mbox{if $\ell<8D$}\\
\max\left\{4\kappa D,\Psi^1(\ell-8D)-2\kappa D\right\}&\mbox{else.}
\end{cases}
\end{equation*}
By induction on $k\in \N$, we thus have that
\begin{equation*}
\Psi^1(\ell)\le \max\{4\kappa D, \Xi^1(0)-2k\kappa D\}
\end{equation*}
for all $\ell \in [8kD,8(k+1)D)$.
The claim of the lemma follows by noting that $\ell\ge 4\Xi^1(0)/\kappa$ results in $k\ge \Xi^1(0)/(2\kappa D)$.
\end{proof}
Note that this lemma shows that $\Psi^1$ self-stabilizes~\cite{dijkstra74self} within $O(\Xi^1(0)/\kappa)$ layers.

We remark that a more careful analysis reveals a bound on $\Psi^1(\ell)$ that converges to $2\kappa D$.
We confine ourselves to stating this result for the small input skew that we guarantee.
\begin{corollary}\label{cor:psi_1}
If $\localskew_0\le 4\kappa$, then $\Psi^1(\ell)\le 2\kappa D$ for all $\ell\in \N$.
\end{corollary}
\begin{proof}
Note that
\begin{equation*}
\Xi^1(0)=\max_{v,w\in V}\{t_{v,0}-t_{w,0}-2\kappa d(v,w)\}\le \max_{v,w\in V}\{(\localskew_0-2\kappa)d(v,w)\}\le (\localskew_0-2\kappa)D\le 2\kappa D.
\end{equation*}
By replacing $8D$ with $4D$ in the induction from the proof of \Cref{lem:psi_1}, we get that
\begin{equation*}
\Psi^1(\ell)\le \begin{cases}
\max\left\{2\kappa D,\Xi^1(0)\right\}&\mbox{if $\ell<4D$}\\
\max\left\{2\kappa D,\Psi^1(\ell-4D)\right\}&\mbox{else,}
\end{cases}
\end{equation*}
implying a uniform bound of $\Psi^1(\ell)\le 2\kappa D$ for all $\ell\in \N$.
\end{proof}
For the sake of completeness, we also infer that $\sup_{\ell\in \N}\{\Psi^0(\ell)\}$, also referred to as the \emph{global skew} in the literature, is in $O(u+(1-1/\vartheta)(\Lambda-d))$.
Provided that $\Lambda\in O(d+u/(\vartheta-1))$, this bound is asymptotically optimal~\cite{biaz01closed}.
\begin{corollary}\label{cor:psi_0}
If $\localskew_0\le 4\kappa$, then $\Psi^0(\ell)\le 6\kappa D\in O(u+(1-1/\vartheta)(\Lambda-d))$ for all $\ell\in \N$.
\end{corollary}
\begin{proof}
Follows from \Cref{cor:psi_1}, the fact that $\Psi^0(\ell)\le \Psi^1(\ell)+4\kappa D$, and the choice of $\kappa$.
\end{proof}

In order to bound the local skew, we now turn to attention to $\Psi^s(\ell)$ for $s>1$.
\begin{lemma}\label{lem:psi_s}
For some $s\in \N$, $s>0$, suppose that $\Psi^{s-1}(\ell)\le \Psi^{s-1}$ for all $\ell \in \N$.
Then
\begin{equation*}
\Psi^s(\ell)\le \begin{cases}
\Xi^s(0)+\frac{\Psi^{s-1}}{2} &\mbox{if $\ell<\Psi^{s-1}/\kappa$}\\
\frac{\Psi^{s-1}}{2} &\mbox{else.}
\end{cases}
\end{equation*}
\end{lemma}
\begin{proof}
Recall that $\kappa=2(u+(1-1/\vartheta)(\Lambda-d))$.
For $\ell<\Psi^{s-1}/\kappa$, by \Cref{thm:psi_bound} with $\bar{\ell}=\ell$ and $\underline{\ell}=0$ we have that
\begin{equation*}
\Psi^s(\ell)\le \Xi^s(0)+\frac{\kappa\ell}{2}\le \Xi^s(0)+\frac{\Psi^{s-1}}{2}.
\end{equation*}

Note that $\Xi^s(\ell)\le \Psi^{s-1}(\ell)\le \Psi^{s-1}$ for all $\ell\in \N$.
Thus, for $\ell\ge \Psi^{s-1}/\kappa$ by \Cref{thm:psi_bound} with $\bar{\ell}=\ell$ and $\underline{\ell}=\ell - \lfloor \Psi^{s-1}/\kappa\rfloor$ we have that
\begin{equation*}
\Psi^s(\ell)\le \max\left\{0,\Xi^s\left(\ell-\left\lfloor \frac{\Psi^{s-1}}{\kappa}\right\rfloor\right)-\left(\left\lfloor \frac{\Psi^{s-1}}{\kappa}\right\rfloor+1\right)\kappa\right\}+\left\lfloor \frac{\Psi^{s-1}}{\kappa}\right\rfloor\frac{\kappa}{2}
\le \frac{\Psi^{s-1}}{2}.\qedhere
\end{equation*}
\end{proof}
Using this lemma, we can bound the local skew by $O(\kappa (1+\log D))=O((u+(1-1/\vartheta)(\Lambda-d))(1+\log D))$.

\thmlocal*
\begin{proof}
By \Cref{lem:layer0}, $\localskew_0\le 4\kappa$.
By \Cref{cor:psi_1}, $\Psi^1(\ell)\le 2\kappa D$ for all $\ell\in \N$.
By the assumption that $\localskew_0\le 4\kappa$, for all $s>1$ we have that
\begin{equation*}
\Xi^s(0)=\max_{v,w\in V}\{t_{v,0}-t_{w,0}-(4s-2)\kappa d(v,w)\}\le\max_{v,w\in V}\{(\localskew_0-6\kappa)d(v,w)\}=0.
\end{equation*}
Hence, inductive use of \Cref{lem:psi_s} yields that $\Psi^s(\ell)\le 2^{2-s}\kappa D$.
In particular, $\Psi^{\lfloor \log D\rfloor}\le 8\kappa$.
The claim now follows by \Cref{obs:skew}.
\end{proof}
Moreover, in addition we obtain the following self-stabilization property.
\begin{theorem}\label{thm:local_stab}
If for $s,s'\in \N$, $s\le s'$, we have that $\Psi^s(\ell)\le \Psi^s$ for all $\ell\ge \underline{\ell} \in \N$, then for $\ell\ge \underline{\ell}$
\begin{equation*}
\localskew_{\ell}\le \begin{cases}
4s\kappa+\Psi^s&\mbox{if $\underline{\ell}\le \ell<\underline{\ell}+2\Psi^s/\kappa$ and}\\
4s'\kappa+\frac{\Psi^s}{2^{s'-s}}&\mbox{if $\ell\ge \underline{\ell}+2\Psi^s/\kappa$.}
\end{cases}
\end{equation*}
\end{theorem}
\begin{proof}
Inductive use\footnote{As is, the lemma applies only if $\underline{\ell}=0$. However, the algorithm and hence all statements are invariant under shifting indices by $\underline{\ell}$.} of \Cref{lem:psi_s} yields for $s'\ge s$ and $\ell\ge \underline{\ell}+\sum_{\sigma=s+1}^{s'}\Psi^{s}/(2^{\sigma-s}\kappa)$ that
\begin{equation*}
\Psi^{s'}\le \frac{\Psi^{s}}{{2^{s'-s}}}.
\end{equation*}
Since the sum forms a geometric series, this in particular applies to all $\ell\ge \underline{\ell}+2\Psi^s/\kappa$.
The claim now follows by applying \Cref{obs:skew}.
\end{proof}

\subsection{Bounding Skews in the Presence of Faults}\label{sec:faults}
To analyze how skews evolve with faults, we relate the setting with faults to the bounds we have for a fault-free system.
The key property the algorithm guarantees is that, up to an additive $2\kappa$, the pulse time is within the interval spanned by the correct predecessors' pulse times plus $\Lambda$.
We first show this for the case that for some node $(v,\ell)$, $(v,\ell-1)$ is faulty.
\begin{lemma}\label{lem:fault_self}
Suppose that the only faulty predecessor of $(v,\ell)\in V_{\ell}$, $\ell>0$, is $(v,\ell-1)$.
Denote
\begin{align*}
t_{\min}&:=\min_{\{v,w\}\in E}\{t_{w,\ell-1}\}\mbox{ and}\\
t_{\max}&:=\max_{\{v,w\}\in E}\{t_{w,\ell-1}\}.
\end{align*}
Then
\begin{equation*}
t_{\min}+\Lambda-2\kappa \le t_{v,\ell} \le t_{\max}+\Lambda+2\kappa.
\end{equation*}
\end{lemma}
\begin{proof}
By the assumption of the lemma, for all $\{v,w\}\in E$, $(w,\ell-1)\notin F$.
We have that
\begin{align*}
H_{\own}-H_{\max}&=\min_{s\in \N}\{H_{\own}-H_{\max}+4s\kappa\}\\
&\le \min_{s\in \N}\{\max\{H_{\own}-H_{\max}+4s\kappa,H_{\own}-H_{\min}-4 s\kappa \}\}\\
&\le \max\{H_{\own}-H_{\max},H_{\own}-H_{\min}\}\\
&= H_{\own}-H_{\min}.
\end{align*}
Hence, abbreviating
\begin{equation*}
\Delta=\min_{s\in \N}\{\max\{H_{\own}-H_{\max}+4s\kappa,H_{\own}-H_{\min}-4 s\kappa \}\}-\frac{\kappa}{2},
\end{equation*}
it holds that
\begin{equation*}
H_{\own}-H_{\max}-\frac{\kappa}{2}\le \Delta\le H_{\own}-H_{\min}-\frac{\kappa}{2}.
\end{equation*}
Taking into account the adjustments in case $\Delta\notin [0,\vartheta \kappa]$ and using that $H_{\min}\le H_{\max}$ we get that
\begin{equation*}
H_{\own}-H_{\max}-\frac{3\kappa}{2} \le \Cor_{v,\ell}
\le H_{\own}-H_{\min}+\frac{3\kappa}{2}.
\end{equation*}
Therefore, the local time $H_{v,\ell}(t_{v,\ell})=H_{\own}+\Lambda-d-\Cor_{v,\ell}$ at which $(v,\ell)$ generates its pulse satisfies
\begin{equation*}
H_{\min}+\Lambda-d-\frac{3\kappa}{2}\le H_{v,\ell}(t_{v,\ell}) \le H_{\max}+\Lambda-d+\frac{3\kappa}{2}.
\end{equation*}
If $H_{\min}>H_{v,\ell}(t_{v,\ell})$, we have that
\begin{equation*}
t_{\min}-t_{v,\ell}\le H_{\min}-H_{v,\ell}(t_{v,\ell}).
\end{equation*}
Applying the lower bound of $d-u$ on message delay and \Cref{eq:kappa}, we get that
\begin{equation*}
t_{v,\ell}\ge t_{\min}+d-u+\Lambda-d-\frac{3\kappa}{2}>t_{\min}+\Lambda-2\kappa.
\end{equation*}
If $H_{\min}\le H_{v,\ell}(t_{v,\ell})$, the bounds on message delays and hardware clock drift together with \Cref{eq:kappa} yield that
\begin{align*}
t_{v,\ell}&\ge t_{\min}+d-u+\frac{\Lambda-d-3\kappa/2}{\vartheta}\\
&>t_{\min}+\Lambda-\frac{3\kappa}{2}-u-\left(1-\frac{1}{\vartheta}\right)(\Lambda-d)\\
&=t_{\min}+\Lambda-2\kappa.
\end{align*}

Concerning the upper bound on $t_{v,\ell}$, note that because $t_{v,\ell}$ is increasing in $H_{v,\ell}(t_{v,\ell})$, to bound $t_{v,\ell}$ from above we may assume that
\begin{equation*}
H_{v,\ell}(t_{v,\ell}) = H_{\max}+\Lambda-d+\frac{3\kappa}{2}>H_{\max},
\end{equation*}
where the last step uses \Cref{eq:Lambda}.
In this case,
\begin{equation*}
t_{v,\ell}-t_{\max}\le H_{v,\ell}(t_{v,\ell}) - H_{\max}+d\le \Lambda+\frac{3\kappa}{2} < \Lambda+2\kappa.\qedhere
\end{equation*}
\end{proof}
Similar reasoning covers the case that for some $(v,\ell)\in V_{\ell}$ and $\{v,w\}\in E$, $(w,\ell-1)$ is faulty.
\begin{lemma}\label{lem:fault_neighbor}
Suppose that for $(v,\ell)\in V_{\ell}$, $\ell>0$, $(v,\ell-1)$ is not faulty, and at most one predecessor is faulty.
Denoting
\begin{align*}
t_{\min}&:=\min_{\substack{((w,\ell-1),(v,\ell))\in E_{\ell-1}\\ (w,\ell-1)\notin F}}\{t_{w,\ell-1}\}\mbox{ and}\\
t_{\max}&:=\max_{\substack{((w,\ell-1),(v,\ell))\in E_{\ell-1}\\ (w,\ell-1)\notin F}}\{t_{w,\ell-1}\},
\end{align*}
then
\begin{equation*}
t_{\min}+\Lambda-2\kappa \le t_{v,\ell} \le t_{\max}+\Lambda.
\end{equation*}
\end{lemma}
\begin{proof}
By \Cref{lem:drift}, $\Cor_{v,\ell}\ge 0$ implies that 
\begin{equation*}
t_{v,\ell}-t_{\min}\le t_{v,\ell}-t_{v,\ell-1}\le \Lambda,
\end{equation*}
while $\Cor_{v,\ell}\le \vartheta \kappa$ yields that 
\begin{equation*}
t_{v,\ell}-t_{\max}\ge t_{v,\ell}-t_{v,\ell-1}\ge d-u+\frac{\Lambda-d}{\vartheta}-\kappa \ge \Lambda - 2\kappa.
\end{equation*}
It remains to show the upper bound on $t_{v,\ell}$ if $\Cor_{v,\ell}<0$ and the lower bound if $\Cor_{v,\ell}>\vartheta \kappa$.

Consider first the case that $\Cor_{v,\ell}<0$.
Accordingly,
\begin{equation*}
\Cor_{v,\ell}=H_{\own}-H_{\min}-\frac{\kappa}{2}+2\kappa>H_{\own}-H_{\min}.
\end{equation*}
It follows that
\begin{align*}
H_{v,\ell}(t_{v,\ell})=H_{\own}+\Lambda-d-\Cor_{v,\ell}\le H_{\min}+\Lambda-d.
\end{align*}
Noting that the reception time of the first message from a predecessor is bounded from above by the reception time of the message from a correct predecessor, we conclude that
\begin{equation*}
t_{v,\ell}\le t_{\min}+\Lambda.
\end{equation*}

Now consider the case that $\Cor_{v,\ell}>\vartheta \kappa$.
Consequently,
\begin{equation*}
\Cor_{v,\ell}=H_{\own}-H_{\max}-\frac{\kappa}{2}-\kappa>H_{\own}-H_{\max}-\vartheta u.
\end{equation*}
It follows that the local time $H$ at which $(v,\ell)$ generates its pulse satisfies that
\begin{align*}
H=H_{\own}+\Lambda-d-\Cor_{v,\ell}\ge H_{\max}+\Lambda-d+\vartheta u.
\end{align*}
Noting that the reception time of the latest message from a predecessor is bounded from below by the reception time of the latest message from a correct predecessor, by \Cref{eq:kappa} we conclude that
\begin{equation*}
t_{v,\ell}\ge t_{\max}+d+\frac{\Lambda-d}{\vartheta}>t_{v,\ell-1}+\Lambda - \kappa.\qedhere
\end{equation*}
\end{proof}
\begin{corollary}\label{cor:fault}
Denote
\begin{align*}
t_{\min}&:=\min_{\substack{((w,\ell-1),(v,\ell))\in E_{\ell-1}\\ (w,\ell-1)\notin F}}\{t_{w,\ell-1}\}\mbox{ and}\\
t_{\max}&:=\max_{\substack{((w,\ell-1),(v,\ell))\in E_{\ell-1}\\ (w,\ell-1)\notin F}}\{t_{w,\ell-1}\}.
\end{align*}
Then
\begin{equation*}
t_{\min}+\Lambda-2\kappa \le t_{v,\ell} \le t_{\max}+\Lambda+2\kappa.
\end{equation*}
\end{corollary}
\begin{proof}
Immediate from \Cref{lem:fault_self,lem:fault_neighbor} and the assumption that no node has more than one faulty predecessor.
\end{proof}
Using this result, we can bound the impact of a fault in layer $\ell-1$ on successors via the skew bounds of close-by nodes on layer $\ell-1$;
we exploit that all bounds we show would in fact also apply to the faulty node if it was correct.
\begin{lemma}\label{lem:faulty_layer}
Suppose for a node $(v,\ell)\in V_{\ell}$, $\ell>0$, that one of its predecessors is faulty.
Moreover, assume that in an execution that differs only in that the faulty predecessor of $(v,\ell)$ is correct, it holds that $\max_{\{v,w\}\in E}\{|t_{v,\ell-1}-t_{w,\ell-1}|\}\le B$.
Then in the execution with the predecessor being faulty, the pulse time of $(v,\ell)$ differs by at most $2B+4\kappa$.
\end{lemma}
\begin{proof}
Denote by standard variables values in the execution without the predecessor being faulty and by primed variables values in the one where it is.
In particular, for node $(v,\ell)\in V_{\ell}\setminus F$
\begin{align*}
t_{\min}&:=\min_{((w,\ell-1),(v,\ell))\in E_{\ell-1}}\{t_{w,\ell-1}\},\\
t_{\max}&:=\max_{((w,\ell-1),(v,\ell))\in E_{\ell-1}}\{t_{w,\ell-1}\},\\
t_{\min}'&:=\min_{\substack{((w,\ell-1),(v,\ell))\in E_{\ell-1}\\ (w,\ell-1)\notin F}}\{t_{w,\ell-1}\}\mbox{, and}\\
t_{\max}'&:=\max_{\substack{((w,\ell-1),(v,\ell))\in E_{\ell-1}\\ (w,\ell-1)\notin F}}\{t_{w,\ell-1}\}.
\end{align*}
denote the earliest and latest pulsing times of (correct) predecessors without and with faults on layer $\ell-1$, respectively.

Observe that
\begin{equation*}
t_{v,\ell-1}-B\le t_{\min}\le t_{\min'}\le t_{\max'}\le t_{\max}\le t_{v,\ell-1}+B.
\end{equation*}
Hence, \Cref{cor:fault} (applied to both executions) shows that 
\begin{align*}
t_{v,\ell-1}-B-2\kappa &\le t_{v,\ell}\le t_{v,\ell-1}+B+2\kappa\mbox{ and}\\
t_{v,\ell-1}-B-2\kappa &\le t_{v,\ell}'\le t_{v,\ell-1}+B+2\kappa.\qedhere
\end{align*}
\end{proof}
Finally, we observe that such a ``time shift'' propagates without further increase, so long as there are no faults.
However, a subtlety here is that this is only true for our \emph{bounds} on timing:
a change in timing might leave more time for drift of the local clock to accumulate;
since our worst-case bounds include the maximum time error that can possibly be accumulated from drift (so long as local skews do not become exceedingly large), this is already accounted for in the bound provided by \Cref{lem:drift}.
Hence, we obtain the following generalized variant of \Cref{lem:drift}.
\begin{lemma}\label{lem:drift_fault}
Suppose that for $v\in V$ and $\ell\in \N_{>0}$ the predecessors of $(v,\ell)$ are correct.
If we shift the pulse times of these predecessors by at most $\delta \in \R$, where \Cref{eq:Lambda} still holds for the shifted times, then
\begin{equation*}
d-u+\frac{\Lambda-d-\Cor_{v,\ell}}{\vartheta}-\delta \le t_{v,\ell}'-t_{v,\ell-1} \le \Lambda-\Cor_{v,\ell}+\delta,
\end{equation*}
where $t_{v,\ell}'$ denotes the pulse time of $(v,\ell)$ in the execution with the shifts applied.
\end{lemma}
\begin{proof}
Pulse times are increasing as functions of pulse times of predecessors.
Therefore, in order to maximize or minimize $t_{v,\ell}'$, we need to maximize or minimize the predecessors' pulse times, respectively.
Shifting all predecessors' pulse times uniformly by $\delta$ also shifts $t_{v,\ell}'$ by $\delta$ relative to $t_{v,\ell}$.
The statement now follows analogously to the proof of \Cref{lem:drift}, carrying the uniform shift through all inequalities.
\end{proof}
With these tools in place, we can conclude that skews do not grow arbitrarily in the face of faults.
\thmexp*
\begin{proof}
We prove by induction on the number $i\le f$ of layers $\ell>0$ with faults that the skew is bounded by $B_i:=4\kappa(2+\log D)5^i\sum_{j=0}^i 5^{-j}\in O(5^f\kappa\log D)$.
By \Cref{cor:layer0}, $\localskew_0\le \kappa/2<4\kappa$.
Thus, if there are no faults in layers $\ell>0$, by \Cref{thm:local} we have that $\localskew_{\ell}\le B_0:= 4\kappa(2+\log D)$ for all $\ell\in \N$.

Assume that we completed step $i\in \N$ and that $\ell_{i+1}$ is the next layer where faults need to be added.
Then we have that for all $\ell\le \ell_{i+1}$ that $\localskew_{\ell'}\le B_i=4\kappa(2+\log D)5^f\sum_{j=0}^i 5^{-j}$ both before and after adding the faults on layer $i+1$.
By \Cref{lem:faulty_layer}, it follows that pulsing times on layer $\ell_{i+1}+1$ do not change by more than $2B_i+4\kappa$ due to the addition of faults.
By \Cref{lem:drift_fault}, this extends to all bounds\footnote{Due to drifting hardware clocks, this does not apply to the pulse times themselves. However, we rely on \Cref{lem:drift} to prove our bounds in the absence of faults, and this is covered by \Cref{lem:drift_fault}.} we compute on pulse times in layers $\ell>\ell_{i+1}$.
Since $D\ge 1$ and thus $\log D\ge 0$, we get that the local skew in step $i+1$ is bounded by
\begin{equation*}
5B_i+4\kappa = 4\kappa(2+\log D)5^{i+1}\sum_{j=0}^i 5^{-j}+4\kappa \le 4\kappa(2+\log D)5^{i+1}\sum_{j=0}^{i+1} 5^{-j}=B_{i+1}.\qedhere
\end{equation*}
\end{proof}

\subsection*{Bounding Skews with Uniform Fault Distribution}\label{sec:uniform}
The bound in \Cref{thm:fault_worst_case}, which is exponential in $f$, seems to suggest that the system can only support a very small number of faults or the local skew explodes.
However, we have not yet taken into account that the starting point of our entire approach is the assumption that faults are sufficiently sparse, meaning that it is highly unlikely that many of them cluster together in a way that causes an exponential pile-up of local skew.
This enables the self-stabilization properties of the algorithm to prevent such a build-up altogether.

In the following, assume that each node fails uniformly and independently with probability $o(n^{-1/2})$.
This is the largest probability of error we can support while guaranteeing that no node has more than one faulty predecessor with probability $1-o(1)$.
A key observation is that this entails that within a fairly large distance of $n^{1/12}$, no node has more than a constant number of faulty nodes that can influence it.
We now formalize and show this claim.
\begin{definition}[Distance-$\delta$ Ancestors]
For node $(v,\ell)\in V_{\ell}$ and $\delta\in \N$, its \emph{distance-$\delta$ ancestors} are all nodes $(w,\ell')\in V_G\setminus \{(v,\ell)\}$ such that there is a (directed) path of length at most $\delta$ from $(w,\ell')$ to $(v,\ell)$ in $G$.
\end{definition}

\begin{definition}[Distance-$\delta$ $k$-faulty]
Node $(v,\ell)\in V_{\ell}$, $\ell\in \N_{>0}$ is \emph{distance-$\delta$ $k$-faulty} if $k\in \N$ is minimal such that there are at most $k$ faulty nodes among the distance-$((k+1)\delta)$ ancestors of $(v,\ell)$.
\end{definition}

\begin{observation}\label{obs:k_faulty}
Suppose that $\delta\le n^{1/12}$.
If nodes fail independently with probability $p\in o(1/\sqrt{n})$, then with probability $1-o(1)$ all nodes are distance-$\delta$ $k$-faulty for $k\le 2$.
\end{observation}
\begin{proof}
In order to be distance-$\delta$ $k$-faulty for $k> 2$, a node must have at least $3$ faults among its distance-$(3\delta)$ ancestors.
The number of these ancestors is bounded by $(3\delta)^2\in O(n^{1/6})$.
Since $p\in o(1/\sqrt{n})$, the probability for this to happen is bounded by $O(p^3 \binom{n^{1/6}}{3})=O(p^3\sqrt{n}) \subset o(1/n)$.
The claim follows by applying a union bound over all $n$ nodes.
\end{proof}

We can exploit this to control how much skews grow as the result of faults much better.
\begin{lemma}\label{lem:fault_accumulation}
Suppose that $\Psi^s(\ell)\le B_{s,\ell}$ and $\localskew_{\ell}\le B$ for all layers $\ell\ge \underline{\ell}$ and $s\in \N$, where $\ell,\underline{\ell}\in \N$, if there are no faults in these layers.
If no node in a layer $\ell\ge \underline{\ell}$ has more than $2$ faulty nodes among its distance-$(\ell-\underline{\ell}\,)$ ancestors, then $\Psi^s(\ell)\le B_{s,\ell}+12B+24\kappa$ for all $\ell \ge \underline{\ell}$.
\end{lemma}
\begin{proof}
We examine by how much adding faults on layers $\ell \ge \bar{\ell}$ might affect pulsing times.
For $\ell\ge \bar{\ell}$ and $(v,\ell)\in V_{\ell}$, denote by $f_{v,\ell}\in \{0,1,2\}$ the number of faulty distance-$(\ell-\bar{\ell})$ ancestors of $(v,\ell)$.
For $f_{v,\ell}=0$, there is no change in $t_{v,\ell}$.
For $f_{v,\ell}>0$, consider two cases.
If $(v,\ell)$ has no faulty predecessor, then by \Cref{lem:drift_fault}, $t_{v,\ell}$ is changed at most by the maximum shift that any of its predecessors undergoes.
On the other hand, if $(v,\ell)$ does have a faulty predecessor, then $f_{v,\ell}>f_{w,\ell-1}$ for all correct predecessors of $(v,\ell)$.
Thus, by \Cref{lem:faulty_layer} we can bound shifts by $B_{f_{v,\ell}}$, where $B_0:=0$ and $B_{f+1}:=2(B+B_f)+4\kappa$.

By assumption, $f_{v,\ell}\le 2$ and hence the maximum shift is bounded by $B_2=6B+12\kappa$.
We conclude that $\Psi^s(\ell)\le B_{s,\ell}+2B_2=B_{s,\ell}+12B+24\kappa$, as claimed.
\end{proof}
Together with \Cref{lem:fault_accumulation}, \Cref{obs:k_faulty} shows that skews do not increase by more than a constant factor within $n^{1/12}$ layers.
However, we need to handle a total of $\Theta(\sqrt{n})$ layers.
To this end, we slice up the task into chunks of $n^{1/12}$ layers and leverage the self-stabilization properties of the algorithm.
For simplicity, in the following we assume that $n^{1/12}$ is an integer.
As we prove asymptotic bounds, this does not affect the results.
\begin{definition}[Slices]
\emph{Slice} $i\in \N_{>0}$ consists of layers $\ell \in [(i-1)n^{1/12},i n^{1/12}-1]$.
\end{definition}
Note that there are no more than $n^{5/12}$ slices, because the nodes are arranged in square grid.
Due to the duplication of nodes on layer $0$ and the boundary nodes on layers $\ell>0$, the number of slices is actually $n^{5/12}-\Theta(1)$.

As our next step towards a probabilistic skew bound, we prove that \emph{if} the local skew remains bounded, then for levels $s$ that are not too large, $\Psi^s$ remains almost as small as without faults.
First, we show a loose bound that naively accumulates shifts slice by slice.
\begin{lemma}\label{lem:fault_mitigation_base}
Suppose that
\begin{itemize}
  \item $\localskew_0\le 4\kappa$,
  \item each node is distance-$n^{1/12}$ $k$-faulty for $k\le 2$, and
  \item $\localskew_{\ell}\le B$ for all $\ell\in \N$.
\end{itemize}
Then for each $s\in \N$ and layer $\ell$ in slice $i\in \N_{>0}$, we have that
\begin{equation*}
\Psi^s(\ell)\le 2^{2-s}\kappa D+i(12B+24\kappa).
\end{equation*}
for all $\ell\in \N$.
\end{lemma}
\begin{proof}
Assume first that there are no faults.
In this case, analogously to the proof of \Cref{thm:local}, we get that $\Psi^s(\ell)\le 2^{2-s}\kappa D$ for all $\ell \in \N$.
Now we ``add'' faults inductively slice by slice, by \Cref{lem:fault_accumulation} each time increasing the bound on $\Psi^s(\ell)$ by $12B+24\kappa$ for all slices $j\ge i$.
\end{proof}
For larger values of $s$, $2^{2-s}\kappa D\ll n^{1/12}$, meaning that this naive bound is insufficient to show that $\Psi^s(\ell)$ does not increase much compared to the fault-free setting.
However, we can take things much further by leveraging \Cref{thm:local_stab}.
\begin{lemma}\label{lem:fault_mitigation}
Suppose that
\begin{itemize}
  \item $\localskew_0\le 4\kappa$,
  \item each node is distance-$n^{1/12}$ $k$-faulty for $k\le 2$, and
  \item $\localskew_{\ell}\le B\in o(n^{1/12}\kappa/\log D)$ for all $\ell\in \N$.
\end{itemize}
Then for\footnote{If $D=1$, we assume the upper bound on $s$ to be negative and the claim is vacuously true. Note that we are making an asymptotic statement in $n$ and that $D$ grows with $n$, so this case is actually of no concern here.} $s\in \N_{>0}$, $s\le \log D-\log(B/\kappa)-2\log\log D$, it holds that
\begin{equation*}
\Psi^s(\ell)\le \Psi^s\in (1+o(1))2^{2-s}\kappa D.
\end{equation*}
\end{lemma}
\begin{proof}
Note that $D\in \Theta(n^{1/2})$ and hence $\log \log D\in \omega(1)$.
Accordingly, the prerequisites of the lemma ensure that $n^{5/12}(B+\kappa)\in o(\kappa D/\log D)$ and $B+\kappa \in o(\Psi^{s-1}/\log D)$.
Hence, we may fix a suitable $\varepsilon \in o(1)$ such that
\begin{align*}
n^{5/12}(12B+24\kappa)&\le \frac{\varepsilon}{\log D}\cdot 2\kappa D\mbox{ and}\\
\left(\left\lceil\frac{\Psi^{s-1}}{n^{5/12}\kappa}\right\rceil+1\right)(12B+24\kappa)&\le \frac{\varepsilon}{4\log D}\cdot \Psi^{s-1}.
\end{align*}
We claim that if $n$ is sufficiently large such that $\varepsilon\le 1$, we have that
\begin{equation*}
\Psi^s(\ell)\le \Psi^s:=2^{2-s}\kappa D \cdot \left(1+\frac{\varepsilon s}{\log D}\right),
\end{equation*}
which we show by induction on $s\in \N_{>0}$.

For the base case of $s=1$, note that there are no more than $n^{5/12}$ slices, yielding by \Cref{lem:fault_mitigation_base} that
\begin{equation*}
\Psi^1(\ell)\le 2\kappa D+n^{5/12}(12B+24\kappa)\le \left(1+\frac{\varepsilon}{\log D}\right)2\kappa D,
\end{equation*}
i.e., indeed $\Psi^1(\ell)\le \Psi^1$.

Now assume that the claim holds for $s-1\in \N_{>0}$.
Then, by \Cref{lem:fault_mitigation_base} and the induction hypothesis, for layers $\ell$ in slices $i\le \lceil(\Psi^{s-1}/(n^{1/12}\kappa)\rceil$, we have that
\begin{equation*}
\Psi^s(\ell)\le 2^{2-s}\kappa D+\left\lceil\frac{\Psi^{s-1}}{n^{1/12}\kappa}\right\rceil(12B+24\kappa)<\frac{\Psi^{s-1}}{2}+\left(\left\lceil\frac{\Psi^{s-1}}{n^{1/12}\kappa}\right\rceil+1\right)(12B+24\kappa).
\end{equation*}
For a layer $\ell$ in a slice $i>\lceil(\Psi^{s-1}/(n^{1/12}\kappa)\rceil$, assume first that we add only faults in slices $j<i-\lceil(\Psi^{s-1}/(n^{1/12}\kappa)\rceil$.
Hence, we can apply \Cref{lem:psi_s}, shifting layer indices such that ``layer $0$'' is the first layer of slice $i-\lceil(\Psi^{s-1}/(n^{1/12}\kappa)\rceil$.
In this setting, we thus have that $\Psi^s(\ell)\le \frac{\Psi^{s-1}}{2}$.
We now apply \Cref{lem:fault_accumulation} inductively to slices $j\in [i-\lceil(\Psi^{s-1}/(n^{1/12}\kappa)\rceil,i]$, adding in total $(\lceil\Psi^{s-1}/(n^{1/12}\kappa)\rceil+1)(12B+24\kappa)$ to the bound, i.e.,
\begin{align*}
\Psi^s(\ell)&\le \frac{\Psi^{s-1}}{2}+\left(\left\lceil\frac{\Psi^{s-1}}{n^{1/12}\kappa}\right\rceil+1\right)(12B+24\kappa)\\
&\le \left(\frac{1}{2}+\left(\frac{\varepsilon}{4\log D}\right)\right)\Psi^{s-1}\\
&=\left(\frac{1}{2}+\left(\frac{\varepsilon}{4\log D}\right)\right)2^{2-(s-1)}\kappa D \cdot \left(1+\frac{\varepsilon (s-1)}{\log D}\right)\\
&=2^{2-s}\kappa D \cdot \left(1+\frac{\varepsilon (s-1/2)}{\log D}+\frac{\varepsilon^2}{2\log^2 D}\right)\\
&\le 2^{2-s}\kappa D \cdot \left(1+\frac{\varepsilon s}{\log D}\right),
\end{align*}
where the last step assumes that $n$ is large enough so that $\varepsilon \le 1$.
\end{proof}
Our goal is to bound $\Psi^{\lfloor \log D\rfloor}$ by $O(\kappa \log D)$, since by \Cref{obs:skew} this implies a bound of $O(\kappa \log D)$ on the local skew.
Thus, we will use the above lemma with $B\in O(\kappa \log D)$, which gets us within $O(\log \log D)$ levels of our ``target'' level $\lfloor \log D\rfloor$.
To bridge this remaining gap, we exploit that the time required for stabilizing the remaining $O(\log \log D)$ levels after a fault-induced increase of skews takes only $\log^{O(1)}D = \log^{O(1)}n \subset o(n^{1/12})$ layers, since the involved potentials are bounded by $o(\kappa n^{1/12})$.

\begin{lemma}\label{lem:local_faults}
Suppose that
\begin{itemize}
  \item $\localskew_0\le 4\kappa$ and
  \item each node is distance-$n^{1/12}$ $k$-faulty for $k\le 2$.
\end{itemize}
Then $\localskew_{\ell}\in O(\kappa \log D)$.
\end{lemma}
\begin{proof}
Assume towards a contradiction that the claim is false, and let $\bar{\ell}\in \N_{>0}$ be minimal such that $\localskew_{\ell}$ is too large.
Hence, for layers $\ell<\bar{\ell}$, we may assume that $\localskew_{\ell}\le C\kappa \log D$ for a sufficiently large constant $C$.
 
Consider $s= \lfloor\log D-\log(B/\kappa)-2\log\log D-\log C\rfloor-5$.
By \Cref{lem:fault_mitigation}, for all $\ell\in \N$, $\ell<\bar{\ell}$ it holds that
\begin{equation*}
\Psi^s(\ell)\in \Psi^s := (1+o(1))2^{2-s}\kappa D \subseteq \left(\frac{1}{4}+o(1)\right)\log^3 D,
\end{equation*}
which for sufficiently large $n$ is smaller than $\lfloor\log^3 D\rfloor/2$.
In fact, this bound also applies to layer $\bar{\ell}$, since the pulsing times of nodes on layer $\bar{\ell}$ depend only on the behavior of nodes on layer $\bar{\ell}-1$ and the delays of messages sent to nodes on layer $\bar{\ell}$.

Now assume that $n$ is sufficiently large.
This ensures that $\log^3 D \le n^{1/12}$, implying by the prerequisites of the lemma that each node is distance-$(\log^3 D)$ $k$-faulty for $k\le 2$.
Consider adjacent correct nodes $(v,\ell),(w,\ell)\in V_{\ell}\setminus F$ for any $\ell\in \N$, $\ell\le \bar{\ell}$, and $\{v,w\}\in E$.
We first show that distance-$(\log^3 D)$ $0$-faulty nodes satisfy that
\begin{equation}
t_{v,\ell}-t_{w,\ell}\in (4+o(1))\kappa(2+\log D)\subset O(\kappa \log D).\label{eq:skew_faulty}
\end{equation}
Since faults that are not among the ancestry of a node cannot affect its pulse time, this follows by applying \Cref{thm:local_stab} with $\underline{\ell}=\ell-\lfloor(\log^3 D)\rfloor\le \ell - 2\Psi^s$ and $s':=\lfloor\log D\rfloor$.

To extend this to distance-$(\log^3 D)$ $k$-faulty nodes for $k\in \{1,2\}$, we show by induction on $k\in \{0,1,2\}$ that such nodes have their pulse time shifted by no more than $O(\kappa \log D)$ relative to an execution in which they are distance-$(\log^3 D)$ $0$-faulty.
The base case of $k=0$ is trivial.

To perform the step from $k-1\in \{0,1\}$ to $k$, assume towards a contradiction that there is a node $(v,\ell)$ with a larger shift, on some minimal layer.
Now consider a distance-$(\log^3 D)$ $k$-faulty node $(v,\ell)\in V_{\ell}\setminus F$, $\ell\le \bar{\ell}$, whose predecessors are all correct.
There must be a distance-$(\log^3 D)$ ancestor of $(v,\ell)$ that is faulty, since otherwise $(v,\ell)$ would be distance-$(\log^3 D)$ $0$-faulty.
Let $d$ be the minimal distance in which there is a faulty ancestor of $(v,\ell)$.
Then all ancestors of $(v,\ell)$ in distance $d$ are distance-$(\log^3 D)$ $k'$-faulty for $k'<k$, as otherwise $(v,\ell)$ would be $k'$-faulty for some $k'>k$.

Consider an ancestor of $(v,\ell)$ in distance $d-1$.
If its predecessors are all correct, by the induction hypothesis and \Cref{lem:drift_fault} their pulse time is shifted by $O(\kappa \log D)$ relative to an execution in which they are distance distance-$(\log^3 D)$ $0$-faulty.
If there is a faulty predecessor, we infer this from the induction hypothesis, \Cref{eq:skew_faulty}, and \Cref{lem:faulty_layer}.\footnote{Here the constants in the $O$-notation change, while \Cref{lem:drift_fault} maintains the bound used in its prerequisites. Since we perform only two inductive steps, we do not need to keep track of how much the constants increase.}
If $d>1$, we now inductively apply \Cref{lem:drift_fault} until having extended this bound to all ancestors of $(v,\ell)$ within distance $d-1$ and finally $(v,\ell)$ itself.
This is a contradiction to $(v,\ell)$ violating the claimed bound on the shift.

We conclude that indeed shifts are bounded by $O(\kappa \log D)$.
From this and \Cref{eq:skew_faulty}, it immediately follows that $\localskew_{\bar{\ell}}\in O(\kappa \log D)$.
As $C$ is sufficiently large, for sufficiently large $n$ this is a contradiction.
We conclude that $\localskew_{\ell}\in O(\kappa \log D)$ for all $\ell\in \N$, as claimed.
\end{proof}
Putting these results together, we arrive the desired bound on the local skew.

\thmlocalfaults*
\begin{proof}
By \Cref{cor:layer0}, with probability $1-o(1)$ it holds that $\localskew_0\le \kappa/2$.
By \Cref{obs:k_faulty}, with probability $1-o(1)$ each node is distance-$n^{1/12}$ $k$-faulty for $k\le 2$.
By a union bound, both events occur concurrently with probability $1-o(1)$.
Hence, the claim follows by applying \Cref{lem:local_faults}.
\end{proof}

%###
\subsection{Obtaining the Final Skew Bounds}\label{app:final}
%###
Recall that our model assumes that message delays and clock speeds do not vary. 
If the behavior of faulty nodes is static, i.e., the timing of their output pulse messages is identical in each pulse as well, a stable input frequency of $1/\Lambda$ results in repeating the exact same message pattern with the same timing every $1/\Lambda$ time.
We can exploit this to bound $\localskew_{\ell,\ell+1}$ in terms of $\localskew_{\ell}$.
\thmbound*
\begin{proof}
By \Cref{cor:fault}, for correct $(v,\ell+1)\in V_{\ell+1}$, $\ell\in \N$,
\begin{equation*}
\min_{\substack{((w,\ell),(v,\ell+1))\in E_{\ell}\\ (w,\ell)\notin F}}\{t_{w,\ell}^k\}+\Lambda-2\kappa \le t_{v,\ell+1}^k \le \max_{\substack{((w,\ell),(v,\ell))\in E_{\ell}\\ (w,\ell)\notin F}}\{t_{w,\ell}^k\}+\Lambda+2\kappa.
\end{equation*}
Because the behavior of fault nodes does not change between pulses, a simple induction shows that $t_{x,\ell'}^{k+1}=t_{x,\ell'}^k+\Lambda$ for all correct nodes $(x,\ell')\in V_{\ell'}$, $\ell'\in \N$.
In particular,
\begin{equation*}
\min_{\substack{((w,\ell),(v,\ell+1))\in E_{\ell}\\ (w,\ell-1)\notin F}}\{t_{w,\ell}^{k+1}\}-2\kappa \le t_{v,\ell+1}^k \le \max_{\substack{((w,\ell),(v,\ell+1))\in E_{\ell}\\ (w,\ell)\notin F}}\{t_{w,\ell}^{k+1}\}+2\kappa.
\end{equation*}

By \Cref{thm:local_faults}, $\localskew_{\ell}\in O(\kappa \log D)$.
Note that this bound applies uniformly over all executions.
Thus, even if $(v,\ell)$ is faulty, using that its neighbors are within distance $2$ of each other, it holds that
\begin{equation*}
\min_{\substack{((w,\ell),(v,\ell+1))\in E_{\ell}\\ (w,\ell-1)\notin F}}\{t_{w,\ell}^{k+1}\}-\max_{\substack{((w,\ell),(v,\ell+1))\in E_{\ell}\\ (w,\ell)\notin F}}\{t_{w,\ell}^{k+1}\}\in O(\kappa \log D),
\end{equation*}
by virtue of comparing to an execution in which $(v,\ell)$ is correct.
As $(v,\ell+1)$ was an arbitrary correct node, the claim of the theorem follows.
\end{proof}
It remains to argue that \emph{some} variation can be sustained.
\corbound*
\begin{proof}
The maximum length of a directed path in $H$ is bounded by $2\sqrt{n}$: at most $D\le \sqrt{n}$ hops in layer $0$, followed by at most $\sqrt{n}$ links from layer to layer.
Thus, accumulating all changes in timing due to link delay and clock speed variation along a path results in a deviation of $O((u+(\vartheta-1)(\Lambda-d))\log D=O(\kappa \log D)$.
This is trivial for layer~$0$ and applies to pulse propagation through the layers as well, because our respective analysis relies on \Cref{cor:fault} and \Cref{lem:drift_fault}.
In order to take into account a constant number of faulty nodes with arbitrary behavior, we reason analogously to the proof of \Cref{thm:fault_worst_case}, i.e., rely on \Cref{cor:fault} as well.
\end{proof}

\bibliographystyle{alphaurl}
\bibliography{references.bib}

\appendix

\section{Generating Synchronized Inputs}\label{app:layer0}

In this appendix we describe a method for generating well synchronised pulses at layer $0$, at a rate of roughly one pulse per $\Lambda$ time units.
There are several ways of approaching this task, but even when aiming for a fault-tolerant solution, this is an easy problem.
The reason is that we merely need to maintain a small local skew on a line topology, with no alternative propagation paths to neighboring nodes.

Since our goal is to handle an independent probability of $p\in o(n^{-1/2})$ of node failures, in fact we can simply exploit that at most $\sqrt{n}$ nodes are required on layer~$0$.
We provide a trivial scheme that is suitable for our specific setting of the base graph $G$ being a line (with replicated endpoints).
\begin{algorithm}[h]
  \caption{Pulse forwarding algorithm for nodes $(i,0)$, $i\in \{1,\ldots,D\}$; node $(0,0)$ is the clock source. The parameter $\Lambda$ is as described in \Cref{alg:Discretised_Gradient_TRIX}.}
  \label{alg:Level0_TRIX}
  \begin{algorithmic}
  \State $H:=\infty$
  \Loop
    \Repeat
    \If{received pulse from $(i-1,0)$}
      \State $H := H_{i,0}(t)$
    \EndIf
    \Until{$H_{i,0}(t)=H+\Lambda-d$}
    \State \textbf{broadcast pulse to $(i+1,0)$ and successors on layer $1$.}
  \EndLoop
  \end{algorithmic}
\end{algorithm}
\begin{lemma}\label{lem:layer0}
For $k\in \N$, assume that the clock source at node $(0,0)$ generates its $k$-th pulse at time $(k-1)\Lambda$.
If all nodes on layer~$0$ are correct, the scheme given in the above algorithm generates pulses with local skew $\localskew_0\le \kappa/2$ and $t_{i,0}^k\in [(k+i-1)\Lambda-i\kappa/2,(k+i-1)\Lambda]$. Moreover, it stabilizes after transient faults within time $D\Lambda$.
\end{lemma}
\begin{proof}
Consider first the case that there are no transient faults.
We prove the statement by induction on $i\in \N$, where the base case is covered by the assumptions on node $0$.

For the step from $i-1\in \N$ to $i$, we perform an induction over the pulse number $k\in \N_{>0}$.
The induction hypothesis is that pulses $1,\ldots,k-1$ have been generated in accordance with the claim of the lemma and the first $k-1$ loop iterations at node $i$ have been completed by the time the $k$-th pulse message from node $i-1$ arrives.
Note that we can use $k=0$ as base case for this induction, for which the claim is vacuously true.
For the step from $k-1\in \N$ to $k$, denote by $t'_{i-1,k}\in [t_{i-1,k}+d-u,t_{i-1,k}+d]$ the reception time of the pulse message from node $(0,i-1)$ at node $(0,i)$.
By the bounds on hardware clock rates, \Cref{eq:kappa}, and the induction hypothesis of the induction on $i$, node $(0,i)$ generates its $k$-th pulse at time
\begin{align*}
t_{i,k}&\in \left[t_{i-1,k}+d-u+\frac{\Lambda-d}{\vartheta},t_{i-1,k}+\Lambda\right]\\
&\subseteq \left[t_{i-1,k}+\Lambda-\frac{\kappa}{2},t_{i-1,k}+\Lambda\right]\\
&\subseteq \left[(k+i-1)\Lambda-\frac{i\kappa}{2},(k+i-1)\Lambda\right],
\end{align*}
unless it receives another pulse message from $(i-1,0)$ before doing so.
This, however, is not the case, since we assume that message delays and hardware clock rates do not vary over time, entailing that these reception times lie $\Lambda$ time apart.\footnote{Note that a separation of $\Lambda-d$ time would suffice. The slack of $d$ means that small changes in timing between pulses are unproblematic, which we exploit in \Cref{cor:bound}.}

It remains to show the claimed bound on stabilization time.
To this end, observe that the only state information that nodes maintain is $H$.
On reception of a pulse message, this state is overwritten.
This will remove spurious state from the system.

We would like to argue that the above induction can therefore be performed as-is, meaning that the system has stabilized by the time each node has generated its first pulse.
However, there is a subtlety: it could happen that a spurious message that is still in transit at time $0$ overwrites the state of node $(1,0)$ \emph{after} it received the first message from $(0,0)$.
Node $(1,0)$ then behaves as if the first message of $(0,0)$ arrived later, at the exact same time as the spurious message.
Because also such a spurious message is delivered within at most $d$ time, we can re-interpret this as a longer delay of still at most $d$ of the first message sent by node $(0,0)$.
Note that this modification reduces the difference between the reception times of the first and second pulse from node $(0,0)$ at node $(1,0)$ by up to $u$, but the separation remains at least $\Lambda-u\ge \Lambda-d$, i.e., the second message is not received before $(1,0)$ generates its first pulse.
We can apply the same scheme to nodes $2,\ldots,D$, resulting in the desired bound on the stabilization time.
\end{proof}

\begin{corollary}\label{cor:layer0}
$\localskew_0\le \kappa/2$ with probability $1-o(1)$.
It is self-stabilizing with stabilization time $\Lambda D$.
\end{corollary}

We remark that for a general base graph $G$, ensuring a small local skew is non-trivial.
However, so long as $|V|$ is small enough such that faults on layer $0$ occur with probability $o(1)$, one is free to fall back on a non-fault-tolerant GCS algorithm.
This achieves $\localskew_0\in O(\kappa \log D)$, which does not increase the asymptotic local skew bound of the pulse forwarding scheme.

\section{Full Pulse Forwarding Algorithm}\label{app:algo}

\begin{algorithm}[H]
  \caption{Discrete GCS at node $(v,\ell)$, $\ell>0$. The parameters $\Lambda$, and $\kappa$ will be determined later, based on the analysis.}
  %\Description{Discrete GCS at node $(v,\ell)$, $\ell>0$. The parameters $\Lambda$, and $\kappa$ will be determined later, based on the analysis.}
  \label{alg:Discretised_Gradient_TRIX}
  \begin{algorithmic}
    \Loop
      \State $H_{\min},H_{\own},H_{\max} := \infty$
      \For{$\{v,w\}\in E$}
        \State $r_w:=0$
      \EndFor
      \Repeat
        \If{received pulse from $v_{\ell-1}$ and $H_{\own}=\infty$}
          \State $H_{\own} := H_{v,\ell}(t)$
        \EndIf
        \If{for some $\{v,w\}\in E$ received pulse from $(w,\ell-1)$ and $r_w=0$}
          \If{$r_{w'}=0$ for all $\{v,w'\}\in E$}
            \State $H_{\min}:=H_{v,\ell}(t)$
          \EndIf
          \State $r_w:=1$
          \If{$r_{w'}=1$ for all $\{v,w'\}\in E$}
            \State $H_{\max}:=H_{v,\ell}(t)$
          \EndIf
        \EndIf
      \Until{$H_{\min}<\infty$ and $H_{v,\ell}(t)\ge \min\{H_{\max}+\kappa/2+\vartheta\kappa,2H_{\own}-H_{\min}+2\kappa)\}$}
      \If{$H_{v,\ell}(t)= H_{\max}+\kappa/2+\vartheta \kappa$}
        \State wait until $H_{v,\ell}(t)=H_{\max}+3\kappa/2+\Lambda-d$
      \Else
        \State $\Cor_{v,\ell}:=\min_{s\in \N}\{\max\{H_{\own}-H_{\max}+4s\kappa,H_{\own}-H_{\min}-4 s\kappa \}\}-\kappa/2$
        \If{$\Cor_{v,\ell} < 0$}
          \State $\Cor_{v,\ell} := \min\set{H_{\own}-H_{\min} + 3\kappa/2, 0}$
        \ElsIf{$\Cor_{v,\ell} > \vartheta\kappa$}
          \State $\Cor_{v,\ell} := \max\set{H_{\own}-H_{\max} - 3\kappa/2,\vartheta\kappa}$
        \EndIf
        \State wait until $H_{v,\ell}(t) = H_{\own} + \Lambda - d - \Cor_{v,\ell}$
      \EndIf
      \State \textbf{broadcast pulse}
    \EndLoop
  \end{algorithmic}
\end{algorithm}

A basic requirement for the algorithm to work correctly is that $(v,\ell)$ receives the $k$-th pulses of all correct predecessors within its $k$-th iteration of the main loop of \Cref{alg:Discretised_Gradient_TRIX}.
\begin{lemma}\label{lem:correct_pulsing}
For all $k\in \N$ and $(v,\ell)\in V_{\ell}$, $\ell>0$, node $(v,\ell)$ receives the $k$-th pulses of all correct predecessors within its $k$-th iteration of the main loop of \Cref{alg:Discretised_Gradient_TRIX}.
\end{lemma}
\begin{proof}
  We show by induction on $\ell\in \N_{>0}$ and $k\in \N_{>0}$ that $(v,\ell)$ broadcasts the $k^{th}$ pulse after receiving the $k$-th pulse from all correct $(w,\ell-1)$ satisfying that $((w,\ell-1),(v,\ell))\in E$, but before receiving the $(k+1)$-th pulse from such a node.
Moreover, for all $k\ge 2$, $t_{v,\ell}^k-t_{v,\ell}^{k-1}=\Lambda$.

For the induction on $\ell$, we use $\ell=0$ as base case, requiring only that nodes generate pulses at frequency $1/\Lambda$.
For the step from $\ell-1\in \N$ to $\ell$, we perform the induction on $k$.
Suppose that the claim holds for all $k'<k\in \N_{>0}$ and consider the $k$-th iteration of the outer loop at $(v,\ell)$.
\begin{itemize}
  \item The inner loop terminated because $H_{v,\ell}(t)=H_{\max} +\kappa/2+\vartheta \kappa$. Then a message from each node $(w,\ell-1)$, $\{v,w\}\in E$, has been received in the current loop iteration. By the induction hypotheses for layer $\ell-1$ and pulse $k-1$, respectively, for correct such nodes this is the $k$-th pulse message.
    
  We need to show that the $k$-th message from $(v,\ell-1)$ is received in time; the induction hypothesis guarantees that it is not received too early. As the minimum degree of $G$ is $2$, at least one node $(w,\ell-1)$, $\{v,w\}\in E$, is correct. If $(v,\ell-1)$ is correct, too, it sent its pulse message at the latest at time $t_{w,\ell-1}+\localskew_{\ell-1}$. By the bounds on message delay and clock speed, this message is received at a local time
  \begin{equation*}
    H\le H_{\max}+\vartheta(\localskew_{\ell-1}+u)\le H_{\max}+\Lambda-d<H_{v,\ell}(t_{v,\ell}^k).
  \end{equation*}
  \item The inner loop terminated because $H_{v,\ell}(t)=2H_{\own}-H_{\min}+2\kappa$. As $H_{\min}<\infty$, also $H_{\own}< \infty$. Using that $H_{\min}\le H_{\max}$, we get that
  \begin{align*}
  \Delta&:=\min_{s\in \N}\{\max\{H_{\own}-H_{\max}+4s\kappa,H_{\own}-H_{\min}-4 s\kappa \}\}- \frac{\kappa}{2}\\
  &\le \max\{H_{\own}-H_{\max},H_{\own}-H_{\min}\}- \frac{\kappa}{2}\\
  &\le H_{\own}-H_{\min}-\frac{\kappa}{2}
  \end{align*}
  and hence $\Cor_{v,\ell}\le H_{\own}-H_{\min}+3\kappa/2\le 3\kappa/2$. It follows that
  \begin{equation*}
  H_{v,\ell}(t_{v,\ell}^k)\ge \max\{H_{\min},H_{\own}\}+\Lambda-d-\frac{3\kappa}{2}.
  \end{equation*}
  We distinguish two subcases.
  \begin{itemize}
    \item $(v,\ell-1)$ is correct. Then by the bounds on message delay and clock speed, for each correct $(w,\ell-1)$, $\{v,w\}\in E$, its $k$-th pulse message is received at a local time
    \begin{equation*}
    H\le H_{\own}+\vartheta(\localskew_{\ell-1}+u)\le H_{\own}+\Lambda-d-\frac{3\kappa}{2}<H_{v,\ell}(t_{v,\ell}^k),
    \end{equation*}
    where the last step uses \Cref{eq:Lambda}.
    \item $(v,\ell-1)$ is faulty, implying that all $(w,\ell-1)$, $\{v,w\}\in E$, are correct. Then by the bounds on message delay and clock speed, for each correct $(w,\ell-1)$, $\{v,w\}\in E$, its $k$-th pulse message is received at a local time
    \begin{equation*}
      H\le H_{\min}+\Lambda-d-\frac{3\kappa}{2}<H_{v,\ell}(t_{v,\ell}^k),
    \end{equation*}
    where we use that in order to guarantee that $\Lambda-d \ge \vartheta(2\localskew_{\ell-1}+u)$ (i.e., \Cref{eq:Lambda}), this must also hold in an execution that differs by $(v,\ell-1)$ being correct; in such an execution, we have that
    \begin{align*}
    &\,\max_{\{v,w\}\in E}\{t_{w,\ell-1}\}-\min_{\{v,w\}\in E}\{t_{w,\ell-1}\}\\
    \le&\, \max_{\{v,w\}\in E}\{t_{w,\ell-1}\}-t_{v,\ell-1}+t_{v,\ell-1}-\min_{\{v,w\}\in E}\{t_{w,\ell-1}\}\\
    \le&\; 2\localskew_{\ell-1}.
    \end{align*}
  \end{itemize}
\end{itemize}
Next, we show that $(v,\ell)$ generates its pulse before receiving a $(k+1)$-th pulse message from a correct predecessor.
We distinguish two cases.
\begin{itemize}
  \item $(v,\ell-1)$ is not faulty. Then the earliest local time $H$ at which $(v,\ell)$ has received a $k$-th pulse from a correct predecessor is bounded from below by
  \begin{equation*}
  H\ge H_{\own}-\vartheta(\localskew_{\ell-1}+u).
  \end{equation*}
  As delays and clock speeds do not change, the induction hypothesis implies that the earliest message reception time for a $(k+1)$-th pulse from a correct predecessor is $\Lambda$ time later. Hence, it is sufficient to show that $H_{v,\ell}(t_{v,\ell}^k)\le H+\Lambda$. We distinguish three subcases.
  \begin{itemize}
    \item The inner loop terminated because $H_{v,\ell}(t)=H_{\max} +\kappa/2+\vartheta \kappa$ and at local time $H_{\min}$ a message from a correct predecessor $(w,\ell-1)$, $\{v,w\}\in E$, was received by $(v,\ell)$.
    Thus,
    \begin{equation*}
    H_{\own}+\vartheta(\localskew_{\ell-1}+u)+2\kappa \ge 2H_{\own}-H_{\min}+2\kappa\ge H_{\max}+\frac{\kappa}{2}+\vartheta\kappa.
    \end{equation*}
    and, by \Cref{eq:d},
    \begin{align*}
    H_{v,\ell}(t_{v,\ell}^k) &=H_{\max}+\frac{3\kappa}{2}+\Lambda-d\\
    &\le H_{\own}+\vartheta(\localskew_{\ell-1}+u)+2\kappa+\Lambda-d\\
    &\le H_{\own}-\vartheta(\localskew_{\ell-1}+u)\\
    &\le H+\Lambda.
    \end{align*}
    \item The inner loop terminated because $H_{v,\ell}(t)=H_{\max} +\kappa/2+\vartheta \kappa$ and at local time $H_{\max}$ a message from a correct predecessor $(w,\ell-1)$, $\{v,w\}\in E$, was received by $(v,\ell)$. Therefore,
    \begin{equation*}
    H_{\own}+\vartheta(\localskew_{\ell-1}+u) \ge H_{\max}
    \end{equation*}
    and, by \Cref{eq:d},
    \begin{align*}
    H_{v,\ell}(t_{v,\ell}^k) &=H_{\max}+\frac{3\kappa}{2}+\Lambda-d\\ 
    &\le H_{\own}+\vartheta(\localskew_{\ell-1}+u)+\frac{3\kappa}{2}+\Lambda-d\\
    &\le H_{\own}-\vartheta(\localskew_{\ell-1}+u)\\
    &\le H+\Lambda.
    \end{align*}
    \item The inner loop terminated because $H_{v,\ell}(t)=2H_{\own}-H_{\min}+2\kappa$ and $\Cor_{v,\ell}\ge 0$. By \Cref{eq:d}, then
    \begin{equation*}
    H_{v,\ell}(t_{v,\ell}^k)=H_{\own}+\Lambda-d-\Cor_{v,\ell}\le H_{\own}+\Lambda-d\le H+\Lambda.
    \end{equation*}
    \item The inner loop terminated because $H_{v,\ell}(t)=2H_{\own}-H_{\min}+2\kappa$ and $\Cor_{v,\ell}< 0$. Then
    \begin{equation*}
    \Cor_{v,\ell}= H_{\own}-H_{\min}+\frac{3\kappa}{2}
    \end{equation*}
    and
    \begin{equation*}
    H_{v,\ell}(t_{v,\ell}^k)=H_{\own}+\Lambda-d-\Cor_{v,\ell}=H_{\min}-\frac{3\kappa}{2}+\Lambda-d.
    \end{equation*}
    Since $H_{\min}$ is bounded from above by the earliest local reception time of a message from a correct node $(w,\ell-1)$, $\{v,w\}\in E$, we have that
    \begin{equation*}
    H_{\min}\le H_{\own}+\vartheta(\localskew_{\ell-1}+u).
    \end{equation*}
    By \Cref{eq:d}, we conclude that
    \begin{equation*}
    H_{v,\ell}(t_{v,\ell}^k)\le H_{\own}+\vartheta(\localskew_{\ell-1}+u)-\frac{3\kappa}{2}+\Lambda-d<H+\Lambda.
    \end{equation*}
  \end{itemize}
  \item $(v,\ell-1)$ is faulty. Then $H=H_{\min}$. Checking all cases in a similar fashion, we see that
  \begin{equation*}
  H_{v,\ell}(t_{v,\ell}^k)\le H_{\max}+\frac{3\kappa}{2}+\Lambda-d.
  \end{equation*}
  Using that \Cref{eq:d} must also apply in an execution where $(v,\ell-1)$ is not faulty and hence $\max_{\{v,w\}\in E}\{t_{w,\ell-1}\}-\min_{\{v,w\}\in E}\{t_{w,\ell-1}\}\le 2\localskew_{\ell-1}$, it follows that
  \begin{align*}
  H_{v,\ell}(t_{v,\ell}^k)&\le H_{\max}+\frac{3\kappa}{2}+\Lambda-d\\
  &\le H_{\min}+2\vartheta(\localskew_{\ell-1}+u)+\frac{3\kappa}{2}+\Lambda-d\\
  &\le H_{\min}+\Lambda\\
  &\le H+\Lambda.
  \end{align*}
\end{itemize}
Finally, we need to show that $t_{v,\ell}^{k+1}-t_{v,\ell}^k=\Lambda$.
This is now immediate from the induction hypothesis, the assumption that delays and hardware clock speeds do not change, and that faulty nodes send their messages with the same relative timing.
\end{proof}
We are now ready to show that \Cref{alg:Discretised_Gradient_TRIX} is equivalent to \Cref{alg:Simplified_Gradient_TRIX} in the absence of faults.

\begin{lemma}\label{lem:equivalence}
Suppose that for $(v,\ell)\in V_{\ell}$, $\ell>0$, and the predecessors of $(v,\ell)$ are correct.
Then running \Cref{alg:Simplified_Gradient_TRIX} instead of \Cref{alg:Discretised_Gradient_TRIX} results in the same pulse times of node $(v,\ell)$.
\end{lemma}
\begin{proof}
Assume towards a contradiction that the claim is false.
Denote by $t_{v,\ell}^k$ and $(t_{v,\ell}^k)'$ the pulse times of \Cref{alg:Simplified_Gradient_TRIX} and \Cref{alg:Discretised_Gradient_TRIX} in executions with identical delays, clock speeds, and behavior of faulty nodes.
W.l.o.g., let $t_{v,\ell}^k$ be minimal with the property that $t_{v,\ell}^k\neq (t_{v,\ell}^k)'$.

Consider the $k$-th loop iteration of \Cref{alg:Discretised_Gradient_TRIX} at node $(v,\ell)$.
We distinguish cases according to why the inner loop terminated.
\begin{itemize}
  \item The inner loop terminated because $H_{v,\ell}(t)=H_{\max} +\kappa/2+\vartheta \kappa$. Then in \Cref{alg:Simplified_Gradient_TRIX}, we have that
  \begin{equation*}
    H_{\own}\ge H_{\max}+\frac{\kappa}{2}+\vartheta \kappa,
  \end{equation*}
  implying that
  \begin{align*}
  \Delta&:=\min_{s\in \N}\{\max\{H_{\own}-H_{\max}+4s\kappa,H_{\own}-H_{\min}-4 s\kappa \}\}- \frac{\kappa}{2}\\
  &\ge \min_{s\in \N}\{H_{\own}-H_{\max}+4s\kappa\}- \frac{\kappa}{2}\\
  &\ge H_{\own}-H_{\min}-\frac{\kappa}{2}\\
  &\ge \vartheta \kappa.
  \end{align*}
  Hence, \Cref{alg:Simplified_Gradient_TRIX} computes
  \begin{equation*}
  \Cor_{v,\ell}=H_{\own}-H_{\max}-\frac{3\kappa}{2}
  \end{equation*}
  and generates its $k$-th pulse at local time
  \begin{equation*}
  H_{v,\ell}(t_{v,\ell}^k)=H_{\max}+\Lambda-d-\Cor_{v,\ell}=H_{\max}+\frac{3\kappa}{2}+\Lambda-d=H_{v,\ell}((t_{v,\ell}^k)'),
  \end{equation*}
  a contradiction.
  \item The inner loop terminated because $H_{v,\ell}(t)=2H_{\own}-H_{\min}+2\kappa$. As $H_{\min}<\infty$, also $H_{\own}< \infty$ for \Cref{alg:Discretised_Gradient_TRIX}. We distinguish two subcases.
  \begin{itemize}
    \item In \Cref{alg:Simplified_Gradient_TRIX}, we have
    \begin{equation*}
    \Delta:=\min_{s\in \N}\{\max\{H_{\own}-H_{\max}+4s\kappa,H_{\own}-H_{\min}-4 s\kappa \}\}- \frac{\kappa}{2}<0.
    \end{equation*}
    Then the same holds in \Cref{alg:Discretised_Gradient_TRIX}, as there $H_{\max}$ is either identical to that of \Cref{alg:Simplified_Gradient_TRIX} of $\infty$. Hence, both algorithms compute $C_{v,\ell}=\min\{H_{\own}-H_{\min}+3\kappa/2,0\}$ and subsequently $H_{v,\ell}(t_{v,\ell}^k)=H_{\own}+\Lambda-d-\Cor_{v,\ell}=H_{v,\ell}((t_{v,\ell}^k)')$, a contradiction.
    \item In \Cref{alg:Simplified_Gradient_TRIX}, we have
    \begin{equation*}
    \Delta:=\min_{s\in \N}\{\max\{H_{\own}-H_{\max}+4s\kappa,H_{\own}-H_{\min}-4 s\kappa \}\}- \frac{\kappa}{2}\ge 0
    \end{equation*}
    Let $s_{\min}\in \N$ be such that
    \begin{equation*}
    \Delta:=\max\{H_{\own}-H_{\max}+4s_{\min}\kappa,H_{\own}-H_{\min}-4 s_{\min}\kappa \}-\frac{\kappa}{2}.
    \end{equation*}
    If $\Delta=H_{\own}-H_{\min}-4 s_{\min}\kappa-\kappa/2$, the fact that $H_{\own}$ and $H_{\min}$ are identical in both algorithms, while $H_{\max}$ is either also identical or $-\infty$ in \Cref{alg:Discretised_Gradient_TRIX}, again leads to the contradiction $H_{v,\ell}(t_{v,\ell}^k)=H_{v,\ell}((t_{v,\ell}^k)')$. Hence, suppose that $\Delta=H_{\own}-H_{\max}+4s_{\min}\kappa-\kappa/2$ in \Cref{alg:Simplified_Gradient_TRIX}. Therefore,
    \begin{align*}
    0&\le \Delta \\
    &=H_{\own}-H_{\max}+4s_{\min}\kappa-\kappa/2\\
    &\le \max\{H_{\own}-H_{\max}+4(s_{\min}-1)\kappa,H_{\own}-H_{\min}-4(s_{\min}-1)\kappa\}-\frac{\kappa}{2}\\
    &=H_{\own}-H_{\min}-4(s_{\min}-1)\kappa-\frac{\kappa}{2}.
    \end{align*}
    Thus,
    \begin{equation*}
    2H_{\own}-H_{\min}+2\kappa\ge H_{\own}+4s_{\min}\kappa-\frac{3\kappa}{2}\ge H_{\max}-\kappa<H_{\max}-\frac{\kappa}{2}-\vartheta\kappa.
    \end{equation*}
    This is a contradiction, as then the inner loop in \Cref{alg:Discretised_Gradient_TRIX} would have terminated at an earlier time.\qedhere
  \end{itemize}
\end{itemize}
\end{proof}

\section{Self-Stabilization}\label{app:self-stab}

%###
\subsubsection*{Transient Faults}
%###
Thus far we have considered a fault model in which nodes which are faulty remain so permanently and distributed in accordance with the model described in \Cref{sec:model}. On VLSI chips this corresponds to fabrication errors or other failures which permanently affect a clock island. In addition to such permanent faults, it is of high interest to handle transient faults which potentially affect every node in the system for a short time. More specifically, after transient faults cease, such nodes should resume correct operation within a bounded, preferably small stabilization time.

This property is known as \emph{self-stabilization}~\cite{dijkstra74self}. On a SoC, such transient fault behavior could be the result of a Single Event Upset (SEU) caused by radiation or a droop in supply voltage that is too rapid for the control loop stabilizing the supply voltage to respond in time.

We emphasize that we require self-stabilization in the presence of permanent faults conforming to the model given in \Cref{sec:model}. This implies a powerful combination of resilience properties suitable for a wide range of real-world scenarios.

%###
\subsubsection*{How to Make the Pulse Propagation Algorithm Self-stabilizing}
%###
Since no restriction is imposed on the type or quantity of transient faults, they might result in an arbitrary state of the system's constituent components. Thus, proving self-stabilization is equivalent to showing that correct operation (re)commences from any possible initial state and feasible distribution and behavior of permanently faulty nodes. Note that corrupt link states are not an issue, since any spurious messages are delivered and processed within at most $d$ time.

Our task is simplified greatly by the fact that pulse propagation is directional. Recall that, on the top level, our analysis proceeds as follows:
\begin{enumerate}
  \item Show that each correct node receives the $k$-th pulses from its predecessors in its $k$-th loop iteration, which ends with sending its $k$-th pulse.
  \item Bound the skew of the $k$-th pulse within each layer.
  \item Argue that timing of consecutive pulses changes little enough such that for adjacent nodes good bounds are obtained on the time difference of pulse $k+1$ in layer $\ell$ and pulse $k$ on layer $\ell+1$, respectively.
\end{enumerate}
An arbitrary initial state disrupts the first step, in that the proper alignment of received and sent pulses might break down. For instance, a node on layer $3$ might incorrectly store that it recently received a pulse from a correct in-neighbor, implying that a faulty in-neighbor can control when it pulses by deciding when to send a pulse on its own. However, once we re-establish a consistent interpretation of what ``the $k$-th pulse'' is, where each correct node receives the $k$-th pulses from its correct in-neighbors in the corresponding loop iteration, the second and third step work without modification.

In summary, our task is to ensure that for a given input pulse, each correct node produces exactly one pulse in response, where each correct node receives all of these pulses from correct predecessors during the same loop iteration. More precisely, this reception must take place before $\Cor_{v,\ell}$ is computed and the final waiting statement before generating the pulse commences. For the pulse propagation through layer~$0$, cf.~\Cref{app:layer0}, this is trivial, since nodes merely forward pulses they receive. This provides the induction basis for showing this property for layers $\ell>0$, after a small modification of the forwarding algorithm that has no effect after stabilization.

\thmself*
\begin{proof}[Proof sketch.]
We break the reasoning down into a sequence of straightforward observations.
\begin{observation}\label{obs:pulses_one_way}
  Pulses that are sent by layer $\ell$ only affect layers $\ell+1$ and beyond. The behaviour of a node on layer $\ell$ depends only on its state and the reception times of incoming pulse messages from layer $\ell-1$.  
\end{observation}

Thus, if layers $0$ to $\ell$ are behaving correctly, then layers $\ell+1,\ell+2,\ldots$ cannot disrupt this in any way. Thus, all we need to show is that if layer $\ell$ behaves correctly, then after receiving a constant number of pulses from layer $\ell$, layer $\ell+1$ also behaves correctly. Then \Cref{alg:Discretised_Gradient_TRIX} stabilizes within $O(\sqrt{n})$ pulses after layer~$0$ functions correctly, which also takes $O(\sqrt{n})$ after the source operates correctly (again).

Based on~\Cref{obs:pulses_one_way}, we can focus on a single node on layer~$\ell\in \N$ and assume that the preceding layer is already operating correctly. Next, we note that ``getting one pulse right'' is good enough, as then the inductive behavior shown in \Cref{lem:correct_pulsing} kicks in.
\begin{observation}
If for a given pulse from layer $\ell$, a correct node on layer $\ell+1$ receives pulse messages from at least two correct predecessors in the same loop iteration of \Cref{alg:Discretised_Gradient_TRIX} prior to commencing the final wait statement, the node will do so for future pulses, too.
\end{observation}

Hence, making the algorithm self-stabilizing breaks down to ensuring this without disrupting its regular operation. To achieve this, we leverage another observation on the separation between reception times of consecutive pulses from correct nodes.
\begin{observation}
The reception times of correct nodes' pulse messages for pulse $k$ are separated by at most $\vartheta(2\localskew+u)$ local time. Consecutive pulses from a correct node are separated by at least $\Lambda-\localskew-u$ local time.\footnote{\Cref{cor:bound} is based on bounding the timing variations between consecutive pulses by $\localskew$, which also implies this lower bound.}
\end{observation}
%###
\paragraph{Algorithm Modification:}
%###
This hands us the key to a minimal change in the algorithm to achieve the desired self-stabilization property. The idea is that correct nodes are well-synchronized, so after receiving the first $k$-th pulse message from a correct node all others must follow within $\vartheta(2\localskew+u)$ local time. Therefore, if by that time both $H_{\own}$ and $H_{\max}$ remain unspecified, this proves that the first received pulse message was not the first pulse message from a correct node for a given pulse. In this case, it is valid to ``forget'' about this message. This carries the advantage that if the node does not receive enough messages in sufficiently short time to complete the loop iteration, it will successively delete all such messages and be ready to ``fully'' register correct nodes' messages for pulse $k+1$.

\begin{algorithm}[t!]
  \caption{This is a restatement of \Cref{alg:Discretised_Gradient_TRIX} with modifications to make it self-stabilizing. The changes are highlighted in blue text.}
  %\Description{This is a restatement of \Cref{alg:Discretised_Gradient_TRIX} with modifications to make it self-stabilising. The changes are highlighted in blue text.}
  \label{alg:Discretised_Gradient_TRIX_self_stab}
  \begin{algorithmic}
    \Loop
      \State $H_{\min},H_{\own},H_{\max}, H_{w} := \infty$ for all $w\in Neighbour(v)$
      \State $T := thread()$
      \State Initialise $T$ with $Wait()$
      \For{$\{v,w\}\in E$}
        \State $r_w:=0$
      \EndFor
      \Repeat
        \If{received pulse from $v_{\ell-1}$ and $H_{\own}=\infty$}
          \State $H_{\own} := H_{v,\ell}(t)$
        \EndIf
        \If{for some $\{v,w\}\in E$ received pulse from $(w,\ell-1)$ and $r_w=0$}
          \If{$r_{w'}=0$ for all $\{v,w'\}\in E$}
            \State $H_{\min}:=H_{v,\ell}(t)$
          \EndIf
          \State $r_w:=1$
          \State $H_w := H_{v,\ell}(t)$
          \If{$r_{w'}=1$ for all $\{v,w'\}\in E$}
            \State $H_{\max}:=H_{v,\ell}(t)$
          \EndIf
        \EndIf
      \Until{$H_{\min}<\infty$ and $H_{v,\ell}(t)\ge \min\{H_{\max}+\kappa/2+\vartheta\kappa,2H_{\own}-H_{\min}+2\kappa\}$}
      \If{$H_{v,\ell}(t)= H_{\max}+\kappa/2+\vartheta \kappa$}
        \State wait until $H_{v,\ell}(t)=H_{\max}+3\kappa/2+\Lambda-d$ or $H_{v,\ell}(t)<H_{\max}$
      \Else
        \State $\Cor_{v,\ell}:=\min_{s\in \N}\{\max\{H_{\own}-H_{\max}+4s\kappa,H_{\own}-H_{\min}-4 s\kappa \}\}-\kappa/2$
        \If{$\Cor_{v,\ell} < 0$}
          \State $\Cor_{v,\ell} := \min\set{H_{\own}-H_{\min} + 3\kappa/2, 0}$
        \ElsIf{$\Cor_{v,\ell} > \vartheta\kappa$}
          \State $\Cor_{v,\ell} := \max\set{H_{\own}-H_{\max} - 3\kappa/2,\vartheta\kappa}$
        \EndIf
        \State wait until $H_{v,\ell}(t) = H_{\own} + \Lambda - d - \Cor_{v,\ell}$ or $H_{v,\ell}(t)<H_{\own}$ or ($\Cor_{v,\ell}<0$ and $H_{v,\ell}(t)<H_{\min}$)
      \EndIf
      \State \textbf{broadcast pulse}
    \EndLoop
    \Function{Wait}{w}
      \State Wait until $H_{min}\neq \infty$.
      \If {($H_{min}\neq \infty$ and $H_{own} = \infty$) or ($H_{min}\neq \infty$ and $H_{max} = \infty$)}
        \State wait for $\vartheta(2\localskew + u)$
        \If{$H_{0} = \infty$ and $H_{max} = \infty$} 
          \State $H_{min} := \infty$
          \State $\forall (w,v) \in E$, $H_{w} = \infty$
          \State $\forall (w,v) \in E$, $r_{w} = 0$
        \Else
          \State return
        \EndIf
      \EndIf
      \State return 
    \EndFunction
  \end{algorithmic}
\end{algorithm}

\begin{observation}
The operation of \Cref{alg:Discretised_Gradient_TRIX} with the above modification is not affected for a correct node that stabilized, provided that the preceding layer already stabilized as well.
\end{observation}
To complete the proof sketch, we argue that this change of the algorithm, alongside simple checks to avoid getting stuck in waiting statements, is sufficient for a correct node to ``catch'' the $k$-th pulses of its correct predecessors if they stabilized a constant number of pulses ago. To see this, recall that two stored pulses are required to proceed to the waiting statements. Thus, at most $\vartheta(2\localskew+u)$ time after receiving the last $k$-th pulse from a correct predecessor, it becomes impossible to move on to a waiting statement until pulse $k+1$ from a correct predecessor is received.

If the node does not move on to the waiting statement, due to \Cref{eq:Lambda}, the node will have deleted the $k$-th pulses from correct predecessors by the time the first $(k+1)$-th pulse arrives. We then can reason analogously to \Cref{lem:correct_pulsing} to show that correct operation commences.

On the other hand, if the node moves on to a waiting statement, we want to make sure that the next loop iteration begins in time to guarantee that pulse $k+1$ from correct predecessors is not missed. There are two waiting statements that could be executed. The first waits until local time $H_{\max}+3\kappa/2+\Lambda-d$. At this point in the code, it should hold that $H_{\max}\le H_{v,\ell}(t)$; if this is not the case or the local time exceeds the local time until which the node should wait, the algorithm will end the loop iteration immediately (which is safe, because it can never happen after stabilization). \Cref{eq:d} then ensures that the loop iteration ends before pulse $k+1$ is received from a correct predecessor.

The second waiting statement waits until local time $H_{\own}+\Lambda-d-\Cor_{v,\ell}$. We can apply the same approach, but need to take into account that $\Cor_{v,\ell}$ could be negative. In order to bound it, we exploit $\Cor_{v,\ell}<0$ entails that $\Cor_{v,\ell}=H_{\own}-H_{\min}+3\kappa/2$ and hence $H_{\own}+\Lambda-d-\Cor_{v,\ell}< H_{\min}+\Lambda-d$. We conclude that executing the second waiting statement should result in waiting no longer than until local time $\max\{H_{\own},H_{\min}\}+\Lambda-d$, where both $H_{\own}$ and $H_{\min}$ are bounded from above by the current local time when the statement is reached. Hence, the algorithm will end the loop iteration immediately when $H_{v,\ell}(t)+\Lambda-d<H_{\own}+\Lambda-d-\Cor_{v,\ell}$ or $H_{v,\ell}(t)\ge H_{\own}+\Lambda-d-\Cor_{v,\ell}$.

To wrap up the argument, we conclude that in all cases, the waiting statements will complete $\Lambda-d+O(\localskew)$ time after a correct predecessor's pulse is received. Using \Cref{eq:d}, we can infer that on the next loop iteration, the node will receive the pulses of at least two correct predecessors before moving on to a waiting statement. By the above observations, it follows that the system stabilizes within $O(\sqrt{n})$ pulses.
\end{proof}

\section{Basic Statements}\label{app:basic}
We first show three basic lemmas.
The first relates the local reception times of pulses to the actual sending times, bounding the error by $\kappa$.
\begin{lemma}\label{lem:estimates}
For $(v,\ell)\in V_{\ell}$, where $\ell\in \N_{>0}$, set $t_{\min}:=\min_{\{v,w\}\in E}\{t_{w,\ell-1}\}$ and $t_{\max}:=\min_{\{v,w\}\in E}\{t_{w,\ell-1}\}$.
Then
\begin{align*}
t_{v,\ell-1}-t_{\max}-\kappa &\le H_{\own}-H_{\max}-\frac{\kappa}{2} \le t_{v,\ell-1}-t_{\max}\\
t_{v,\ell-1}-t_{\min}-\kappa &\le H_{\own}-H_{\min}-\frac{\kappa}{2} \le t_{v,\ell-1}-t_{\min}.
\end{align*}
\end{lemma}
\begin{proof}
We prove the first inequality; the second is shown analogously. Let $t'_{v,\ell-1}$ and $t'_{\max}$ denote the times when the pulse messages sent at time $t_{v,\ell-1}$ and $t_{\max}$ are received at $v_{\ell}$, respectively. From the bounds on message delays, it follows that
\begin{align*}
  t_{v,\ell-1} + d-u &\leq t'_{v,\ell-1} \leq t_{v,\ell-1} + d\mbox{ and}\\
  t_{\max} + d - u&\leq t'_{\max} \leq t_{\max} + d.
\end{align*}
Thus,
\begin{align*}
  t_{v,\ell-1} - t_{\max} - u\leq t'_{v,\ell-1} - t'_{\max} \leq t_{v,\ell-1} - t_{\max} + u.
\end{align*}
Using the bounds on hardware clock rates, we get that
\begin{equation*}
|t_{v,\ell-1}' - t_{\max}' -(H_{\own} - H_{\max})|\le (\vartheta-1)|t_{v,\ell-1}' - t_{\max}'|\le (\vartheta-1)(|t_{v,\ell-1} - t_{\max}|+u).
\end{equation*}
Applying \Cref{eq:Lambda}, we infer that
\begin{align*}
|t_{v,\ell-1} - t_{\max} -(H_{\own} - H_{\max})| &\le |t_{v,\ell-1}' - t_{\max}' -(H_{\own} - H_{\max})|+u\\
&\le (\vartheta-1)|t_{v,\ell-1} - t_{\max}| + \vartheta u\\
&\le (\vartheta-1)\localskew_{\ell-1} + \vartheta u\\
&\le (\vartheta-1)\left(\frac{\Lambda-d}{\vartheta}-u\right)+\vartheta u\\
&=\left(1-\frac{1}{\vartheta}\right)(\Lambda-d)+ u.
\end{align*}
Finally, using \Cref{eq:kappa}, we conclude that
\begin{align*}
t_{v,\ell-1}-t_{\max}-\kappa &\le t_{v,\ell-1}-t_{\max}-2\left(\left(1-\frac{1}{\vartheta}\right)(\Lambda-d)+ u\right)\\
&\le H_{\own}-H_{\max}-\frac{\kappa}{2}\\
& \le t_{v,\ell-1}-t_{\max}.\qedhere
\end{align*}
\end{proof}
The second lemma shows that corrections are not too large.
\begin{lemma}\label{lem:corrections}
For all $v\in V$ and $\ell\in \N_{>0}$, $\Cor_{v,\ell} \le \Lambda - d$.
\end{lemma}
\begin{proof}
Abbreviate
\begin{equation*}
\Delta=\min_{s\in \N}\left\{\max\{H_{\own}-H_{\max}+4s\kappa,H_{\own}-H_{\min}-4 s\kappa \}\right\}-\frac{\kappa}{2}.
\end{equation*}
We distinguish three cases.
\begin{itemize}
  \item $\Delta<0$. Then \Cref{alg:Simplified_Gradient_TRIX} sets
  \begin{equation*}
  \Cor_{v,\ell}\le \min\left\{H_{\own}-H_{\min}-\frac{\kappa}{2}+2\kappa,0\right\}\le 0.
  \end{equation*}
  As $\Lambda\ge d$ by \Cref{eq:Lambda}, the claim of the lemma holds in this case.
  \item $0\le \Delta\le \vartheta \kappa$. Then, using the notation of \Cref{lem:estimates},
  \begin{align*}
  \Cor_{v,\ell}=\Delta &<\min_{s\in \N}\left\{\max\{H_{\own}-H_{\max}+4s\kappa,H_{\own}-H_{\min}-4 s\kappa \}\right\}-\frac{\kappa}{2}\\
  &\le \min_{s\in \N}\left\{\max\{t_{v,\ell-1}-t_{\max}+4s\kappa,t_{v,\ell-1}-t_{\min}-4 s\kappa \}\right\}\\
  &\le \min_{s\in \N}\left\{\max\{\localskew_{\ell-1}+4s\kappa,\localskew_{\ell-1} -4s\kappa \}\right\}\\
  &=\localskew_{\ell-1},
  \end{align*}
  which is smaller than $\Lambda-d$ by \Cref{eq:Lambda}.
  \item $\Delta>\vartheta \kappa$. Note that then
  \begin{equation*}
  \vartheta\kappa<\Delta \le \max\{H_{\own}-H_{\max},H_{\own}-H_{\min}\}-\frac{\kappa}{2}
  =H_{\own}-H_{\max}-\frac{\kappa}{2},
  \end{equation*}
  as $H_{\max}\ge H_{\min}$. Therefore, applying \Cref{lem:corrections},
  \begin{align*}
  \Cor_{v,\ell}&=\max\left\{H_{\own}-H_{\max}-\frac{\kappa}{2}-\kappa,\vartheta\kappa\right\}\\
  &\le H_{\own}-H_{\max}-\frac{\kappa}{2}\\
  &\le t_{v,\ell-1}-t_{\max}\\
  &\le \localskew_{\ell-1}\\
  &<\Lambda-d.\qedhere
  \end{align*}
\end{itemize}
\end{proof}

The third lemma bounds the time difference between the pulses of $(v,\ell-1)$ and $(v,\ell)$.
\begin{lemma}\label{lem:drift}
For all $v\in V$ and $\ell\in \N_{>0}$ it holds that
\begin{align*}
d-u+\frac{\Lambda-d-\Cor_{v,\ell}}{\vartheta} \le t_{v,\ell}-t_{v,\ell-1} \le \Lambda-\Cor_{v,\ell}.
\end{align*}
\end{lemma}
\begin{proof}
Let $t'_{v,\ell-1}$ denote the time at which $(v,\ell)$ receives the pulse sent by $(v,\ell-1)$ at time $t_{v,\ell-1}$.
Inspecting the code of \Cref{alg:Simplified_Gradient_TRIX}, we see that
\begin{equation*}
H_{v,\ell}(t_{v,\ell})=H_{\own}+\Lambda -d - \Cor_{v,\ell}=H_{v,\ell}(t'_{v,\ell-1})+\Lambda -d - \Cor_{v,\ell}.
\end{equation*}
Since $\Cor_{v,\ell}\le \Lambda-d$ by \Cref{lem:corrections}, it follows that $H_{v,\ell}(t_{v,\ell-1})\ge H_{v,\ell}(t'_{v,\ell})$ and hence $t_{v,\ell}\ge t'_{v,\ell-1}$.
Using the bounds on message delays and hardware clock speeds, we get that
\begin{align*}
t_{v,\ell}-t_{v,\ell-1}&=t_{v,\ell}-t'_{v,\ell-1}+t'_{v,\ell-1}-t_{v,\ell-1}\\
&\le H_{v,\ell}(t_{v,\ell})-H_{v,\ell}(t'_{v,\ell-1})+d\\
&=\Lambda - \Cor_{v,\ell}
\end{align*}
and
\begin{align*}
t_{v,\ell}-t_{v,\ell-1}&=t_{v,\ell}-t'_{v,\ell-1}+t'_{v,\ell-1}-t_{v,\ell-1}\\
&\ge \frac{H_{v,\ell}(t_{v,\ell})-H_{v,\ell}(t'_{v,\ell-1})}{\vartheta}+d-u\\
&=\frac{\Lambda-d-\Cor_{v,\ell}}{\vartheta}+d-u,
\end{align*}
showing the claimed bounds.
\end{proof}

Next, we prove that \Cref{alg:Discretised_Gradient_TRIX} implements the slow, fast, and jump conditions.

\begin{lemma}\label{lem:slow_holds}
For all $s\in \N$ and $(v,\ell)\in V_{\ell}$, $\ell\in \N_{>0}$, $\SC(s)$ holds at $(v,\ell)$.
\end{lemma}
\begin{proof}
Using \Cref{lem:equivalence}, we prove the claim for \Cref{alg:Simplified_Gradient_TRIX}.
Set $t_{\min}:=\min_{\{v,w\}\in E}\{t_{w,\ell-1}\}$ and $t_{\max}:=\min_{\{v,w\}\in E}\{t_{w,\ell-1}\}$.
If $\Cor_{v,\ell}\le 0$, $\SCthree$ is trivially satisfied.
Hence, assume that $\Cor_{v,\ell}> 0$.
Abbreviate
\begin{align*}
\Delta=\min_{s\in \N}\left\{\max\{H_{\own}-H_{\max}+4s\kappa,H_{\own}-H_{\min}-4s\kappa \}\right\}-\frac{\kappa}{2}\\
=\max\{H_{\own}-H_{\max}+4s_{\min}\kappa,H_{\own}-H_{\min}-4 s_{\min}\kappa \}-\frac{\kappa}{2},
\end{align*}
where $s_{\min}\in \N$ is an index for which the minimum is attained.

If $\Delta\le \vartheta \kappa$, then $\Cor_{v,\ell}=\Delta$.
Otherwise,
\begin{equation*}
\Cor_{v,\ell}= \max\left\{H_{\own}-H_{\max}-\frac{\kappa}{2}-\kappa,\vartheta\kappa\right\}\le \max\{\Delta,\vartheta \kappa\}=\Delta.
\end{equation*}
Either way, we get that $\Cor_{v,\ell}/\vartheta<\Cor_{v,\ell}\le \Delta$.

We distinguish two cases.
\begin{itemize}
  \item $H_{\own}-H_{\max}+4s_{\min}\kappa\ge H_{\own}-H_{\min}-4 s_{\min}\kappa$. Then for $s\in \N$, $s\ge s_{\min}$, by \Cref{lem:estimates} we have that
  \begin{equation*}
  \Delta\le H_{\own}-H_{\max}+4s_{\min}\kappa-\frac{\kappa}{2}\le t_{v,\ell-1}-t_{\max}+4s \kappa,
  \end{equation*}
  i.e., $\SCone$ holds. Now consider $s\in \N$, $s<s_{\min}$. Since $H_{\own}-H_{\max}-\vartheta u+4s\kappa <H_{\own}-H_{\max}-\vartheta u+4s_{\min}\kappa\le \Delta$, but the minimum is attained at index $s_{\min}$, we must have that
  \begin{equation*}
  \Delta\le H_{\own}-H_{\min}-4s\kappa-\frac{\kappa}{2}
  \le t_{v,\ell-1}-t_{\min}-4s\kappa,
  \end{equation*}
  where the second step again applies \Cref{lem:estimates}.
  Thus, $\SCtwo$ holds.
  \item $H_{\own}-H_{\max}+4s_{\min}\kappa < H_{\own}-H_{\min}-4 s_{\min}\kappa$. In this case, we analogously infer that $\SCone$ holds for $s>s_{\min}$ and $\SCtwo$ holds for $s\le s_{\min}$.\qedhere
  \end{itemize}
\end{proof}

\begin{lemma}\label{lem:fast_holds}
For all $s\in \N$ and $(v,\ell)\in V_{\ell}$, $\ell\in \N_{>0}$, $\FC(s)$ holds at $(v,\ell)$.
\end{lemma}
\begin{proof}
Using \Cref{lem:equivalence}, we prove the claim for \Cref{alg:Simplified_Gradient_TRIX}.
Set $t_{\min}:=\min_{\{v,w\}\in E}\{t_{w,\ell-1}\}$ and $t_{\max}:=\min_{\{v,w\}\in E}\{t_{w,\ell-1}\}$.
If $\Cor_{v,\ell}\ge \vartheta\kappa$, trivially $\FCthree$ is satisfied.
Hence, assume that $\Cor_{v,\ell}<\vartheta\kappa$.

Abbreviate
\begin{align*}
\Delta=\min_{s\in \N}\left\{\max\{H_{\own}-H_{\max}+4s\kappa,H_{\own}-H_{\min}-4s\kappa \}\right\}-\frac{\kappa}{2}\\
=\max\{H_{\own}-H_{\max}+4s_{\min}\kappa,H_{\own}-H_{\min}-4 s_{\min}\kappa \}-\frac{\kappa}{2},
\end{align*}
where $s_{\min}\in \N$ is an index for which the minimum is attained.

If $\Delta\ge 0$, then $\Cor_{v,\ell}=\Delta$.
Otherwise,
\begin{equation*}
\Cor_{v,\ell}= \min\left\{H_{\own}-H_{\min}-\frac{\kappa}{2}+2\kappa,0\right\}\ge \Delta.
\end{equation*}
Either way, we get that $\Cor_{v,\ell}\ge \Delta$.

For $s\in \N$, $s\le s_{\min}$, by \Cref{lem:estimates} and \Cref{eq:kappa} it holds that
\begin{equation*}
\Delta\ge H_{\own}-H_{\max}+4s\kappa -\frac{\kappa}{2}\ge t_{v,\ell-1}-t_{\max}+(4s-2)\kappa+\kappa,
\end{equation*}
proving that $\FCone$ holds.
For $s\in \N$, $s> s_{\min}$, by \Cref{lem:estimates} and \Cref{eq:kappa} we get that
\begin{equation*}
\Delta\ge H_{\own}-H_{\min}-4(s-1)\kappa -\frac{\kappa}{2}\ge t_{v,\ell-1}-t_{\min}-(4s-2)\kappa+\kappa,
\end{equation*}
showing that $\FCtwo$ holds.
\end{proof}

\begin{lemma}\label{lem:jump_holds}
Suppose that layer $\ell-1\in \N$ and $v_{\ell}\in V_{\ell}$ are correct.
Then $\JC$ holds at $v_{\ell}$.
\end{lemma}
\begin{proof}
Using \Cref{lem:equivalence}, we prove the claim for \Cref{alg:Simplified_Gradient_TRIX}.
Set $t_{\min}:=\min_{\{v,w\}\in E}\{t_{w,\ell-1}\}$ and $t_{\max}:=\min_{\{v,w\}\in E}\{t_{w,\ell-1}\}$.
We distinguish three cases.
\begin{itemize}
  \item $0\le \Cor_{v,\ell}\le \vartheta\kappa$. Then $\JCthree$ is satisfied trivially.
  \item $\Cor_{v,\ell}<0$. By \Cref{lem:estimates} and \Cref{eq:kappa}, then
  \begin{equation*}
  \Cor_{v,\ell}=H_{\own}-H_{\min} -\frac{\kappa}{2}+2\kappa
  \ge t_{v,\ell-1}-t_{\min}+\kappa,
  \end{equation*}
  i.e., $\JCtwo$ holds.
  \item $\Cor_{v,\ell}>\vartheta \kappa$. By \Cref{lem:estimates}, then
  \begin{equation*}
  \Cor_{v,\ell}=H_{\own}-H_{\max} -\frac{\kappa}{2}-\kappa
  \le t_{v,\ell-1}-t_{\max}-\kappa,
  \end{equation*}
  i.e., $\JCthree$ holds.\qedhere
\end{itemize}
\end{proof}

\end{document}